\documentclass[11pt]{article}

\usepackage[margin=1in]{geometry}
\usepackage{microtype}
\usepackage{float}
\usepackage[section]{placeins} %
\usepackage{graphicx}
\usepackage{svg}
\usepackage{tikz}
\usepackage{algorithm}
\usepackage{algpseudocode}
\usepackage{tcolorbox}
\usepackage{titlesec}
\tcbuselibrary{breakable}
\usepackage{subcaption}
\usepackage{booktabs}
\usepackage{enumerate}
\usepackage{amsmath, amssymb, amsthm}
\usepackage{thmtools}
\usepackage{mathtools}
\usepackage[mathscr]{euscript}
\usepackage{xparse}
\usepackage{mathabx} %

\usepackage[numbers]{natbib}
\usepackage[colorlinks,citecolor=blue,pagebackref=true]{hyperref}
\usepackage[capitalise,nameinlink]{cleveref}
\usepackage{autonum}
\crefname{equation}{}{}

\newcommand{\st}{\text{ such that }}
\newcommand{\expp}[1]{\exp\paren{#1}}

\newcommand{\tsum}[1]{\sum_{i=1}^{t}{#1}}

\newif\ifverbose
\verbosefalse

\newcommand{\verbose}[1]{\ifverbose\textcolor{gray}{\textit{verbose $\rightarrow$}}\quad #1\textcolor{gray}{\textit{non-verbose $\rightarrow$}}\quad \fi}

\newcommand{\jupyter}[1]{\href[pdfnewwindow=true]{#1}{\smash{\begingroup
\setbox0=\hbox{\includegraphics[height=1.5em]{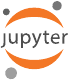}}%
\parbox{\wd0}{\box0}\endgroup}}}
\newcommand{\figcode}[2][py]{\jupyter{https://github.com/chasehmathis/asymptotic-confidence-horizons/blob/main/#2.#1}}

\let\oldstar\star
\renewcommand{\star}[1]{#1^{\oldstar}}

\newcommand{\R}{\mathbb{R}}
\newcommand{\N}{\mathbb{N}}
\newcommand{\Z}{\mathbb{Z}}
\newcommand{\Q}{\mathbb{Q}}

\def\ddefloop#1{
  \ifx
    \ddefloop#1
  \else\ddef{#1}
    \expandafter
    \ddefloop
  \fi}
\def\ddef#1{\expandafter\newcommand\csname c#1\endcsname{\ensuremath{\mathcal{#1}}}}

\ddefloop
ABCDEFGHIJKLMNOPQRSTUVWXYZ
\ddefloop

\renewcommand{\P}{\mathsf P}
\renewcommand{\Q}{\mathsf Q}

\newcommand{\eps}{\varepsilon}

\newcommand{\paren}[1]{\mathchoice{\left(#1\right)}{(#1)}{(#1)}{(#1)}}
\newcommand{\abs}[1]{%
\mathchoice{\left\lvert#1\right\rvert}%
           {\lvert#1\rvert}%
           {\lvert#1\rvert}%
       {\lvert#1\rvert}}
\let\Tilde\widetilde
\NewDocumentCommand{\pspace}{s O{\P}}{%
  \IfBooleanTF{#1}
    {\paren{\Tilde\Omega, \Tilde\cF, \Tilde{#2}}}
    {\paren{\Omega, \cF, #2}}%
}
\newcommand{\bracket}[1]{\left[ #1 \right]}
\newcommand{\curly}[1]{\mathchoice{\left\{#1\right\}}{\{#1\}}{\{#1\}}{\{#1\}}}

\newcommand{\norm}[2][]{\left\lVert #2 \right\rVert_{#1}}
\newcommand{\floor}[1]{\left\lfloor #1 \right\rfloor}

\newcommand{\n}{^{-1}}

\newcommand{\dd}{\mathrm{d}}
\newcommand{\dQ}{\dd \mathsf Q}
\newcommand{\dP}{\dd \P}
\newcommand{\dX}{\dd x}
\newcommand{\dZ}{\dd z}
\newcommand{\dT}{\dd t}
\newcommand{\dV}{\dd v}

\NewDocumentCommand{\EE}{m o}{%
  \IfValueTF{#2}{\mathbb{E}_{#2}\!\bracket{#1}}{\mathbb{E}\!\bracket{#1}}%
}
\NewDocumentCommand{\Var}{m o}{%
  \IfValueTF{#2}{\operatorname{Var}_{#2}\!\bracket{#1}}{\operatorname{Var}\!\bracket{#1}}%
}

\NewDocumentCommand{\p}{m o}{%
  \IfValueTF{#2}{p_{#2}\paren{#1}}{p\paren{#1}}%
}

\NewDocumentCommand{\PM}{m o}{%
  \IfValueTF{#2}{\mathsf{P}_{#2}}{\mathsf{P}}\mathchoice{\!}{}{}{}\paren{#1}%
}

\newcommand{\indic}[1]{1\{{#1}\}}
\newcommand{\simiid}{\mathrel{\smash{\overset{\text{i.i.d.}}{\sim}}\vphantom{\sim}}}
\newcommand{\iid}{{\text{i.i.d.}}}

\newcommand{\infseqt}[1]{\paren{#1_t}_{t=1}^\infty}
\newcommand{\infseqtb}[1]{\paren{#1}_{t=1}^\infty}

\tikzset{
  nv/.style ={circle, color=red, fill=red, inner sep=0.5mm},
  rv/.style ={circle, draw, thick, minimum size=6.5mm, inner sep=0.5mm},
  fv/.style ={rectangle, draw, thick, minimum size=6mm, inner sep=0.5mm},
  lv/.style ={circle, color=red, fill=gray!30, draw, thick, minimum size=6.5mm, inner sep=0.5mm},
  sv/.style ={circle, double, draw, thick, minimum size=6.5mm, inner sep=0.5mm},
  rve/.style={ellipse, draw, thick, minimum size=6.5mm, inner sep=0.5mm},
  rvs/.style ={circle, draw, thick, minimum size=5.5mm, inner sep=0.5mm},
  fvs/.style ={rectangle, draw, thick, minimum size=5mm, inner sep=0.5mm},
  lvs/.style ={circle, color=red, fill=gray!30, draw, thick, minimum size=5.5mm, inner sep=0.5mm},
  svs/.style ={circle, double, draw, thick, minimum size=5.5mm, inner sep=0.5mm},
  rves/.style={ellipse, draw, thick, minimum size=5.5mm, inner sep=0.5mm},
  deg/.style ={->, very thick, color=blue},
  sdeg/.style={*->, very thick, color=blue},
  degl/.style={->, very thick, color=red},
  beg/.style ={<->, very thick, color=red},
  cdeg/.style={{Circle[length=+2pt 2.5,width=+2pt 2.5, fill=none]}->, very thick, color=blue},
  cceg/.style={{Circle[length=+2pt 2.5,width=+2pt 2.5, fill=none]}-{Circle[length=+2pt 2.5,width=+2pt 2.5, fill=none]}, very thick},
  uceg/.style={{Circle[length=+2pt 2.5,width=+2pt 2.5, fill=none]}-, very thick},
  ueg/.style ={very thick}
}

\theoremstyle{plain}
\newtheorem{theorem}{Theorem}[section]
\newtheorem{lemma}[theorem]{Lemma}
\newtheorem{corollary}[theorem]{Corollary}
\newtheorem{proposition}[theorem]{Proposition}
\newtheorem{definition}{Definition}

\newtheorem{condition}{Condition}

\theoremstyle{definition}

\theoremstyle{remark}
\newtheorem{remark}[theorem]{Remark}

\usepackage[normalem]{ulem}
\newcommand{\nto}{{n \to \infty}}

\newcommand{\Deltabar}{{\overline \Delta}}

\title{Confidence Horizons}
\author{Chase Mathis \& Ian Waudby-Smith\vspace{0.2cm}\\{\small Department of Statistics} \\{\small University of California, Berkeley}\vspace{0.2cm}\\
{\small \texttt{\{cmathis,ianws\}@berkeley.edu}}}
\date{\today}

\begin{document}
\maketitle
\begin{abstract}
    Anytime-valid inference enables analysts to continuously monitor their data and stop experiments early. However, the majority of these methods incur a certain conservativeness by remaining valid on infinite time horizons. In practice, a bound on the horizon may be imposed due to budgetary, practical, or ethical constraints. In this paper, we ask the question: \emph{``Is it possible to obtain sharper large-sample anytime-valid inference by forgoing validity beyond some finite time horizon?''}. We provide a positive answer to this question by proposing a family of statistical objects that we call ``confidence horizons''. These objects can be viewed as large-sample confidence sequences on bounded time horizons, or alternatively as group sequential repeated confidence intervals with a maximal number of interim peeking times. We make explicit connections to the group sequential boundaries of Pocock [1977], O'Brien--Fleming [1979], and Wang--Tsiatis [1987].
We derive closed-form distribution functions of certain statistics which can be used to calculate the asymptotic quantiles of confidence horizons exactly, sidestepping the repeated integration typically employed in group sequential methods. 
We illustrate the use of confidence horizons for treatment effect estimation in sequentially randomized experiments under adaptive Neyman allocation.

\end{abstract}
\section{Introduction}
While classical frequentist statistical procedures are only valid at fixed, pre-specified sample sizes, sequential inference provides one way to halt data collection adaptively and make decisions early. One fundamental object in sequential inference is the confidence sequence (CS). Confidence sequences are sequentially valid analogues of confidence intervals (CIs), enjoying frequentist guarantees under repeated ``peeking''. Despite their desirable guarantees, confidence sequences can be conservative in certain real-world settings. This conservativeness tends to be the result of two key factors.

The first is that the vast majority of confidence sequences in the literature are nonasymptotic, meaning that their type-I error guarantees hold in finite samples. Methods with nonasymptotic validity tend to be conservative when compared to their asymptotic counterparts. Analogously, confidence intervals based on the central limit theorem (CLT) are much tighter than those based on a Chernoff bound, and the same intuition carries over to the ``anytime-valid'' sequential setting (see, e.g., \citep[Figure 7]{waudby-smith_time-uniform_2024}). In the same way that it is possible to use a CLT-based confidence interval in place of a nonasymptotic one, there exists work on deriving asymptotic analogues of confidence sequences which trade off finite-sample validity for power and the ability to make weaker moment assumptions \citep{robbins_boundary_1970,waudby-smith_time-uniform_2024,bibaut2022near,waudby-smith_distribution-uniform_2026}. Making this tradeoff does effectively remove the first source of conservativeness.

The second factor is that confidence sequences enable considerable flexibility by permitting optional stopping at arbitrarily large sample sizes. This comes at a cost of ``power'' (broadly defined) at those sample sizes that are typically encountered in practice. For example, a pharmaceutical company is unlikely---and in many cases, unable---to enroll hundreds of thousands or even tens of thousands of patients. Large technology companies cannot perform randomized experiments on more units than the number of users they have, which although may be large, is ultimately finite. While practitioners in both settings may benefit from the ability to stop adaptively on their respective time horizons, they both have practical limits on how long those horizons can be. In this work, we focus on eliminating this second source of conservativeness by deriving methods resembling asymptotic confidence sequences but that are only valid on bounded time horizons. We call these methods ``asymptotic confidence horizons'' (AsympCHs). Empirical differences between AsympCHs and AsympCSs as well as CLT-based CIs are illustrated in \cref{fig:intro-fig}. The following informal theorem summarizes the essence of a confidence horizon, leaving more detailed statements to \cref{def:asympch}, \cref{prop:asympch-q-1/2}, and \cref{theorem:general-mean}.

\begin{tcolorbox}[breakable]
\begin{theorem}[Confidence horizons; \emph{brief \& informal}]\label{theorem:brief-informal}
Suppose $\infseqt{X}$ are independent (or martingale-dependent) random variables from a distribution $\P$ with means $\mu_\P$. Then for a start time $m \in \N$ and a multiplicative horizon $\Delta \geq 1$, we say that $\paren{\bar C_{t}^{\smash{(\Delta)}}}_{t \in [m,\Delta m]}$ forms a $(1-\alpha)$-confidence horizon for $\mu_\P$ if 
\begin{equation}\label{eq:asympch guarantee brief and informal}
\liminf_{m \to \infty} \PM{\forall t \in [m, \Delta m]:  \mu_\P \in \bar C_{t}^{(\Delta)}} \geq 1-\alpha.
\end{equation}
In particular, we derive a distribution function $\Psi(1-\alpha; \Delta)$ for which \eqref{eq:asympch guarantee brief and informal} holds with
\[
\widebar C_t^{(\Delta)} = \frac{1}{t}\sum_{i=1}^t X_i \pm \widehat \sigma_t \frac{\Psi \n(1-\alpha; \Delta)}{\sqrt{t}}.
\]
\end{theorem}
\end{tcolorbox}

\begin{figure}
    \centering
    \includegraphics[width=0.8\linewidth]{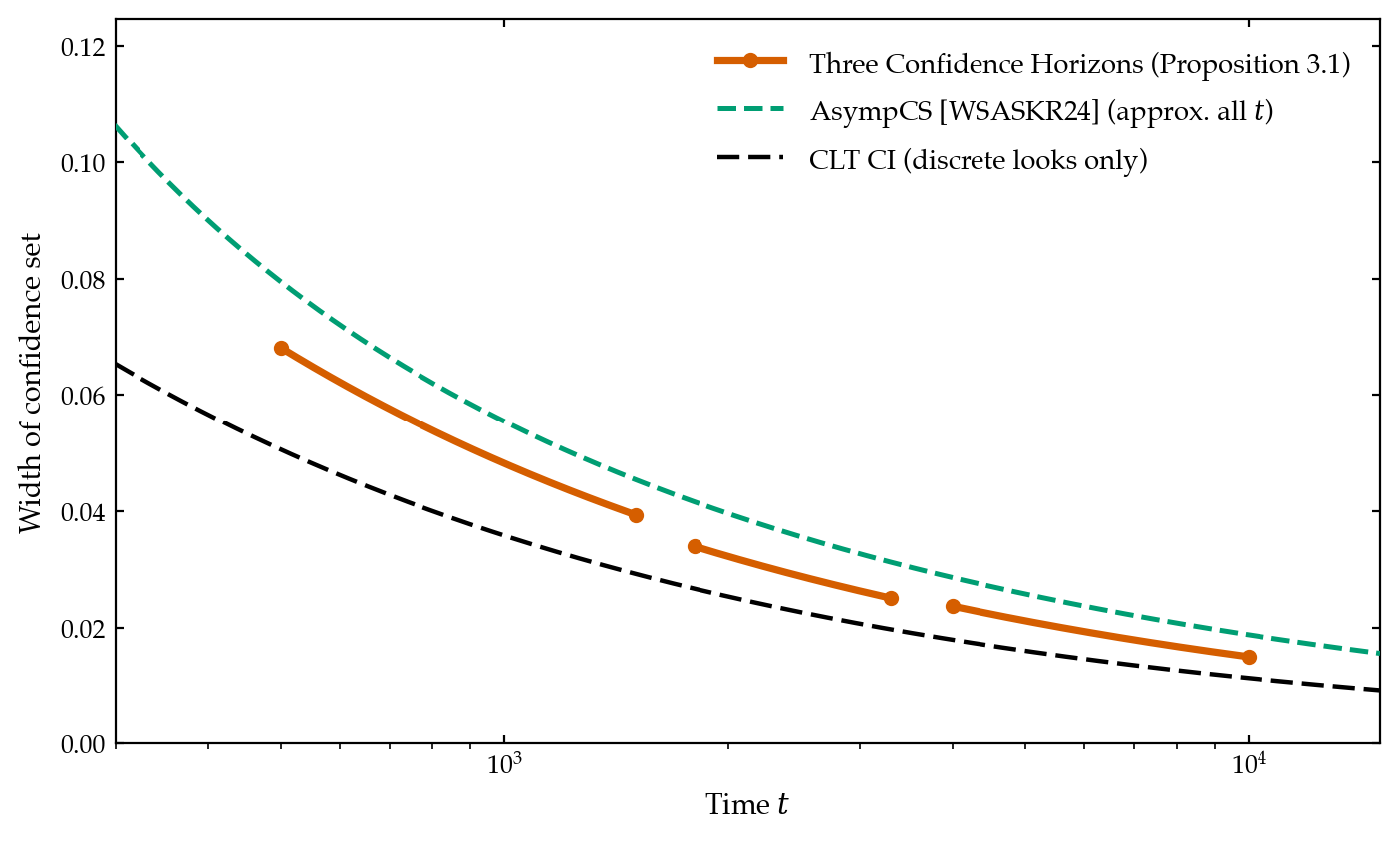}
    \caption{A comparison of widths of CLT-based confidence intervals, asymptotic confidence sequences (AsympCSs) \citep{waudby-smith_time-uniform_2024,bibaut2022near,waudby-smith_distribution-uniform_2026}, and three particular asymptotic confidence horizons (AsympCHs). The dashed green line uses the AsympCS of \citet[Theorem 2.8]{waudby-smith_time-uniform_2024}, the dotted black line uses the usual CLT-based CI, and the three solid orange lines are three separate AsympCHs with different start times and multiplicative horizons.}
    \label{fig:intro-fig}
\end{figure}

To make matters more formal, suppose that data $X_1, X_2, \dots$ are observed \iid{} in an online stream from a distribution $\P$ and it is of interest to estimate the mean $\mu = \EE{X_1}[\P]$.\footnote{We consider non-independent and non-identically distributed data in \cref{sec:martingales}.} A classical (fixed-$n$) 0.95-confidence interval $C_n$ for $\mu$ is an interval satisfying
\begin{equation}\label{eq:intro-ci}
    \lim_\nto \P \left ( \mu \in C_n \right ) = 0.95,
\end{equation}
while an asymptotically valid confidence sequence is a sequence $\paren{C_k^{(m)}}_{k = m}^\infty$ satisfying 
\begin{equation}\label{eq:intro-cs}
    \lim_{m \to \infty} \P \left ( \forall k \in \{m, m+1, \dots \},\ \mu \in C_k^{(m)} \right ) = 0.95.
\end{equation}
See \citep{waudby-smith_time-uniform_2024,bibaut2022near,waudby-smith_distribution-uniform_2026} for example constructions satisfying \eqref{eq:intro-cs}.
Clearly, \eqref{eq:intro-cs} provides more flexibility to the analyst than \eqref{eq:intro-ci} in the sense that inference can be carried out at sample sizes of $m, m+1, m+2$, and so on, without penalties for multiple data-dependent peeks. However, \eqref{eq:intro-ci} ``invests'' all of its error budget on a single sample size of $n$, whereas \eqref{eq:intro-cs} spreads that budget over infinitely many and arbitrarily large sample sizes. In this work, we develop \emph{confidence horizons} that aim to occupy a middle ground between \eqref{eq:intro-ci} and \eqref{eq:intro-cs}. While these are defined formally in \cref{def:asympch}, the essential property they must satisfy (in juxtaposition with \eqref{eq:intro-ci} and \eqref{eq:intro-cs}) is that for a fixed multiplicative horizon window size $\Delta \geq 1$,
\begin{equation}\label{eq:intro-ch}
    \lim_{m \to\infty} \P \left ( \forall k \in \{m, \dots, \Delta m \},\ \mu \in C_k^{(m)} \right ) = 0.95.
\end{equation}
Taking $\Delta = 1$ and $\Delta = \infty$ recover the guarantees in \eqref{eq:intro-ci} and \eqref{eq:intro-cs}, respectively. We note that if the horizon $\{m ,\dots, \Delta m\}$ were replaced by a subset with a small and evenly spaced number of peeking times (e.g., $\sim 10$), then \eqref{eq:intro-ch} would be reminiscent of group sequential repeated confidence intervals \citep{jennison_group_2000}.
We make explicit connections with (and some theoretical upgrades to) group sequential methods in \cref{sec:theoretical-conseq-gs}.
 See \cref{fig:methods-validity-windows} for an illustration of the difference in guarantees between AsympCHs and other sequential inference methods found in the literature.
\begin{figure}[!htbp]
    \centering
    \includegraphics[width=0.9\linewidth]{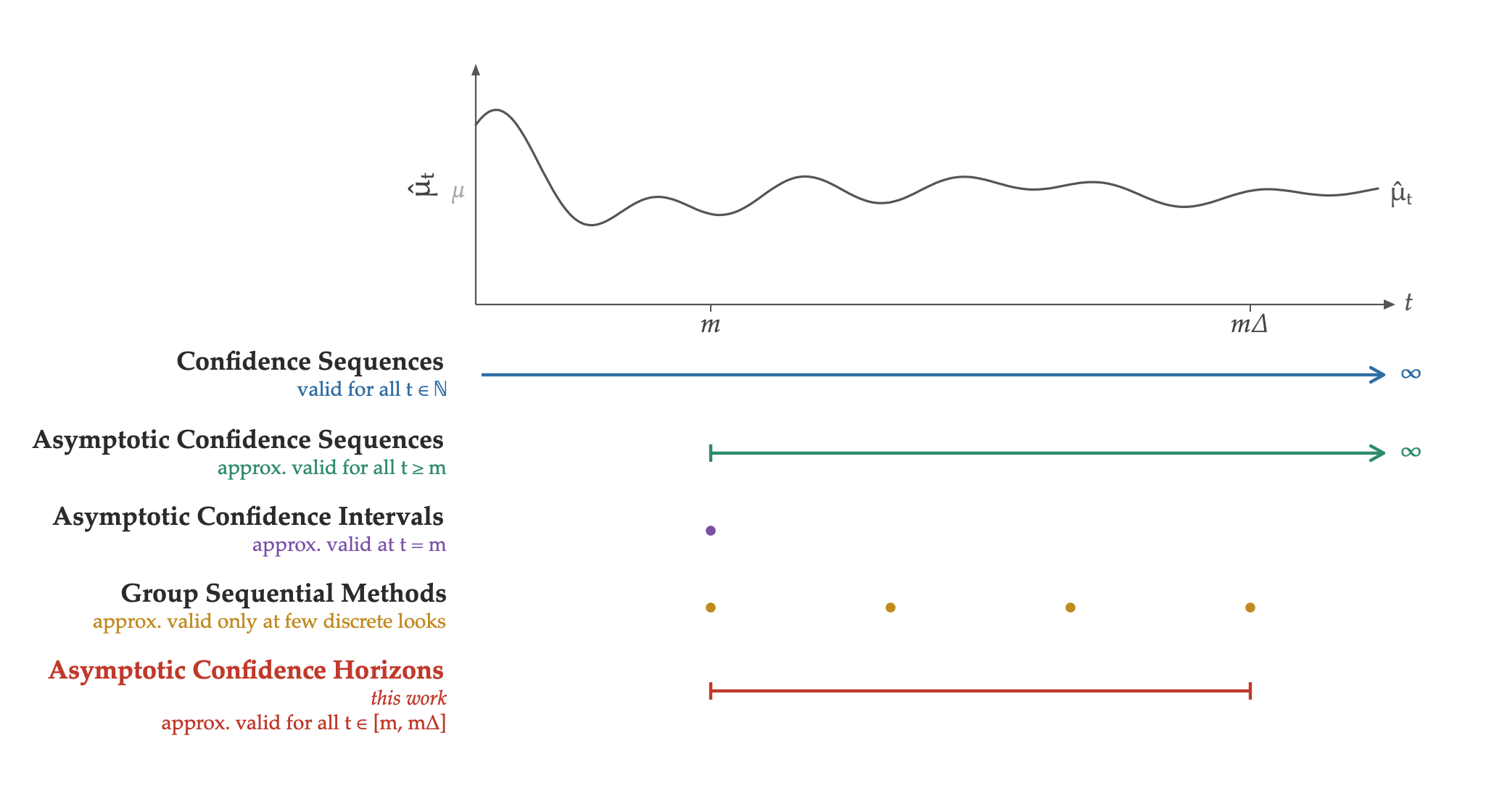}
    \caption{An illustration of how confidence horizons compare to other objects in the literature on sequential inference. The lines (or dots in the case of group sequential methods) signify where ``peeking'' at data is allowed (formally the stopping times at which inference remains valid). Confidence sequences are valid for all $t \in \N$, asymptotic ones are approximately valid at all $t = m, m+1, \dots$ for large $m$; group sequential methods are approximately valid at a (small) finite set of prespecified times $t_1, t_2, \dots, t_K$ for large $t_1$; asymptotic confidence horizons are approximately valid for all $t \in \N \cap [m, \Delta m]$ and for large $m$. 
    } 
    \label{fig:methods-validity-windows}
\end{figure}

\paragraph{Outline.}
The rest of the paper proceeds as follows. In \cref{sec:preliminaries}, we provide background on nonasymptotic confidence sequences and their asymptotic counterparts. We provide a definition for asymptotic confidence horizons in \cref{sec:power-q} and explicitly derive such horizons for means with \iid{} data in \cref{theorem:general-mean} and with martingale-dependent data in \cref{theorem:mgale-asympcv}. In \cref{sec:neyman-allocation}, we illustrate the use of asymptotic confidence horizons for estimating treatment effects in adaptive sequential experiments. In \cref{sec:theoretical-conseq-gs}, we draw connections between asymptotic confidence horizons and repeated confidence intervals in group sequential designs. As a byproduct of our analyses, we show that group sequential designs are valid under martingale dependence and in various uniform senses. In \cref{sec:simulation-study}, we conduct simulations studying coverage and widths of confidence horizons. Each figure has a clickable Jupyter logo $\smash{\begingroup
\setbox0=\hbox{\includegraphics[height=1.5em]{figures/jupyter-logo.pdf}}%
\parbox{\wd0}{\box0}\endgroup}$, which links to the code that generates that figure.\footnote{\href{https://github.com/chasehmathis/asymptotic-confidence-horizons}{https://github.com/chasehmathis/asymptotic-confidence-horizons}}

\paragraph{Notation.}
For a probability distribution $\P$, we let $\pspace$ be a filtered probability space. We use the notation $\pspace[\cP]$ to represent a collection of probability spaces $\pspace_{\P \in \cP}$. 
For a random variable $X$, we write its expectation as $\EE{X}[\P] = \int x \,\dP$. For a sequence of random variables $\infseqt{X}$ we write $\infseqt{X} \sim \cP$ to mean that $\infseqt{X}$ are defined on $(\Omega, \cF)$ with potentially different laws for each $\P \in \cP$.
For a sequence of random variables $\infseqt{X}$, we define the conditional means and variances by
\[
\mu_{\P, t} = \EE{X_t \mid X_{1}, \dots, X_{t-1}}[\P]\quad\text{and} \quad \sigma^2_{\P, t} = \EE{X_t^2 \mid X_{1}, \dots, X_{t-1}}[\P] - \mu_{\P, t}^2,
\]
respectively.
If $\infseqt{X}$ are independent, the above reduce to $\mu_{\P,t} = \EE{X_t}[\P]$ and $\sigma^2_{\P,t} = \EE{X_t^2}[\P] - \mu_{\P,t}^2$.
Define the
cumulative mean and variance parameters:
\[
\Tilde \mu_{\P, t} = t \n \tsum{\mu_{\P, i}}\quad\text{and} \quad V_{\P, t} = \tsum{\sigma^2_{\P, i}}, \quad \Tilde \sigma^2_{\P, t} = t \n V_{\P, t}.
\]

\paragraph{Brief Literature Review.} Various works have studied problems at the intersection of sequential inference and bounded time horizons. For instance, \citet{waudby2024estimating}, \citet{voracek2026star},
\citet{taga_learning_2026}, and \citet{clerico_time-sensitive_2026} consider nonasymptotic anytime-valid inference but tune certain parameters so that confidence intervals and tests are most effective at a fixed time horizon. Taking a different approach, \citet{koning_anytime_2026} focuses on the ability to construct nonasymptotic anytime-valid tests $\phi_1, \dots, \phi_N$ such that $\phi_N$ is the same test as some classical (non-sequential) test at time $N$. 
Since we operate in the asymptotic regime, the two most closely related works to confidence horizons are (1) asymptotic confidence sequences \citep{waudby-smith_time-uniform_2024,bibaut2022near,waudby-smith_distribution-uniform_2026} and (2) group sequential designs \citep{obrien1979multiple,pocock_group_1977,jennison_interim_1989,jennison_distribution_1997}. Given their close relationship, we briefly review the former in \cref{sec:preliminaries} and provide an extensive review and comparison with the latter in \cref{sec:connection-gs}. In particular, from a practical perspective, \cref{fig:gs-comparison} illustrates how confidence horizons can be viewed as ``continuous'' analogues of well-known group-sequential procedures. From a theoretical perspective, the proofs take an alternative approach that may be of independent interest. We also mention the work of \citet{chugg2026post} and \citet{massiani2026asymptotic} who study asymptotic analogues of so-called \emph{e-processes}, which are otherwise central objects in nonasymptotic sequential inference as they satisfy Ville's maximal inequality on unbounded time horizons \citep{ville1939etude}. In particular, \citet{massiani2026asymptotic} study notions of asymptotic e-processes for bounded time horizons. See also \citet{shafer2011test,grunwald2024safe,howard_time-uniform_2021,ramdas2022testing,ramdas2020admissible,ruf2023composite} for some seminal works on these topics in the nonasymptotic regime and \citet[\S 7]{ramdas_hypothesis_2025} for a textbook chapter dedicated to them. However, we neither make use of e-processes nor Ville's inequality in this work.
In a separate but related line of work, \citet{anderson_modification_1960} aims to decrease the expected stopping time of Wald's sequential probability ratio test when data are generated by a distribution lying strictly between the null and alternative. He uses a finite-horizon construction to fix this problem, deriving closed form quantile functions in the Normal case using a Wiener approximation.

\section{Preliminaries}\label{sec:preliminaries}
In this section, we review relevant facts about nonasymptotic confidence sequences, their asymptotic counterparts, and various notions of uniformity that are important to consider in the asymptotic regime.

\subsection{Confidence Sequences}\label{sec:preliminaries-conf-seq}
Let $\infseqt{X} = (X_1, X_2, \ldots)$ be an \iid{} sequence of random variables drawn from a distribution $\P$ equipped with a parameter of interest, $\theta_{\P}$ (e.g.~the mean $\mu_{\P}$). There exist many methods to construct a nonasymptotic confidence interval $\dot C_t = \dot C(X_1, \ldots, X_t)$ under different assumptions on $\P$, such as boundedness or sub-Gaussian tails \citep{hoeffding_probability_1963,mcdiarmid_method_1989}. When using one of these methods, $\dot C_t$ has the property that
\begin{equation}\label{eq:conf-int}
\forall t \in \N, \PM{\theta_{\P} \in \dot C_t} \geq 1-\alpha.
\end{equation} 
The probability statement in \cref{eq:conf-int} is valid only for a \emph{prespecified} number $t \in \N$. 
Said differently, the analyst cannot continuously compute $\dot C_1, \dot C_2, \ldots, \dot C_\tau$, and draw inferences from $C_\tau$ at a data dependent stopping time $\tau$.

Confidence sequences (CS) are sequences of confidence intervals $\infseqt{\bar C}$ so the analyst \emph{can} continuously calculate $\bar C_1, \dots, \bar C_\tau$, stopping at a data-dependent stopping time $\tau$ without sacrificing control of type-I error \citep{darling1967confidence,robbins_statistical_1970,lai1976confidence,howard_time-uniform_2021,waudby2024estimating,cortinovis2025confidence}.\footnote{We use the bar notation, $\bar C_t$, in contrast to dot notation, $\dot C_t$, to indicate that we are interested in cumulative coverage over $i = 1, 2, \ldots, t$ instead of single-point coverage just at time $t$.} To be precise, a CS for a parameter $\theta_{\P}$ is a sequence of confidence intervals $\infseqt{\bar C}$ 
that are functions of the data $\infseqt{X} \sim \P$ with the property that
\begin{equation}\label{eq:conf-seq}
   \PM{\forall t \in \N, \theta_{\P} \in \bar C_t} \geq 1-\alpha, ~~~\text{or equivalently,}~~~\PM{\theta_{\P} \in \bar C_\tau} \geq 1-\alpha
\end{equation}for all stopping times $\tau$ \citep[Lemma 3]{howard_time-uniform_2021}. 
The key difference between \cref{eq:conf-int} and \cref{eq:conf-seq} is that the universal quantifier ``$\forall t$'' lies \emph{inside} the probability statement for confidence sequences. 
This distinction in requirements is often called ``time-uniformity'' or ``anytime validity''.\footnote{In the time-uniform inference literature, the terms ``time'' and ``sample size'' are often used interchangeably and we do the same throughout this work.} 

Notice that the inequalities in \eqref{eq:conf-int} and \cref{eq:conf-seq} are valid in finite samples; in other words, they are nonasymptotically valid. While desirable in certain problems, nonasymptotic methods tend to be conservative relative to their asymptotic counterparts and necessarily require stronger moment or tail assumptions; see \citet{bahadur_nonexistence_1956} for a formal description of the latter claim. For this reason, when faced with fixed sample sizes in non-sequential settings, confidence intervals used in practice are typically derived via the central limit theorem and only enjoy asymptotic validity.

In the same way that \eqref{eq:conf-seq} is a time-uniform analogue of \eqref{eq:conf-int}, there exist time-uniform analogues of CLT-based confidence intervals. Such objects tend to be referred to as ``asymptotic confidence sequences''; see \citet{robbins_boundary_1970,waudby-smith_time-uniform_2024,waudby-smith_distribution-uniform_2026}, and \citet{bibaut2022near}. We review these objects here.
\begin{definition}[Asymptotic Confidence Sequences]\label{def:AsympCS}
Let $\infseqt{X} \sim \P$ and $m \in \N$. The sequence of intervals $(\bar C_{t}^{\smash{(m)}})_{t \geq m}$ is said to form a $(1-\alpha)$-\uline{Asymptotic Confidence Sequence} (AsympCS) for a parameter $\theta_\P$ if 
\begin{equation}\label{eq:def-lim-asympcs}
\liminf_{m \to \infty} \PM{\forall t \geq m: \theta_{\P}  \in \bar C_{t}^{(m)}} \geq 1-\alpha.
\end{equation}
If the limit inferior is a limit and the inequality is an equality, then the AsympCS is said to be \uline{sharp}.
\end{definition}
The statement in \cref{eq:def-lim-asympcs} can be interpreted as saying: for large enough initial data collection $m$, 
the analyst can peek at new data and stop early without inflating type-I error too much. The methods that we derive in this work will all be ``sharp'', but we provide the more general definition for completeness.
\begin{remark}
It is important to distinguish AsympCSs from delayed-start
\emph{nonasymptotic} confidence sequences. For instance, \citep[Eq.~20]{robbins_statistical_1970} gives a formula to construct a nonasymptotic $\bar{C}^{\rm R}_{t,m}$ for some $m\geq 1$ that has the property 
\[
\PM{\forall t \geq m: \theta_{\P} \in \bar{C}^{\rm R}_{t,m}}  \geq 1-\alpha.
\]Note that this probability must hold for all finite samples greater than or equal to some finite $m$ whereas the probability statement in \cref{eq:def-lim-asympcs} need only hold in a limiting sense as $m \to \infty$.
\end{remark}

\section{Confidence Horizons}\label{sec:power-q}
In this section, we define and derive our main methodological contribution: \emph{asymptotic confidence horizons} or simply ``confidence horizons'' for short. We show how it is possible to form confidence horizons that decay at a rate $t^{-q}$ for any given $q \in \R$ over a bounded time horizon. The class of power $q$ boundaries differs from recent work in anytime-valid inference, where boundaries tend to scale as $\bar C_t \asymp (\log t/t)^{1/2}$ or $(\log \log t/t)^{1/2}$ \citep{waudby-smith_time-uniform_2024, howard_time-uniform_2021}. For $q \in \{1/2, 1\}$, we recover a continuous version of common group sequential designs from \citet{pocock_group_1977} and \citet{obrien1979multiple}; see \cref{sec:connection-gs}.

\begin{figure}[!hb]
    \centering
\includegraphics[width=0.8\linewidth]{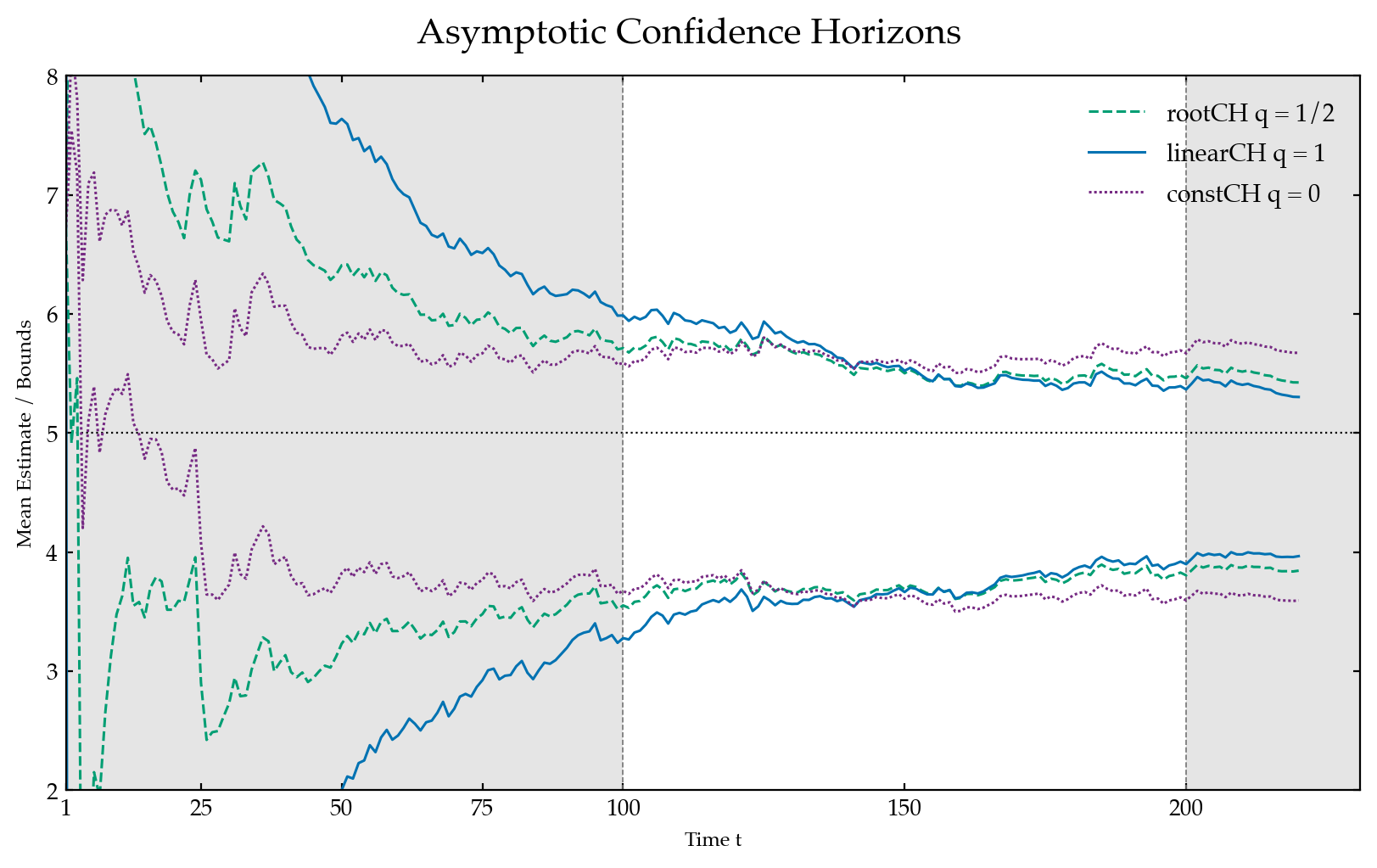}
    \hfill
   \caption{Example of possible confidence horizons illustrating their differences from confidence sequences. The rootCH ($q=1/2$, green) and the linearCH ($q=1$, solid blue) are the continuous analogues of Pocock and O'Brien--Fleming respectively. While the rootCH is tighter for most of the window, the linearCH is more aggressive at the end of the window. 
   \figcode{figure3}}
\end{figure}

\subsection{Defining Confidence Horizons}
Recall from \cref{def:AsympCS} that AsympCSs require approximate uniform coverage for an infinite horizon $\{m, m+1, m+2,\dots \}$ beyond some initial start time $m$.
A confidence horizon aims to forgo that coverage beyond $m \Delta$ for some $\Delta \geq 1$, hence ultimately aiming for approximate validity in the 
window $[m, \Delta m]$. A formal definition follows.

\begin{definition}[Asymptotic Confidence Horizon]\label{def:asympch}
Let $\infseqt{X} \sim \cP$ with parameters of interest $\paren{\theta_\P}_{\P \in \cP}$. Let $m \in \N$ and $\Delta \geq 1$. The sequence of intervals $\paren{\bar C_{t}^{\smash{(m,\Delta)}}}_{t \in [m, \Delta m]}$ is said to form a $(1-\alpha)$-\uline{Asymptotic Confidence Horizon} (AsympCH) for $\paren{\theta_\P}_{\P \in \cP}$ if for every $\P \in \cP$
\begin{equation}\label{eq:asympch-def}
 \liminf_{m \to \infty} \PM{\forall t \in [m, \Delta m]: \theta_\P \in \bar C_{t}^{(m, \Delta)}} \geq 1-\alpha.
\end{equation}If the limit inferior is a limit and the inequality is an equality, then the AsympCH is said to be \uline{sharp}.
\end{definition}
Since we do not consider nonasymptotic confidence horizons in this work, we sometimes omit the ``asymptotic'' qualifier when referring to asymptotic confidence horizons.
The statement in \cref{eq:asympch-def} can be interpreted as saying: for large enough initial sample size $m$, the analyst can peek at new data and stop early without inflating type-I error, up until some prespecified horizon $m\Delta$. 

The guarantee in \cref{eq:asympch-def} is pointwise over $\cP, \Delta,$ and $\alpha \in (0,1)$. We now consider uniform variants of \cref{def:asympch}. 

\begin{definition}[Weakly and Strongly Uniform Confidence Horizons]\label{def:uniform-asympch}
Let $\infseqt{X} \sim \cP$ with parameters of interest $\paren{\theta_\P}_{\P \in \cP}$. Let $m \in \N$ and $\Delta \geq 1$. Fix $\Deltabar \in [1, \infty)$. The collection of intervals $\paren{\bar C_{t}^{\smash{(m,\Delta)}}}_{t \in [m, \Delta m]}$ is said to be a $\cP$-\uline{weakly uniform} $(1-\alpha)$-AsympCH for $\paren{\theta_\P}_{\P \in \cP}$ if
\begin{equation}\label{eq:all-uniform}
\lim_{m \to \infty}\sup_{\Delta \in [1, \Deltabar]}\sup_{\alpha \in (0, 1)}\sup_{\P \in \cP} \abs{\PM{\forall t \in [m, \Delta m]: \theta_\P \in \bar C_{t}^{(m, \Delta)}}- (1-\alpha)} = 0.
\end{equation}
It is said to be $\cP$-\uline{strongly uniform} if the same holds with $\Deltabar = \infty$. 
\end{definition}

 For all $q \neq 1/2$, the confidence horizons are strongly uniform; when $q = 1/2$ the confidence horizon is shown to be weakly uniform. 
\subsection{Confidence Horizons for Independent Data}
We now use confidence horizons for a canonical inference problem: repeated inference of the mean of an \iid{} sequence. As a warmup, we state a special case of our main result (\cref{theorem:general-mean}) for $q = 1/2$ meaning the boundary decays at a $(1/t)^{1/2}$ rate. We may call this fundamental AsympCH the ``rootCH".

\begin{proposition}[Confidence Horizons for the Mean when $q=1/2$]\label{prop:asympch-q-1/2}
Let $\infseqt{X} \simiid \P$ with mean $\mu_\P$, variance $\sigma^2_\P > 0$, and for $\kappa > 2, \EE{\abs{X_1}^\kappa} < \infty$. Let $\hat \mu_t$ be the sample mean and $\hat \sigma^2_t$ the sample variance. Define, for a standard Wiener process $W(t)$, the random variable
\begin{equation}\label{eq:x-delta-q}
\zeta(\Delta, q) = \sup_{s \in [1,\Delta]}\abs{W(s)}s^{q-1},
\end{equation}with a distribution function $\Psi(x; \Delta, q) = \PM{\zeta(\Delta, q) \leq x}$. Then we have
\begin{equation}\label{eq:asympch-q-1/2}
\bar C_{t}^{(\Delta)} \equiv \hat \mu_t \pm \frac{\hat \sigma_t \Psi \n(1-\alpha; \Delta, 1/2)}{\sqrt{t}}.
\end{equation}forms a sharp $(1-\alpha)$-AsympCH for $\mu_\P$ over $[m, \Delta m]$, meaning 
\[
\lim_{m \to \infty}\PM{\forall t \in [m, \Delta m]: \mu_\P \in \bar C_{t}^{(\Delta)}} = 1-\alpha.
\]
\end{proposition}
\cref{prop:asympch-q-1/2} is the consequence of a more general result in \cref{theorem:general-mean}. However, we first presented the \iid{} case with a square root boundary due to its similarity to familiar CLT-based confidence intervals. In fact, one can trivially obtain the latter as a corollary of \cref{prop:asympch-q-1/2}. 

\begin{corollary}\label{corollary:wald-type}
Assume the same as in \cref{prop:asympch-q-1/2}. Taking $\Delta = 1$ returns the familiar large-sample confidence interval for the mean of \iid{} random variables. That is, 
\begin{equation}\label{eq:wald-type-ci}
\dot C_{m}^{\rm CLT} = \hat \mu_m \pm \frac{\hat \sigma_m z_{1-\alpha/2}}{\sqrt m}
\end{equation}is an asymptotic $(1-\alpha)$-confidence interval.
\end{corollary}
\begin{proof}[Proof of \cref{corollary:wald-type}]
Note we have for $\Phi$, the standard normal CDF,
\[
\Psi\n(1-\alpha; 1, 1/2) = \Phi\n(1-\alpha/2) = z_{1-\alpha/2}.
\] Replacing $\Psi \n (1-\alpha; \Delta, 1/2)$ with $z_{1-\alpha/2}$ in \cref{eq:asympch-q-1/2} gives us \cref{eq:wald-type-ci}. The desired result follows from \cref{prop:asympch-q-1/2} after noting that when $\Delta = 1, [m, \Delta m] = \{m\}$.
\end{proof}

Although we motivate confidence horizons by studying \iid{} data drawn from a single distribution $\P$ with a square root decaying boundary ($q = 1/2$), the results extend to more general setups. In this section, we will show that asymptotic confidence horizons can be applied to independent but not necessarily identically distributed data drawn from families of distributions with arbitrary polynomial boundary decay. We emphasize that all of the AsympCHs that follow are weakly uniform as described in \cref{eq:all-uniform}.
Let us first articulate some needed technical conditions on the measures $\cP$ enabling weak uniformity. Let $\infseqt{X}$ be a sequence of independent random variables on $\pspace[\cP]$.
\begin{condition}[Uniform Summability]\label{condition:ui-condition}
Fix $\kappa > 2$. We require that the $\kappa^\text{th}$ moment is $\cP$-uniformly summable in the sense that there exists $\delta \in (0, (\kappa-2)/2)$ such that
\begin{equation}\label{eq:ui-condition}
\sup_{\P \in \cP}\sum_{n = 1}^{\infty}\frac{\EE{\abs{X_n - \EE{X_n}[\P]}^\kappa}[\P]}{n^{1 + \delta}}
<\infty \quad \text{ and } \quad \lim_{m \to \infty}\sup_{\P \in \cP} \sum_{n = m}^{\infty}\frac{\EE{\abs{X_n - \EE{X_n}[\P]}^\kappa}[\P]}{n^{1+\delta}} = 0.
\end{equation}
\end{condition}
Note that in \cref{eq:ui-condition}, the latter condition does not imply the former in general, but it would if $\cP$ were a singleton. A sufficient and more interpretable condition implying \cref{eq:ui-condition} is to require that the $\kappa^\text{th}$ moment is uniformly bounded for some $\kappa > 2$. That is, there exists $B < \infty$ and $\kappa > 2$ for which
\begin{equation}\label{eq:restricted-class-p}
\sup_{t \in \N}\sup_{\P \in \cP}\EE{\abs{X_t}^\kappa}[\P] \leq B.
\end{equation}
The conditions in \cref{eq:ui-condition} would then follow since $\sum_{n=1}^{\infty}{n^{-(1+\delta)}}<\infty$.

\begin{condition}[Variance Stability]\label{condition:no-deterministic}
For some $\rho > 0$ we have
\begin{equation}\label{eq:converging-variance-ind}
\lim_{m \to \infty}\sup_{\P \in \cP}\sup_{t \geq m}t^\rho\abs{\frac{V_{\P, t}}{t} - \sigma^2_{\P}} = 0,
\end{equation}and that
\begin{equation}\label{eq:no-deterministic}
\liminf_{t \to \infty}\inf_{\P \in \cP}\sigma^2_{\P, t} >0.
\end{equation}
\end{condition}
When the data are \iid{}, \cref{condition:no-deterministic} is equivalent to the one required for 
\citep[Theorem 2.2]{waudby-smith_time-uniform_2024}.

Lastly, we require a strongly consistent variance estimator with a $\log$ rate.
\begin{condition}[Strongly Consistent Variance Estimator]\label{condition:strong-consistent-var-estimator}
Given an estimator $\hat \sigma_t^2$, we require $\cP$ consistent estimation, meaning
\[
\forall s > 0, \lim_{m \to \infty} \sup_{\P \in \cP} \PM{\sup_{t \geq m}\log t\abs{\hat \sigma_t^2 / \Tilde \sigma^2_{\P, t} - 1} \geq s} = 0.
\]
\end{condition}
We remark that estimators typically satisfy this rate of convergence except in some pathological cases. In the \iid{} setting, for example, the usual sample variance estimator satisfies \cref{condition:strong-consistent-var-estimator} by the distribution-uniform strong law of large numbers (SLLN) \citep{waudby-smith_distribution-uniform_2024}. For independent random variables, the sample variance satisfies \cref{condition:strong-consistent-var-estimator} if the variation in means vanishes, i.e., 
\begin{equation}\label{eq:mean-variation}
t \n \tsum{(\mu_{\P, i} - \Tilde \mu_{\P, t})^2} = o(1/\log t).
\end{equation}We will discuss later the importance of variance estimation; even if \cref{condition:strong-consistent-var-estimator} is met, different heuristic derivations may elicit better empirical performance.

With these conditions in mind, we are ready to derive confidence horizons for means of independent random variables. 

\begin{theorem}[Asymptotic Confidence Horizons for Independent Data]\label{theorem:general-mean}
Let $\infseqt{X}$ be a sequence of independent random variables on $\pspace[\cP]$ that satisfy \cref{condition:ui-condition} with parameter $\kappa$ and \cref{condition:no-deterministic} with parameter $\rho$. Let $\hat \mu_t = t\n \tsum{X_i}$ be the sample mean and $\hat \sigma_t$ be any sample variance that satisfies \cref{condition:strong-consistent-var-estimator}. Then for any $q \in \R$,
\begin{equation}\label{eq:asympch-formula}
\bar C_{t}^{(m, \Delta)} \equiv \paren{\hat \mu_t \pm \mathfrak{B}_t^{(m, \Delta)}}\coloneqq \paren{\hat \mu_t \pm \hat \sigma_t (m/t)^qm^{-1/2} \Psi\n(1-\alpha; \Delta, q)}
\end{equation}forms a sharp $(1-\alpha)$-AsympCH for $\paren{\Tilde \mu_{\P, t}}_{\P \in \cP}$ over $[m, \Delta m]$. For $q \neq 1/2$, the confidence horizon is $\cP$-strongly uniform and if $q = 1/2$, the confidence horizon is $\cP$-weakly uniform.

\end{theorem}
The proof of \cref{theorem:general-mean} is in \cref{proof:general-mean}. The theorem states for an $m$ sufficiently large, we can guarantee with $1-\alpha$ probability that $\Tilde \mu_{t}$ lies in $\bar C_{t}^{\smash{(m, \Delta)}}$ for all times in $[m, \Delta m]$.

For $q \in \{ 0, 1/2, 1\}$, we give closed-form integral expressions of the distribution function $\Psi(x; \Delta, q)$ in \cref{sec:exact-quantiles-q01} and \cref{sec:exact-quantiles-q12}. 
The two-sided distribution functions have appeared in recent work on so-called demonstration experiments \citep{imbens_demonstration_2026}. 
For a general $q$, closed-form distribution functions are unknown to us. Nevertheless for general $q$, we can calculate $\Psi\n(\alpha; \Delta, q)$ through Monte Carlo simulation with arbitrary precision, but these may be anticonservative for any fixed level of precision.

Note that the AsympCH defined in \cref{theorem:general-mean} is two-sided. However, it is possible to derive an analogous one-sided boundary that can be more powerful for testing one-sided nulls at the expense of power for the opposite directions. Consider the random variable $\zeta_+(\Delta, q)$ defined as
\begin{equation}\label{eq:x-1-delta-q}
\zeta_+(\Delta, q) = \sup_{s \in [1, \Delta]}W(s)s^{q-1},
\end{equation} with a distribution function $\Psi_{+}(x; \Delta, q) = \PM{\zeta_+(\Delta, q) \leq x}$; $x \in \R$. 

\begin{proposition}[One-Sided Asymptotic Confidence Horizons]\label{prop:one-sided-confidence-horizons}
Let the assumptions be the same as \cref{theorem:general-mean}. Then for any $q \in \R$,
\begin{align}
\bar L_{t}^{(m, \Delta)}\equiv \paren{\hat \mu_t - \mathfrak{L}_t^{(m, \Delta)}, \infty}&\coloneqq \paren{\hat \mu_t -  \hat \sigma_t (m/t)^qm^{-1/2} \Psi_+\n(1-\alpha; \Delta, q), \infty}\text{ and }\\
\bar U_{t}^{(m, \Delta)}\equiv \paren{-\infty, \hat \mu_t + \mathfrak{U}_t^{(m, \Delta)}}&\coloneqq \paren{-\infty, \hat \mu_t + \hat \sigma_t (m/t)^qm^{-1/2} \Psi_+\n(1-\alpha; \Delta, q)}
\end{align} form sharp $(1-\alpha)$-AsympCHs for the running mean $\Tilde \mu_{\P, t}$ over $t \in \N \cap [m, \Delta m]$. For $q \neq 1/2$, the confidence horizon is $\cP$-strongly uniform, with the quantiles restricted to be $\alpha \in (0,1-\alpha_1]$ for any $\alpha_1 > 1/2$. If $q = 1/2$, it is $\cP$-weakly uniform.
\end{proposition}
The proof of \cref{prop:one-sided-confidence-horizons} is in \cref{proof:one-sided-mean}
We remark that the functional form of the one-sided and two-sided confidence bounds are identical up to a change in the quantile of a supremum of scaled Wiener processes. The difficulty arises from the fact that $\Psi\n(1-\alpha; \Delta, q) \neq \Psi_{+}\n(1-\alpha/2; \Delta, q)$. The phenomenon of having different boundaries for one- and two-sided sequential inference is not unique to confidence horizons. For example, see the Robbins--Siegmund distributions in \citet{waudby-smith_distribution-uniform_2026}. Note, however, that this distinction is inconsequential for fixed-time inference. That is, when $\Delta = 1$, it holds that $\Psi\n(1-\alpha; \Delta, q) = \Psi_{+}\n(1-\alpha/2; \Delta, q)$. Expressions of $\Psi_{+}(\alpha; \Delta, q)$ are provided in \cref{sec:exact-quantiles-q01,sec:exact-quantiles-q12} when $q \in \{0, 1/2, 1\}$. For $q \in \{0, 1\}$ we only need a bivariate normal calculation, whereas $q =  1/2$ uses an integral expression.

When $q \in \curly{0,1}$, the $\Psi_+$ expression can be simply expressed as
\begin{align}
\Psi_+(x; \Delta, 0) &= 2\Phi_2\paren{x\sqrt{\Delta}, x; 1/\sqrt{\Delta}} - \Phi(x\sqrt{\Delta}) \text{ and }\label{eq:dist-function-0} \\
\Psi_+(x; \Delta, 1) &= 2\Phi_2\paren{x, x/\sqrt{\Delta}; 1/\sqrt{\Delta}} - \Phi(x) \label{eq:dist-function-1},
\end{align}for $\Phi_2$ the bivariate normal distribution function.

\begin{remark}[On the closed-form distribution functions $\Psi(x; \Delta, 1/2)$ and $\Psi_+(x; \Delta, 1/2)$]
The probability a Wiener process crosses a square root boundary over a finite horizon has been studied in the probability and statistics literature for some time. The first attempt is usually credited to \citet{darling_first_1953} and \citet{bellman_recurrence_1951}, who give the density. \citet{breiman_first_1967} recounts these contributions, but focuses on the distribution function of an empirical process. \citet{de_long_crossing_1981} considers a generalization of the hitting probability for $d$-dimensional Bessel processes, for which $d=1$ recovers this previous work. While the classical literature studied the two-sided problem $\Psi(x; \Delta, 1/2)$, little attention was spent on how $\Psi(x; \Delta, 1/2)$ could be used for anytime-valid inference or thinking of $\Psi(x; \Delta, q)$ for general $q$. More recently, there has been interest in deriving equations for hitting times of the Ornstein--Uhlenbeck process\citep{alili_representations_2005,blanchet-scalliet_pseudo-likelihood_2024, blanchet-scalliet_distribution_2025}. We use these recent results to derive the two-sided distribution function, and in addition, derive the one-sided distribution function $\Psi_+(x; \Delta, 1/2)$, which has been historically overlooked in the literature.
\end{remark}

\paragraph{On the Duality between Confidence Regions and P-values}
The present work uses the language of confidence regions (e.g.~confidence intervals, confidence sequences, and so on), but these are not the only statistical inference tools implied by the theory derived herein. In the following proposition, we present p-values and hypothesis tests whose validity under a given null hypothesis holds over a bounded time horizon.

\begin{proposition}[P-values and Hypothesis Tests for the Mean over a Bounded Horizon]\label{prop:sequential-testing-horizon}
Let $\infseqt{X}$ be a sequence of \iid{} random variables on $\pspace[\cP]$ satisfying \cref{condition:ui-condition} and \cref{condition:no-deterministic}. Let $\hat \mu_t = t\n \tsum{X_i}$ be the sample mean and $\hat \sigma_t$ be any variance estimator satisfying \cref{condition:strong-consistent-var-estimator}. For some $\mu_0 \in \R$, consider the weak null hypothesis
\begin{equation}
    \cP_0 = \left \{ \P \in \cP : \forall t \in [m, \Delta m], ~~ \mu_\P \leq \mu_0 \right \}
\end{equation}
For $\Delta \geq 1$, consider the sequence of values $p_m, \dots, p_{\lfloor m\Delta \rfloor}$ given by
\begin{equation}
    p_t := 1- \Psi_+ \left ( \frac{\sqrt{m}\paren{\widehat \mu_t - \mu_0}}{\widehat \sigma_t  (m / t)^q}; \Delta, q \right )
\end{equation}
for each $t$. This sequence forms a sharp weakly uniform p-value for $\cP_0$ on the horizon $[m, \Delta m]$, meaning that for any $\Deltabar \geq 1$ and fixed $\alpha_1 > 1/2$,
\begin{equation}
    \lim_{m \to \infty}\sup_{\Delta \in [1, \Deltabar]}\sup_{\P \in \cP_0} \sup_{\alpha \in (0, 1-\alpha_1]} \left [ \P \left ( \exists t \in [m, \Delta m] : p_t \leq \alpha \right ) - \alpha \right ]_+ = 0.
\end{equation}
\end{proposition}
An analogous result can be stated for the two-sided case but we omit it for the sake of brevity.

\paragraph{Variance Estimation}
The analyst can incorporate one of many variance estimators when deriving AsympCHs, as long as they satisfy \cref{condition:strong-consistent-var-estimator}. However, as hinted at in \cref{eq:mean-variation}, variance estimators could potentially depend on the behavior of the means, where in some cases the raw sample variance may not satisfy \cref{condition:strong-consistent-var-estimator}. Even if the sample variance estimator satisfies the theoretical guarantees, in practice the sample variance may give wider bounds because it conflates the mean signal with the variance signal. While the focus of this paper is not on variance estimation for independent but non-\iid{} observations, we demonstrate informally that residual-type estimators can be stronger.
\begin{figure}[!htbp]
    \centering
\includegraphics[width=0.8\linewidth]{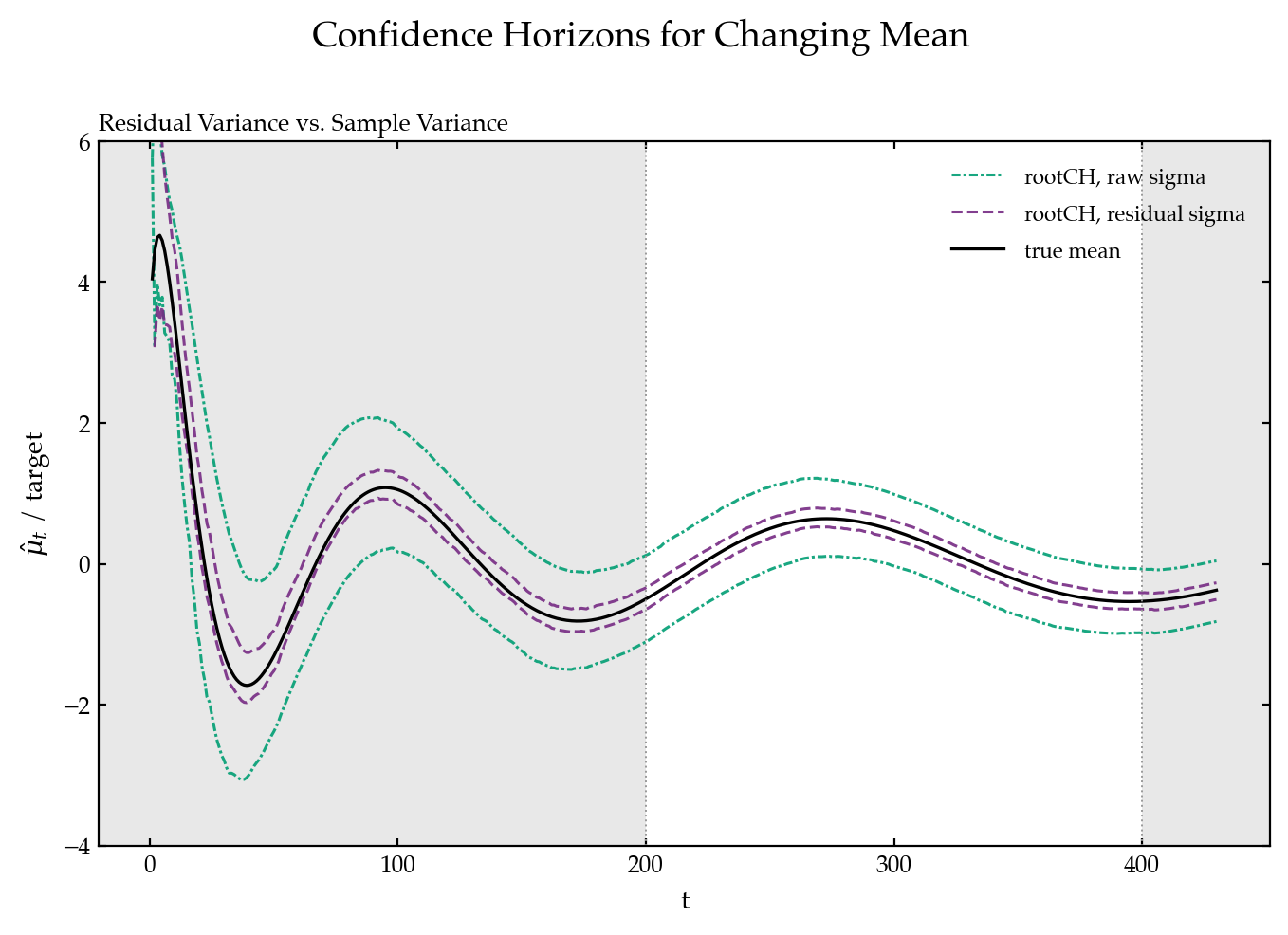}
    \caption{$\infseqt{X}$ are independent but not identically distributed random variables. We let $X_t = \mu_t + \eps_t$ where $\mu_{t} = 5 \sin \paren{3/10 \pi \sqrt{t}}$ and $\eps_t \simiid \mathrm{Laplace}(0, 1/\sqrt{2})$. We see that using a residualized variance estimator like in \cref{eq:residual-var-estimator} gives tighter bounds, suggesting that one should consider variance estimation carefully when dealing with time-varying means. \figcode{figure4}}
    \label{fig:martingale-sine-raw-vs-resid}
\end{figure}
We display a sinusoidal running mean in \cref{fig:martingale-sine-raw-vs-resid} and compare AsympCHs using two different variance estimators. The raw variance estimator is the usual 
\(
\hat \sigma_t^2  = t\n\tsum{X_i^2} - \hat \mu_t^2.
\) A residual-type estimator discussed by \citet{hall_asymptotically_1990} attempts to isolate the noise terms giving 
\begin{equation}\label{eq:residual-var-estimator}
\hat \sigma_{t+1}^2 = (2t)\n\sum_{i=2}^{t+1}{(X_i - X_{i-1})^2} \approx (2t)\n\sum_{i=2}^{t+1}(\eps_i - \eps_{i-1})^2.
\end{equation}Since $\eps_i, \eps_{i-1}$ are independent and mean zero, we are summing the squares of mean zero variance $2\sigma^2$ random variables. Hence, we divide by $2t$. Residualizing away the changing mean gives closer variance estimates and hence tighter confidence horizons. This can be seen in \cref{fig:martingale-sine-raw-vs-resid}.

\subsection{Confidence Horizons under Martingale Dependence}\label{sec:martingales}
This section presents asymptotic confidence horizons that are valid under martingale dependence. Constructing AsympCHs for data with martingale dependence allows the analysis of a wider variety of problems, and we expand on one such problem in \cref{sec:neyman-allocation}.

Let $\infseqt{Y}$ be a stream of data defined on $\pspace$ where we recall for $\cF_{t} = \sigma \paren{Y_1, \dots, Y_t}$, $\mu_{\P, t} = \EE{Y_t \mid  \cF_{t-1}}[\P]$ and $\sigma^2_{\P, t} = \EE{Y_t^2 \mid \cF_{t-1}}[\P] - \mu_{\P, t}^2$.
We first state conditions that are sufficient to apply Strassen's strong approximation \citep{strassen_almost_1967}. These conditions are similar in spirit to the conditions required for \cref{theorem:general-mean}, but differ in that they must occur with probability one.

\begin{condition}[Variances Converge Almost Surely]\label{condition:non-vanishing-variances}
We require  that there exists some $\sigma^2_\P > 0$ such that
for a $\rho > 0$,
\[
\PM{\lim_{t \to \infty} t^\rho\abs{\frac{V_{\P , t}}{t} - \sigma^2_{\P}} = 0} = 1.
\]
\end{condition}
Notice that \cref{condition:non-vanishing-variances} differs from \cref{condition:no-deterministic} in that we require the convergence to occur almost surely. The need for almost sure convergence stems from the fact that $V_{\P, t}$ is a random variable under martingale-dependence. It is sometimes easier to show that $\sigma^2_{\P, t}$ converges to $\sigma^2_\P$ almost surely at some rate $\rho$. If $\rho \in (0,1)$ then the average $V_{\P, t}/t$ also converges to $\sigma^2_\P$ at that same rate (see \cref{sec:check-neyman-conditions} for more).

The second requirement is a Lindeberg-type condition on the tails of the data
\begin{condition}[Lindeberg Uniform Integrability]\label{condition:lindeberg-type-mgale}We require that there exists a $\kappa > 2$ such that
\[
\sum_{t=1}^{\infty}\frac{\EE{\paren{Y_t - \mu_{\P, t}}^2\indic{\paren{Y_t - \mu_{\P, t}}^2 > V_t^{2/\kappa}}\mid \cF_{t-1}}[\P]}{V_t^{2/\kappa}} < \infty \text{ almost surely.}
\]
\end{condition}
We note that if there exists some $r > 2$ and $B < \infty$ such that
\[
\sup_t \EE{\abs{Y_t}^r \mid \cF_{t-1}}[\P] < B,
\] then \cref{condition:lindeberg-type-mgale} is satisfied for $\kappa \in (2,r)$ (see \cref{lemma:suff-mgale-condition}). 
Lastly, we must assume that we have access to a strongly consistent variance estimator. That is, we assume \cref{condition:strong-consistent-var-estimator} with $\cP = \P$. While for independent data, the sample variance estimator converges at a polynomial rate governed by the Marcinkiewicz--Zygmund SLLN \citep{marcinkiewicz_sur_1937}, the rates do not translate to martingale-dependent data. If $r > 4$, however, we can use the martingale law of the iterated logarithm (LIL) \citep{stout_hartman-wintner_1970} to obtain polynomial rates.
Given the defined conditions we state the guarantees of AsympCHs under martingale-dependence.

\begin{theorem}[Asymptotic Confidence Horizons under Martingale-Dependence]\label{theorem:mgale-asympcv}
Suppose we have a stream of data $\infseqt{Y}$ defined on some $\pspace$ that satisfy  
\cref{condition:non-vanishing-variances} with parameter $\rho$ and \cref{condition:lindeberg-type-mgale} with parameter $\kappa$. Let $\hat \mu_t = t\n \tsum{Y_i}$ be the sample mean and $\hat \sigma^2_t$ be any variance estimator that satisfies \cref{condition:strong-consistent-var-estimator}. Then for all prespecified $q \in \R$, 
\[
\bar C_{t}^{(m, \Delta)} \equiv \paren{\hat \mu_t \pm \mathfrak{B}_t^{(m, \Delta)}}
\]as defined in \cref{eq:asympch-formula} forms a sharp $(1-\alpha)$-AsympCH for $\Tilde \mu_{\P, t}$ over $[m, \Delta m]$. For $q \neq 1/2$, the confidence horizon is $\P$-strongly uniform and if $q = 1/2$, it is $\P$-weakly uniform.
\end{theorem}
The proof of \cref{theorem:mgale-asympcv} is in \cref{proof:mgale-asympch}. The form of the AsympCH applied to martingale-dependent data is identical to independent data. What differs are the conditions and the uniformity over $\cP$. A one-sided AsympCH for martingale-dependent data can also be developed using the $\Psi_+$ one-sided quantiles described for independent data in \cref{prop:one-sided-confidence-horizons}.

\section{Illustration: Neyman Allocation for Treatment Effect Estimation}\label{sec:neyman-allocation}
In this section, we highlight an application of confidence horizons to online causal inference. Consider running an adaptive experiment with many arms where the analyst aims to simultaneously track estimates of treatment effects and stop early if there exists significant signal that one arm is outperforming others.

Neyman allocation \citep{neyman_two_1934} is a method for assigning treatments to subjects in a way that yields the most efficient estimates of treatment effects. Suppose an experiment has a fixed number of experimental units $n$ and two arms (control and treatment). Neyman demonstrates that in order to minimize the variance of a treatment effect estimator, one should allocate subjects to treatment and control groups proportionally to the respective standard deviations of their outcomes.  

More concretely, let us adopt the potential outcomes framework \citep{neyman_application_1923, rubin_estimating_1974} and let $Y_i(j)$ be the potential outcome the $i$th experimental unit would have received if assigned to arm $j$ and  $\sigma^2_j = \Var{Y(j)}[\P]$. If using an inverse-probability-weighted (IPW) estimator, then the variance of that estimator is minimized by assigning
\begin{equation}\label{eq:neyman-allocation}
    n_j = \frac{\sigma_j}{\sum_{\ell=1}^{J}\sigma_\ell} n
\end{equation}
many units to arm $j$.
Notice that \eqref{eq:neyman-allocation} suggests that arms with higher variance should be assigned more units.
However, $n_j$ depends on the unknown quantities $\sigma_1, \dots, \sigma_J$. The works of \citet{kato2020efficient}, \citet{dai_clip-ogd_2023}, \citet{cook2024semiparametric}, and \citet{neopane_logarithmic_2024} focus on addressing this concern by either analyzing the behavior of plug-in estimators of $\sigma_1, \dots, \sigma_J$ or constructing bespoke online algorithms. Broadly speaking, these algorithms yield estimators whose variances are asymptotically the same as those that would have been attained by Neyman allocation.

We follow the approach of \citet[Section 3]{kato2020efficient} as their methods are simple and admit central limit theorem-type behavior of estimators, which is central to the use of confidence horizons.
Roughly, the data collection algorithm from \citep{kato2020efficient} has a warmup phase where all arms are sampled to obtain some initial estimate of variances. After this warmup phase we allocate according to Neyman's equation \cref{eq:neyman-allocation}. For more details, refer to \citep{kato2020efficient}. 

Let our stream of data be $\paren{A_t, Y_t(1), \dots, Y_t(J)}_{t=1}^{\infty} \sim \P$ where $A_t \in [J]$ are the different treatment arms and $Y_t(j)$ are the potential outcomes which are \iid{}.\footnote{We use the $J$ notation since $K$ is used for number of groups in a group sequential method.} For each $j \in [J]$, $Y_t(j) \allowbreak \simiid \P_j$ where $\mu_j \equiv \EE{Y_1(j)}[\P]$ and $\sigma^2_{j} \equiv \Var{Y_1(j)}[\P]$. That is, each arm comes from some distribution with its own mean and variance. We assume for some $r > 2$ and $B < \infty$:
\begin{align}\label{eq:potential-outcomes-conditions}
\EE{\abs{Y_t(j)}^r}[\P] < B \quad \text{ and }\quad\inf_{j}\sigma^2_j >0.
\end{align}
The potential outcomes are never jointly observed. Thus, we define the observed outcome 
\[
Y_t = \sum_{j=1}^{J}Y_t(j)\indic{A_t = j}.
\] The output of \citep[Algorithm 1]{kato2020efficient} is a sequence of random variables $\infseqtb{A_t,\, Y_t}$ that are notably dependent.
The outcomes of previous units determine the probability of different arms being chosen in the future. Hence, it is necessary to appeal to the martingale structure. To be precise, the sequence $\infseqtb{A_t,\, Y_t}$ is equipped with a filtration $\infseqt{\cF}$ which we take to be $\cF_t = \sigma \paren{A_1, \dots, A_t, Y_1 \dots, Y_t}$ and a particular functional obeys a nice martingale-dependence. 

First define the following estimators of the mean of arm $j$ at time $t$
\[
\hat \mu_t(j) = \frac{\tsum{Y_i \indic{A_i = j}}}{\tsum{\indic{A_i = j}}}
\]Essentially, this estimator subsets to units that received arm $j$ and takes the average of this subset. Now define $ Z_t(j)$ to be
\[
Z_{t}(j) = \frac{\paren{Y_t - \hat \mu_{t-1}(j)}\indic{A_t = j}}{\PM{A_t = j \mid \cF_{t-1}}} + \hat \mu_{t-1}(j) \quad \text{ and } \quad \bar Z_t(j) = t\n \sum_{i=1}^{t} Z_i(j).
\]
The estimator $Z_t(j)$ is equivalent to the A2IPW estimator from \citep{kato2020efficient}. It is straightforward to check that $\EE{Z_{t}(j)  \mid \cF_{t-1}} = \mu_j$. The other conditions required for \cref{theorem:mgale-asympcv} are less straightforward and are discussed in \cref{sec:check-neyman-conditions}. We indeed have the following guarantee from \cref{theorem:mgale-asympcv}.

\begin{proposition}[Confidence Horizons under Neyman Allocation]\label{prop:conf-horizons-neyman-alloc}
Let $\infseqtb{A_t,\, Y(1)_t, \dots, Y(J)_t}$ be a sequence of random variables on $\pspace$ that satisfy \cref{eq:potential-outcomes-conditions}. Assume the confidence horizon variance estimator $\hat \nu^2_t$ satisfies \cref{condition:strong-consistent-var-estimator}. Using the data collection algorithm from \citet{kato2020efficient}, we can construct a sharp AsympCH that gives valid coverage:
\begin{align}
\forall j \in [J],\lim_{m \to \infty}\PM{\forall t \in [m, \Delta m]: \mu_j \in  \bar Z_t(j) \pm \mathfrak{B}_t^{(m,  \Delta)}} = 1-\alpha.
\end{align}
\end{proposition}

The proof of \cref{prop:conf-horizons-neyman-alloc} is in \cref{sec:check-neyman-conditions}. It amounts to showing that the conditions required for \cref{theorem:mgale-asympcv} are satisfied.

Adaptive experiments \citep{liang_experimental_2023, hadad_confidence_2021} have gained much attention with the contributions from both bandit \citep{lattimore_bandit_2020} and causal inference \citep{neyman_application_1923, rubin_estimating_1974} literatures. There is sometimes a tension between optimizing cumulative rewards (i.e. finding the best arm) and optimizing for inference; see e.g., \citet[Proposition 6.2]{dai_clip-ogd_2023}. The algorithm from \citet{kato2020efficient}, and  Neyman allocation more broadly, falls squarely in the latter regime. This need not mean that the current illustration is \emph{only} designed for this latter regime. For instance, \citet{hadad_confidence_2021} gives methods to do asymptotic inference after a bandit-like algorithm is run. Confidence horizons should, although we do not verify, connect with the framework developed by \citet{hadad_confidence_2021}. However, not every bandit algorithm can be applied. Assumptions 1, 2, and 3 in \citep{hadad_confidence_2021} are in fact very similar to our conditions and rule out some algorithms. 

\paragraph{Numeric Experiment}
We now simulate an adaptive experiment that uses \citep[Algorithm 1]{kato2020efficient} with $J = 6$ different arms. We let $\infseqt{Y(j)}\simiid 2{\rm Beta}(a_j, b_j) - 1$, choosing parameters $a_j, b_j$ to give various means, variances, and skewness. The distributions of the outcome variables are shown in \cref{fig:outcome-y-arms}.

We first show the results of one realization of the adaptive experiment. The shaded regions correspond to AsympCH regions instantiated with $q = 1/2$. Subsequently, we focus on the top two arms (arm 5 and arm 6) and compare against uniform allocation and also against AsympCSs.
\begin{figure}[!ht]
    \centering
\includegraphics[width=0.95\linewidth]{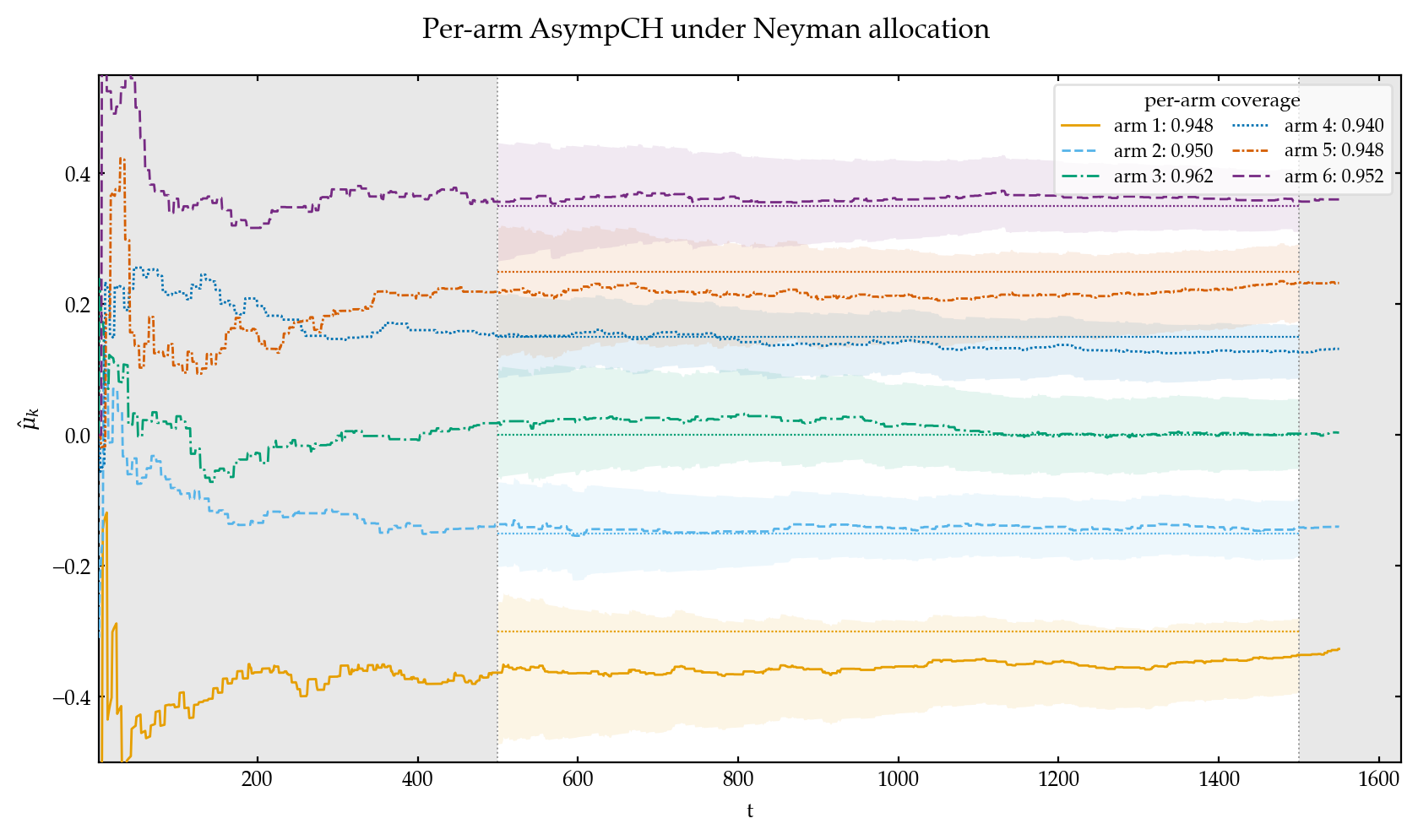}
    \caption{One realization of an adaptive experiment with 6 arms, each with different means and variances. Using \citep[Algorithm 1]{kato2020efficient} induces dependence across $t$ but applying \cref{prop:conf-horizons-neyman-alloc} allows us to bypass the issues created by the dependence. The top right legend shows per-arm coverage estimated by a Monte Carlo simulation. Note per-arm $\alpha$-values can be adjusted using a correction to achieve a family-wise error guarantee. While all per-arm coverages are close to 0.95, some fall beneath. This is a byproduct of the finite-sample performance of an asymptotic guarantee. \figcode{figure8}}
    \label{fig:per-arm-neyman-vs-uniform}
\end{figure}
At $t=500$, the beginning of the window in \cref{fig:per-arm-neyman-vs-uniform}, the confidence horizons for different arms overlap. However, at the end of the window, each arm's confidence horizon is tight enough to separate from other arms. The per-arm coverage in the upper right of \cref{fig:per-arm-neyman-vs-uniform} is calculated using a Monte-Carlo simulation for which the present figure is one such realization. We recall that a miscoverage event is when the mean is not contained in any single point in the region $t \in [500, 1500]$.

\begin{figure}[!htbp]
    \centering
\includegraphics[width=0.95\linewidth]{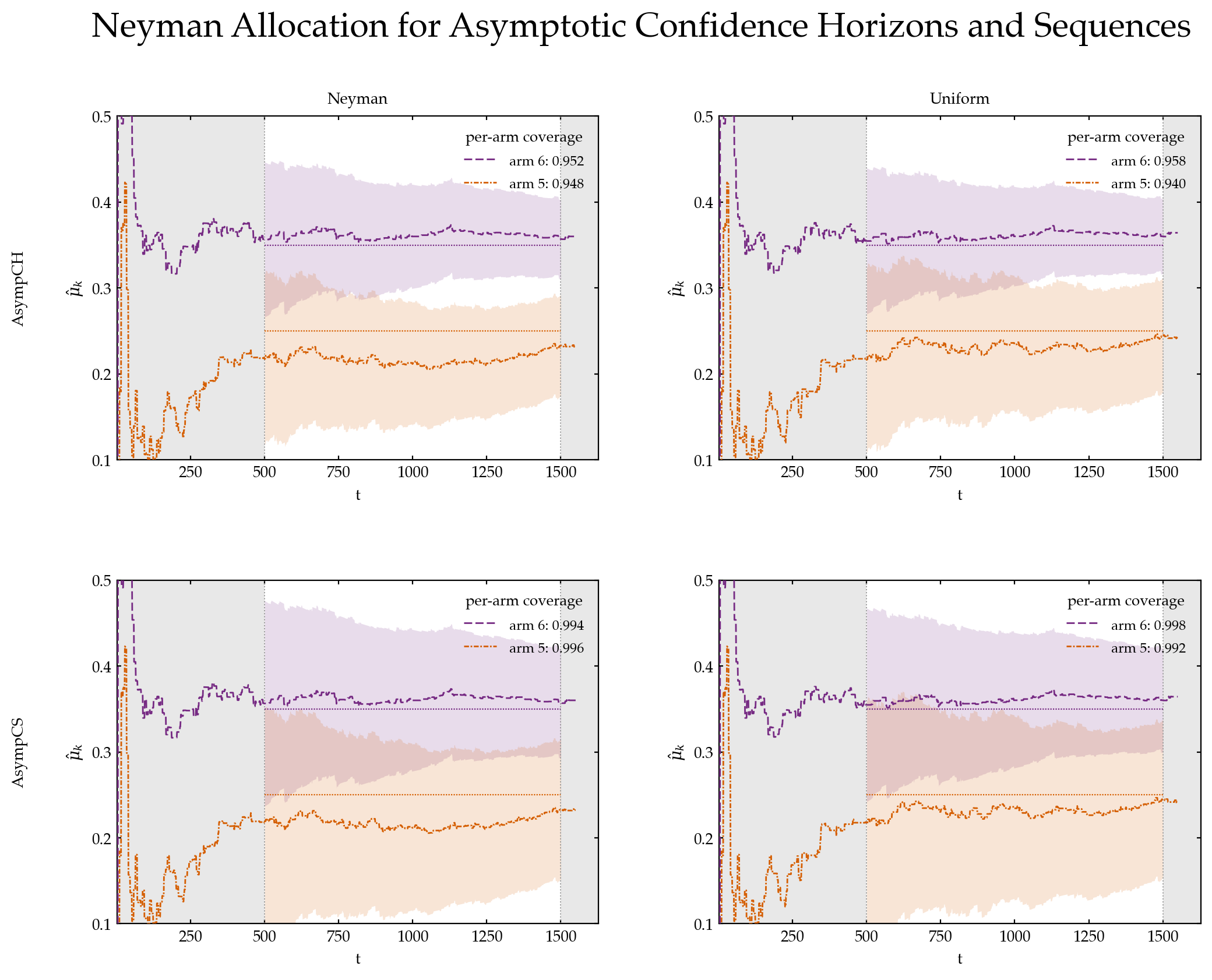}
    \caption{We focus on the top two arms and compare AsympCHs against AsympCSs and \citep[Algorithm 1]{kato2020efficient} against uniform allocation. We see that moving up and to the left, confidence regions are tighter and experimenters can stop the experiment earlier. In the case of both AsympCS experiments, the experimenter can never show with high probability that arm 6 has a higher mean than arm 5. In the first row, which utilizes AsympCHs, the experimenter can stop early for either allocations scheme. Moving from the top right to the top left, we see improvements for both AsympCHs and AsympCSs \figcode{figure9}}
    \label{fig:per-arm-neyman-vs-uniform-top2}
\end{figure}

\begin{remark}[Ability to wait and see]
While the visualization shows the full path until $t = 1500$, if the experimenter is content at time $t \approx 800$ when arm 6 shows a statistically significant improvement upon arm 5, the experimenter can stop early. Perhaps more importantly, the experimenter can stop early at $t \approx 800$ or can continue to monitor the experiment. We highlight \emph{this} ability of confidence horizons in comparison to simulation-based methods \citep{fischer_sequential_2025, besag_sequential_1991}.
\end{remark}

AsympCSs that are valid under martingale dependence can be deployed under certain adaptive experimentation schemes like the one described in \citep[Algorithm 1]{kato2020efficient}; see e.g., \citet{cook2024semiparametric}. In the following experiment, we compare AsympCHs against AsympCS and the benefits of an adaptive experiment against a uniform one. Adaptive experiments are more complex than traditional ones that allocate treatment uniformly across arms. Therefore, these results illustrate the strength of utilizing a potentially more complicated method. In order to simplify the illustration, we ``zoom" in on the top two performing arms: arm 5 and arm 6. Other arms are still receiving allocation, but for visualization purposes they are omitted from the plots (see \cref{fig:full-neyman-experiment} for the full experiment). In general, we see major improvements using AsympCHs against AsympCSs and substantial improvements using Neyman allocation compared to uniform allocation.

\section{Confidence Horizons and Group Sequential Methods}\label{sec:theoretical-conseq-gs}
In this section, we draw some connections between confidence horizons and group sequential methods (GSMs). We show that the confidence horizons discussed in \cref{sec:power-q} can be viewed as continuous extensions of the tests of Wang \& Tsiatis \citep{wang_approximately_1987}, and by extension, those of \citet{pocock_group_1977} and \citet{obrien1979multiple}. 
\subsection{Background on Group Sequential Methods}
We first give a simplified introduction to the tests of \citet{pocock_group_1977} and \citet{obrien1979multiple}; see also the introduction to the textbook of \citet[Ch. 2]{jennison_group_2000}. Suppose $\paren{X_i}_{i=1}^{\infty} \simiid \cN(\mu, \sigma^2)$, with known variance $\sigma^2 > 0$. The data arrive in groups, indexed by $k = 1, \dots, K$ with cumulative sizes $\paren{n_1, \dots, n_K}$. A key parameter is the ratio between the last look and the first look $\Delta = n_K/n_1$, what we have been referring to thus far as the ``multiplicative horizon''. Consider the standardized statistic:
\begin{equation}\label{eq:standardized-statistic}
Z_k = \frac{1}{\sqrt{ n_k\sigma^2}}\sum_{i=1}^{n_k}(X_{i} - \mu) \sim \cN(0,1).
\end{equation}
For a given $\alpha \in (0, 1)$, the essence of Pocock's test is to find a calibrated constant $c_{\rm P}(\alpha)$ for which
\begin{equation}
    \P \left ( \max_{1 \leq k \leq K} |Z_k| \geq c_{\rm P} (\alpha) \right ) = \alpha.
\end{equation}
In other words, Pocock's test aims to calibrate a boundary for the centered partial sum statistic $\sigma^{-1}\sum_{i=1}^n (X_i - \mu)$ with a shape of $\sqrt{n}$ at times $n_1, \dots, n_K$.

Similarly, from this vantage point, the test of \citet{obrien1979multiple} aims to calibrate a boundary for the same statistic but with a constant shape. That is, \citet{obrien1979multiple} calibrate a constant $c_{\rm OBF} (\alpha)$ for which
\begin{equation}
    \P \left ( \max_{1 \leq k \leq K} \sqrt{n_k}\left \lvert Z_k \right \rvert \geq c_{\rm OBF}(\alpha) \sqrt{n_K} \right ) = \alpha.
\end{equation}

\citet{wang_approximately_1987} propose a generalized sequential test that encompasses both the test of Pocock and of O'Brien--Fleming. The test of \citet{wang_approximately_1987} is indexed by a parameter $q \in [1/2, 1]$, recovering the aforementioned two tests at the endpoints.\footnote{The paper of \citet{wang_approximately_1987} indexes their boundaries by a parameter denoted by the symbol $\Delta$. Note that this can be thought of as $1-q$ in the language of our paper, and has no relationship to the multiplicative horizon $\Delta$ appearing throughout our work.} That is, they calibrate a constant $c_{\rm WT}(\alpha)$ for which 
\begin{equation}\label{eq:wt-decision}
\P \left ( \max_{1 \leq k \leq K} (n_k)^{q - 1/2}\abs{Z_k} \geq c_{\rm WT}\paren{\alpha} (n_K)^{q - 1/2} \right ) = \alpha.
\end{equation}
We have left two important questions unanswered thus far: (1) How does one compute the appropriate constant $c_{\rm WT}$? (2) How can these boundaries be used for non-Gaussian data with unknown variances? Let us now discuss answers to these questions.

\paragraph{Calculating the Cutoff}
Notice that the test statistics $Z_1, \dots, Z_K$ in \cref{eq:standardized-statistic} are jointly multivariate Gaussian. That is, $\paren{Z_1, \dots, Z_K} \sim \cN_K(0, \Sigma)$, where 
\begin{equation}\label{eq:gsm-sigma}
\Sigma_{i,j} = \sqrt{\frac{\min(n_i,n_j)}{\max(n_i,n_j)}}.
\end{equation}
In order to find $c_{\rm WT}$ satisfying \cref{eq:wt-decision},
recursive methods such as those employed in \citet{armitage_repeated_1969} can be used to calculate this probability by leveraging the Markovian structure of the sequential statistics and appealing to numerical integration. In total, the algorithms that calculate this probability require $K$ separate univariate integrals. Note that these integrals are performed in a dynamic programming-type structure and hence cannot be fully parallelized. Analysts may also use Monte Carlo simulation to estimate the cutoff values. Despite the similarity between numerical integration and Monte Carlo simulation, calculating cutoffs with simulation appears to be less common in the literature, with numerical integration being broadly preferred (see, for instance \citep{jennison_group_2000}). For more information on how to calculate cutoffs see \citep[Ch 19]{jennison_group_2000} and the \textsf{gsDesign} package \citep{anderson_gsdesign_2026} for modern implementation details.

\begin{remark}[Pitfalls of Repeated Numerical Integration \& Approximate Methods]
One may suggest taking the number of groups $K$ to be the size of the window $m\Delta - m$ in order to replicate the confidence horizon. As we discuss here, computing the quantile $c_{\rm WT}(\alpha, K, q)$ requires computing for its objective function $K$ many integrals. \textsf{gsDesign} \citep{anderson_gsdesign_2026}, the popular software package for GSMs, limits $K \leq 30$ for the tests of \citet{wang_approximately_1987}. The algorithmic dependence is linear in $K$, which may be overcome with modern computational power. Another bottleneck for the $K$-ary numerical integration is the need to increase the resolution parameter (i.e., how accurate is the approximate integral?). As $K$ increases, each individual integral must be made more precise so that these errors do not compound. As such, the resolution parameter must be made finer as $K$ increases. We contrast the explicit dependence on $K$ with the exact quantile formulas we find for AsympCHs in \cref{sec:exact-quantiles-q01} and \cref{sec:exact-quantiles-q12}. The exact quantiles we derive are at most only one numerical integration. See \cref{fig:integration-runtime} in \cref{sec:ld-approx-comparison} for the numerical challenges facing repeated numerical integration. 

One may also suggest bypassing numerical integration by using approximate $\alpha$-spending methods. \citet{lan_discrete_1983} provide one such $\alpha$-spending approximation to the Pocock and O'Brien--Fleming boundaries. The $\alpha$-spending approach obviates the need of repeated numerical integration, thus allowing $K$ to grow larger without increasing discretization error. However, the algorithms face numerical instability when $K$ grows that cannot be fixed with a finer discretization grid. This, combined with the fact that when $K$ is small the approximations falter (see \cref{fig:ld-ch} in \cref{sec:ld-approx-comparison}), illustrates how $\alpha$-spending has its own challenges. 
\end{remark}
\paragraph{Asymptotic Approximations}
In this work, we are interested not in making exact normality assumptions but instead appealing to approximate normality of statistics for large sample sizes. While group sequential trials began by assuming Gaussian data (see e.g. \citep{pocock_group_1977}), asymptotic theory emerged thereafter with work from Jennison \& Turnbull \citep{jennison_group-sequential_1997, jennison_distribution_1997}. Subsequently, \citet{scharfstein_semiparametric_1997} showed that if asymptotic approximations hold for individual analyses then under a broad set of semi-parametric statistics we can apply sequential-type analyses. Thus censored-survival data, which is studied in \citep{tsiatis_sequential_1995}, is one such example of the broad set \citet{scharfstein_semiparametric_1997} considered.\footnote{For a larger literature review on asymptotic approximations for group sequential methods see \citep[Chapter 11.7]{jennison_group_2000}.} We now recount the established asymptotic framework for GSMs.

Let $\infseqt{X} \simiid \P$ where $\EE{X}[\P] = \mu_\P$ and $\Var{X}[\P] = \sigma_\P^2 < \infty$. We prespecify a finite set of groups $k = 1, \dots, K$ with cumulative sizes $n_1, \dots, n_K$. We again construct the standardized statistic $Z_k$, which we define in \cref{eq:standardized-statistic}. However, in \cref{eq:standardized-statistic} the data were normal and therefore $Z_k$ was normal. Since each $Z_k$ was normal we had that the vector $\paren{Z_1, \dots, Z_K}$ was multivariate normal. In the current setting, we do not assume the data follow a Gaussian distribution and thus have to directly appeal to asymptotics.

The predominant way GSMs appeal to asymptotics is through the multivariate CLT \citep[Eq 11.16]{jennison_group_2000}. Collecting the standardized statistics, we have as $n_1 \to \infty$ the vector
\begin{equation}\label{eq:multivariate-normal-clt}
\paren{ Z_1, \dots, Z_K} \rightsquigarrow \cN_K(\mathbf{0}_K, \Sigma),
\end{equation}where $\Sigma$ is determined by known covariances described in \cref{eq:gsm-sigma}. Therefore, calibrated cutoffs can be calculated using Gaussian quantiles. Appealing to a weak or strong law of large numbers under suitable moment conditions allows one to replace the true variance $\sigma_\P^2$ with the sample variance $\hat \sigma_n^2$.

To the best of our knowledge, group sequential methods have only been established with pointwise guarantees over $\cP, \Delta$, and $K$. Let us consider uniform analogues of group sequential methods starting with the following definition.
\begin{definition}[Uniform Group Sequential Methods]\label{def:uniform-gs-methods}
Let $\infseqt{X} \sim \cP$ with parameters of interest $\paren{\theta_\P}_{\P \in \cP}$. Let $m \in \N, K \in \N$ and $\Delta \geq 1$. Fix $\Deltabar \in [1, \infty)$. A collection of intervals $\paren{\bar C_{t}}_{t \in \curly{n_1, \dots, n_K}}$ is said to be a $\cP$-\uline{weakly uniform} $(1-\alpha)$-GSM for $\paren{\theta_\P}_{\P \in \cP}$ if
\begin{equation}\label{eq:gs-all-uniform}
\lim_{n_1 \to \infty}\sup_{\Delta \in [1, \Deltabar]}\sup_{K \in \N}\sup_{\alpha \in (0, 1)}\sup_{\P \in \cP} \abs{\PM{\forall t \in \curly{n_1, \dots, n_K}: \theta_\P \in \bar C_{t}}- (1-\alpha)} = 0.
\end{equation}
It is said to be $\cP$-\uline{strongly uniform} if the same holds with $\Deltabar = \infty$. 
\end{definition}
With this definition in mind, we now demonstrate that existing GSMs enjoy uniform guarantees, and we provide explicit connections between those methods and confidence horizons.
\subsection{Explicit Connections with Group Sequential Methods}\label{sec:connection-gs}
We can now investigate the connections between GSMs and AsympCHs. We will focus our attention on the tests of \citet{wang_approximately_1987}.  We demonstrate that by using the language of AsympCHs, it is straightforward to see that the tests of \citet{wang_approximately_1987} can be applied to independent but not identically distributed and martingale-dependent data and that these approximations are uniform in a variety of ways. In order to make this comparison, we will invert the sequential tests to make the implied confidence regions. 

First, in order to keep notation consistent and help in the proofs, we will define the \citet{wang_approximately_1987} distribution function $\Xi$ where for $W(t)$ a Wiener process,
\begin{align}
\Xi(x; \curly{1, n_2/n_1 \dots, \Delta}, q)= \PM{\max_{s \in \curly{1, n_2/n_1, \dots, \Delta}}\abs{W(s)}s^{q-1} \leq x}\label{eq:wt-dist-function}.
\end{align}
Note that this distribution $\Xi$ is defined as the inverse of quantiles given in \eqref{eq:wt-decision}. Formally, by \cref{lemma: calc-quantile} we have that
\begin{align}
c_{\rm WT}(\alpha, \curly{n_1, \dots, n_K}, q) \Delta^{q-1/2} = \Xi \n(1-\alpha; \curly{1,n_2/n_1, \dots, \Delta}, q).
\end{align}
The following proposition demonstrates that the GSM of \citet{wang_approximately_1987} defined through $\Xi^{-1}$ enjoys uniform guarantees in the sense of \cref{def:uniform-gs-methods}.

\begin{proposition}[Uniformly Valid Group Sequential Methods for Independent Data]\label{prop:unif-valid-gsm-ind}
Suppose $\paren{X_n}_{n=1}^{\infty}$ is an infinite sequence of independent random variables on $\pspace[\cP]$ that satisfy \cref{condition:ui-condition} with parameter $\kappa$ and \cref{condition:no-deterministic} with parameter $\rho$. Let $\hat \mu_n$ be the sample mean and $\hat \sigma_n$ be any variance estimator that satisfies \cref{condition:strong-consistent-var-estimator}. Then,
\begin{equation}\label{eq:radius-wt}
\bar C_{n_k}^{\rm WT} \equiv \paren{\hat \mu_{n_k} \pm \mathfrak{B}_{n_k}^{\rm WT}}\coloneqq\paren{\hat \mu_{n_k} \pm \frac{\hat \sigma_{n_k} \Xi\n\paren{1-\alpha; \curly{1, n_2/n_1, \dots, \Delta},  q} n_1^{ q-1/2}}{n_k^{ q}}}
\end{equation} is a sharp $(1-\alpha)$-GSM on $\{ n_1, \dots, n_K \}$. For $q \neq 1/2$, the GSM is $\cP$-strongly uniform and for $q = 1/2$, it is $\cP$-weakly uniform.
\end{proposition}
The proof of \cref{prop:unif-valid-gsm-ind} is in \cref{proof:unif-valid-gsm-ind}. Identifying $n_1$ with $m$ and $n_1,\dots, n_K$ with $k \in [m, \Delta m]$ in \cref{theorem:general-mean}, we see that our AsympCHs can be viewed as ``continuous'' analogues of the boundaries of \citet{wang_approximately_1987}. This phenomenon is illustrated empirically in \cref{fig:gs-comparison}.
\begin{figure}
    \centering
\includegraphics[width=0.95\linewidth]{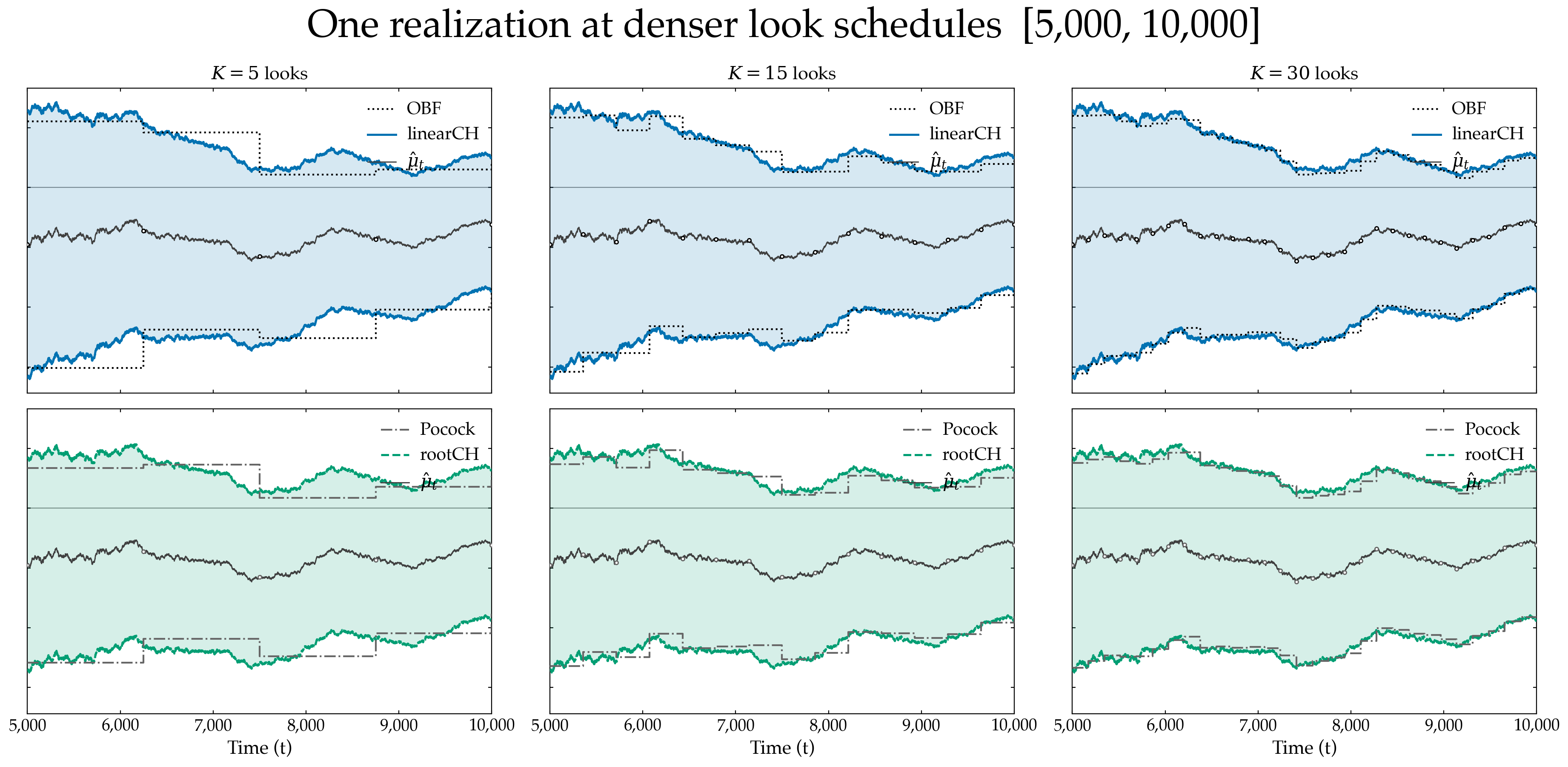}
\caption{Asymptotic confidence horizons form continuous analogues of common group sequential methods. The rootCH is the continuous analogue of Pocock's test \citep{pocock_group_1977}. The linearCH is the continuous analogue of the O'Brien--Fleming test \citep{obrien1979multiple}. \figcode{figure10}}
\label{fig:gs-comparison}
\end{figure}

\begin{remark}[On Uniformity in $K \in \N$]
 While one might suspect that the group sequential methods can be made distribution and quantile-uniform under certain assumptions, uniformity over $K \in \N$ is perhaps more surprising since it is not obvious that applications of multivariate central limit theorems should be uniform over the dimension of those vectors. The results in \cref{prop:unif-valid-gsm-ind} and \cref{prop:unif-valid-gsm-mgale} argue that there exists some uniform rate of convergence that is independent of the number of groups $K$. This implies that analysts need not worry about approximation error due to the number of groups for tests of \citet{wang_approximately_1987}.    
\end{remark}
The following result establishes that GSMs enjoy validity under martingale dependence, akin to \cref{theorem:mgale-asympcv}.

\begin{proposition}[Uniformly Valid Group Sequential Methods for Martingale-Dependent Data]\label{prop:unif-valid-gsm-mgale}
Suppose $\paren{X_n}_{n=1}^{\infty}$ is an infinite sequence of martingale-dependent random variables defined on $\pspace[\cP]$ that satisfy \cref{condition:non-vanishing-variances} with parameter $\rho$ and \cref{condition:lindeberg-type-mgale} with parameter $\kappa$. Let $\hat \mu_n$ be the sample mean and $\hat \sigma_n$ be any variance estimator that satisfies \cref{condition:strong-consistent-var-estimator}. Recalling that 
\begin{equation}
\bar C_{n_k}^{\rm WT} \equiv \paren{\hat \mu_{n_k} \pm \mathfrak{B}_t^{\rm WT}},
\end{equation}we have that $\bar C_{n_k}^{\rm WT}$ is a sharp $(1-\alpha)$-GSM on $\{n_1, \dots, n_K\}$. For a fixed distribution $\P$ and $q \neq 1/2$, the GSM is $\P$-strongly uniform and for $q = 1/2$, it is $\P$-weakly uniform.
\end{proposition}
A consequence of \cref{prop:unif-valid-gsm-mgale} is that the tests of \citet{wang_approximately_1987} can be used in the Neyman allocation example discussed in \cref{sec:neyman-allocation}.

\paragraph{Tradeoffs for Peeking at Any Time on the Horizon}
One natural question is: when should an analyst use group sequential methods and when should they use confidence horizons? We find a fundamental tradeoff between the two as measured through traditional notions of power and expected stopping times. Confidence horizons empirically outperform GSMs when the metric is the expected stopping time---i.e., the expected number of samples required to reject the null---while GSMs empirically outperform confidence horizons when measuring traditional notions of power---i.e., the probability of \emph{ever} rejecting the null under the alternative. See \cref{fig:expected-power-tradeoff}.

\begin{figure}
    \centering
\includegraphics[width=0.95\linewidth]{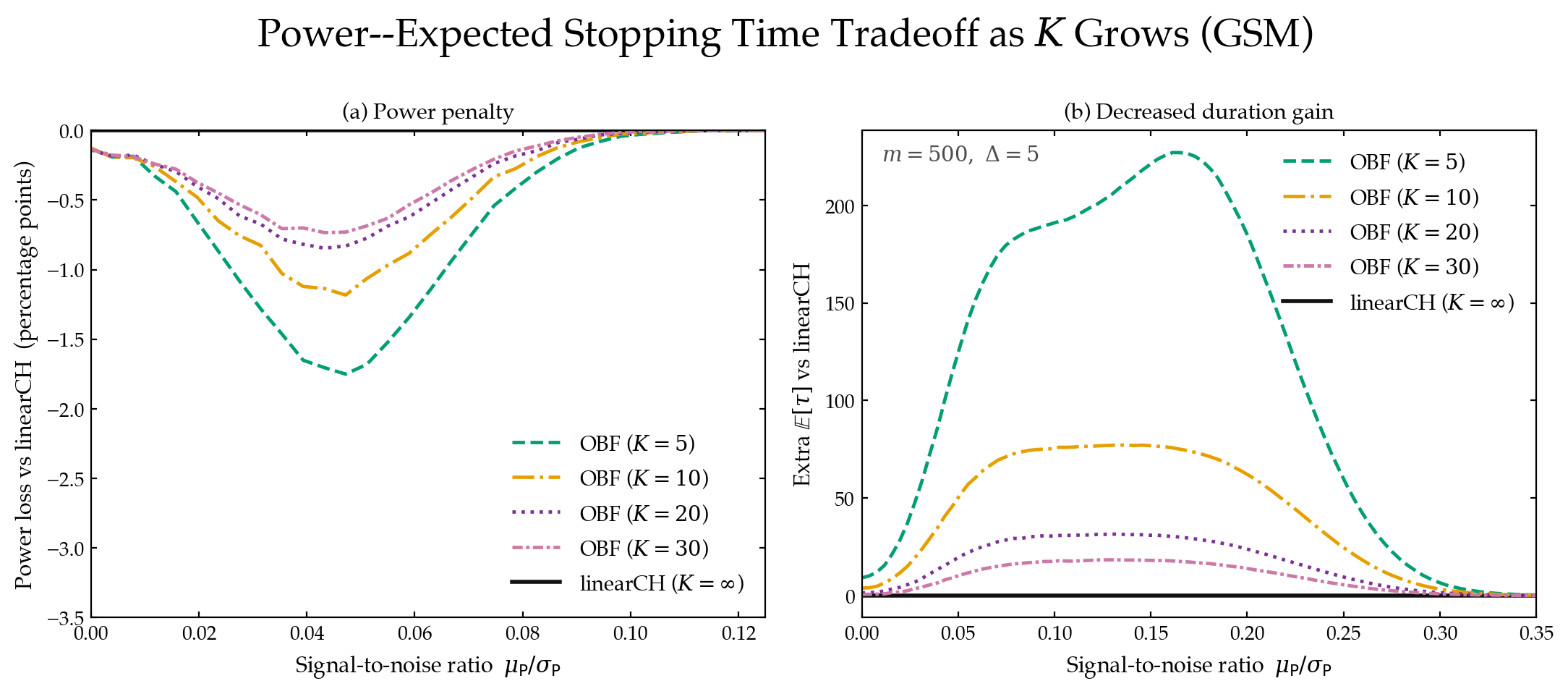}
    \caption{For a Gaussian experiment, we consider the tradeoffs between power and expected stopping time. Power is defined as the probability of rejecting at the specified peeking times. The expected stopping time, which we call $\tau$, is defined to be $\tau = \min \curly{m\Delta, \inf\curly{n: \text{Reject } H_0}}$, where $\inf \emptyset = \infty$. Notice that while the tradeoffs in power are quite small as $K$ varies, the tradeoffs for expected stopping time is more pronounced.\figcode{figure11}}
    \label{fig:expected-power-tradeoff}
\end{figure}

\section{Simulation Studies: Widths and Empirical Coverage}\label{sec:simulation-study}
We now conduct a simulation study focusing on empirical widths and coverage guarantees. We will reserve this section to compare confidence horizons 
with other confidence sequences showing the increased power and exact coverage of confidence horizons. 
We also look at the asymptotic approximation of the confidence horizons. We show the tradeoff for how large the initial time $m$ must be to have approximately valid coverage.
We note that we restrict our analysis to $q \in \{0, 1/2, 1\}$, though AsympCHs are well-defined for any $q \in \R$.
These values of $q$ yield closed-form distribution functions whose quantiles are simple to compute, and we suspect that they will be most commonly adopted in real statistical applications.

We first compare confidence horizons with established confidence sequence methods. We use as our comparators,
the nonasymptotic confidence sequences developed by \citet{howard_time-uniform_2021} and the asymptotically valid confidence sequences considered
by \citet{waudby-smith_time-uniform_2024}.

In \cref{fig:compare-widths-coverage-1}, the first empirical result, we consider the sequence $\infseqt{X} \simiid U(0,1)$.
We fix $m = 1000$ and vary the end point over $\Delta \in \{1.5, 2, 3\}$. All AsympCH methods achieve exact coverage (up to
approximation and Monte Carlo error) at the end of the window. The different $q$ values illustrate how 
$\alpha$ is spent differently over the window. When $q = 1$, for instance, the horizon is more aggressive
at the end of the window, while when $q = 0$ we are at a constant width.
\begin{figure}[!hb]
    \centering
\includegraphics[width=0.95\linewidth]{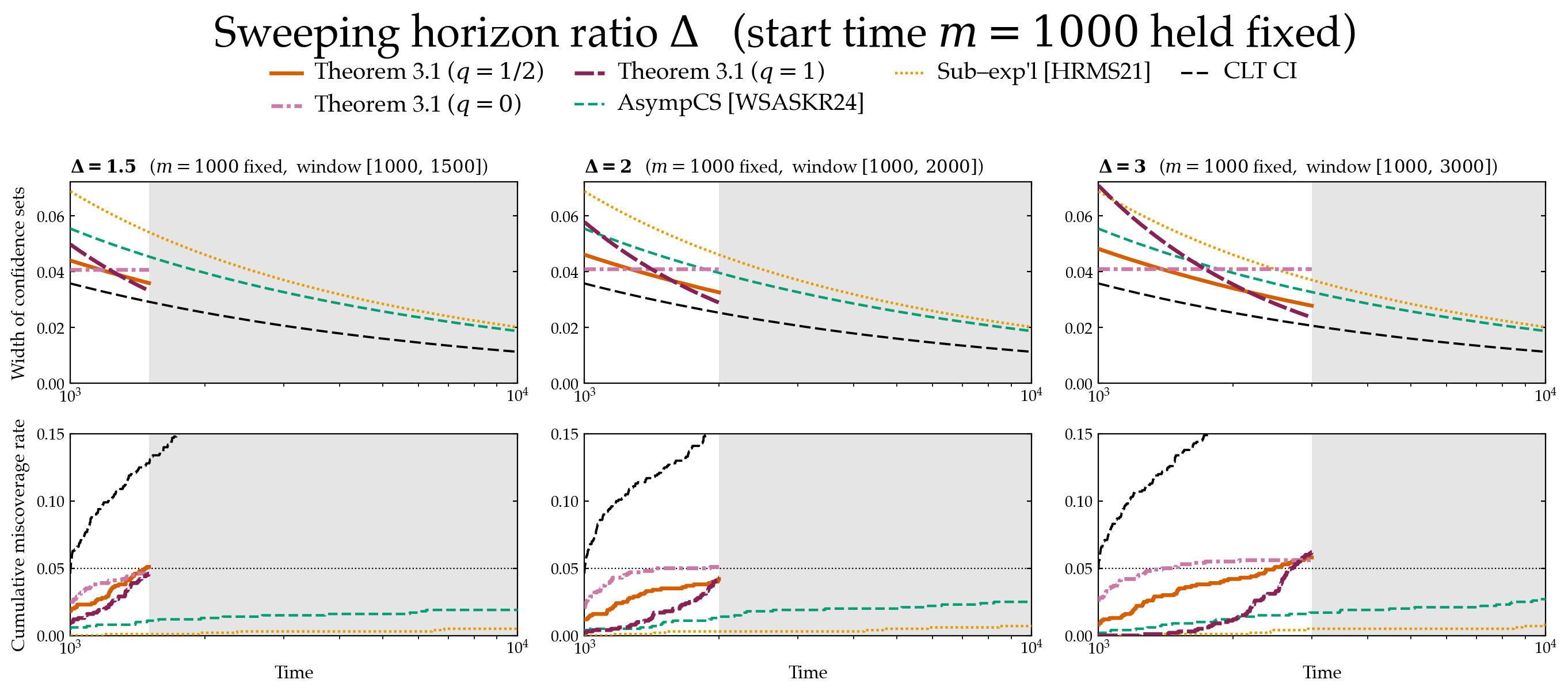}
    \caption{Comparing widths and coverage over a bounded window $[m, \Delta m]$ for $\infseqt{X} \simiid U(0,1)$. 
    Moving to the right, we let $\Delta$ increase hence changing the horizon the confidence horizon must be valid over.
    The top row shows AsympCHs uniformly have the tightest widths. The bottom row shows the coverage is exact over the window.
    The naive CLT-type sequence of confidence intervals shows large miscoverage while the confidence sequences comparators
    are conservative. \figcode{figure12}}
    \label{fig:compare-widths-coverage-1}
\end{figure}

\cref{fig:compare-widths-coverage-2}, the second empirical result, contains the same methods and data
sequences as \cref{fig:compare-widths-coverage-1}. However, this figure keeps $\Delta =2$ constant and instead lets
$m$ vary.
\begin{figure}[!hb]
    \centering
    \includegraphics[width=0.95\linewidth]{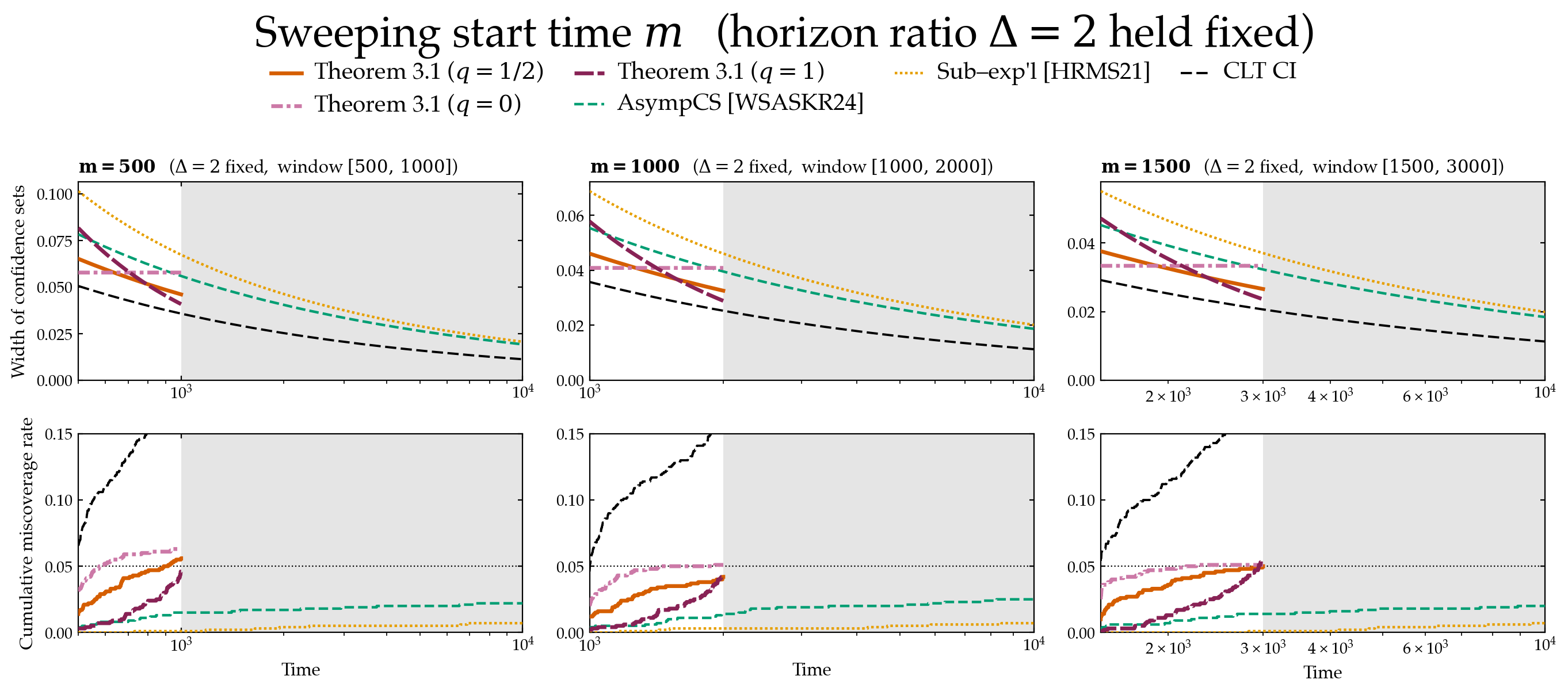}
    \caption{Comparing widths and coverage over a bounded window $[m, \Delta m]$ for $\infseqt{X} \simiid U(0,1)$. 
    Moving to the right, we let $m$ increase hence extending the horizon the confidence horizon must be valid over.
    The top row shows AsympCHs uniformly have the tightest widths. The bottom row shows the coverage is exact over the window.
    The naive CLT-type sequence of confidence intervals shows large miscoverage while the confidence sequences comparators
    are conservative. \figcode{figure13}}
    \label{fig:compare-widths-coverage-2}
\end{figure}
\begin{figure}
    \centering
\includegraphics[width=0.9\linewidth]{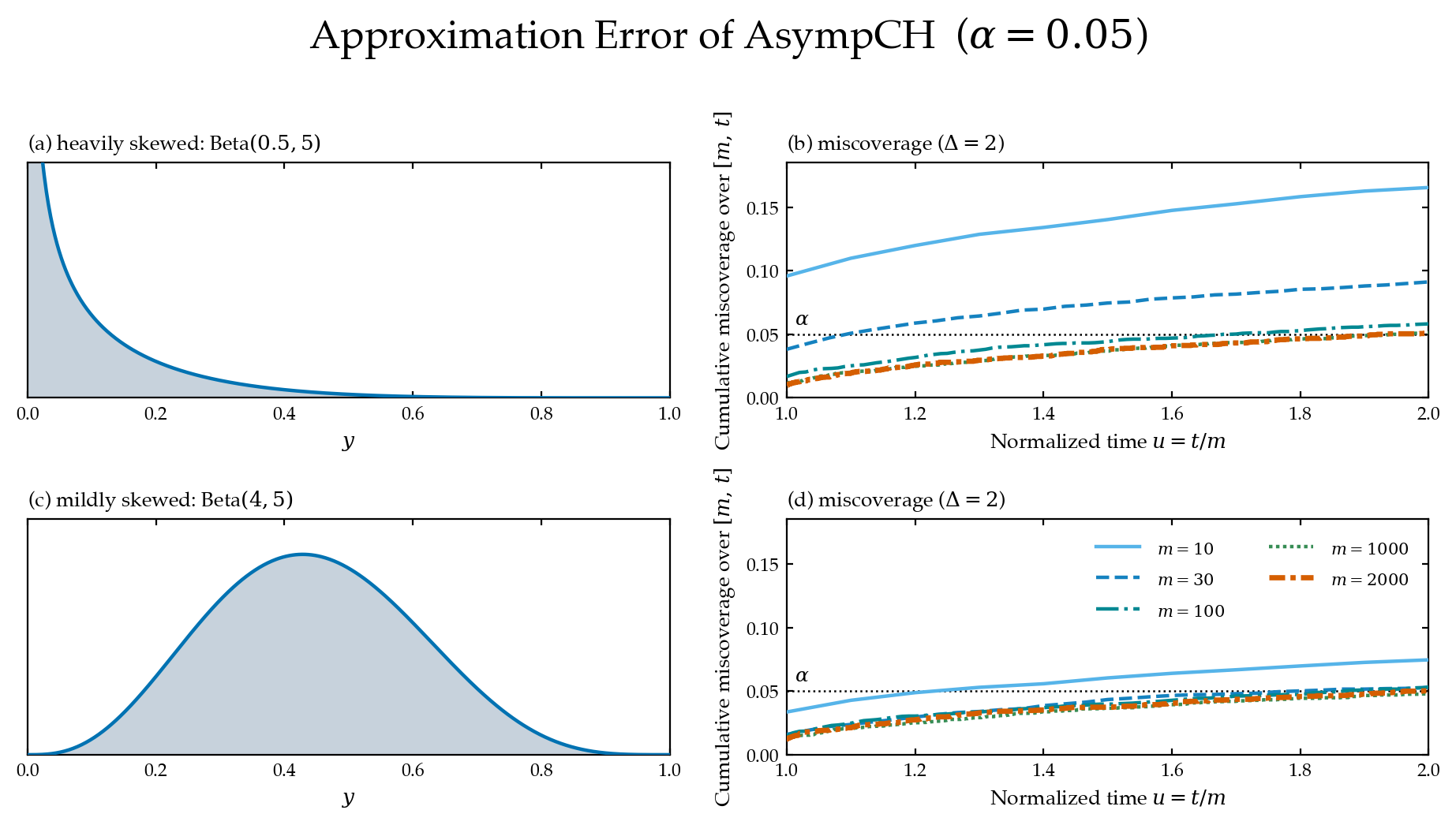}
    \caption{Approximation error for AsympCH as $m \to \infty$. For a skewed Beta distribution, 
    asymptotics kick in around $m = 100$. If the data are more symmetric, this approximation happens quite rapidly. 
    The phenomenon connects with how the third moment influences the pointwise Berry--Esseen bound \citep{berry_accuracy_1941, esseen_liapunoff_1942}. \figcode{figure14}}
    \label{fig:approximation-error}
\end{figure}

We now turn our attention to the approximation error as $m \to \infty$. The asymptotic confidence horizons are
asymptotic and therefore require $m$ to be ``large'' in order for coverage to hold. We consider $\infseqt{X} \simiid {\rm  Beta}(a,b)$,
where we vary the $a,b$ to induce skewness in the distribution. Similar to the pointwise CLT error, distributions with larger skewness
take longer for asymptotics to hold.

\section{Discussion}
In this work, we derived ``asymptotic confidence horizons'' which occupy a space within (and draw connections between) group sequential methods and confidence sequences. In many data analysis plans, data collection is limited from the outset. Using confidence horizons, this fact can be exploited to increase statistical power, reduce average stopping times, and ultimately save resources. 

We introduced a family of confidence horizons indexed by a parameter $q \in \R$ dictating the shape of the boundary in the confidence horizon. Researchers can apply confidence horizons for independent and martingale-dependent data. One such application of using confidence horizons for martingale-dependent data is under Neyman allocation. In this example, confidence horizons show substantial improvements over infinite horizon methods like AsympCSs. Appealing to strong approximations we give certain uniformity guarantees for these methods.

We show a connection between instances of asymptotic confidence horizons and Pocock's test \citep{pocock_group_1977} and O'Brien--Fleming's test \citep{obrien1979multiple}. In this way, confidence horizons can be viewed as similar in spirit to group sequential methods but with a maximal number of interim analyses. Using the techniques underlying confidence horizons, we show that group sequential methods enjoy certain uniformity guarantees and can be applied to independent and martingale-dependent data. We demonstrated a tradeoff between expected stopping time and power as a function of the number of interim peeks, solidifying confidence horizons as lying on the Pareto frontier of sequential inference methods, and as an important tool in the statistical toolbox.

\subsection*{Acknowledgments}
The authors thank Aur\'elien Bibaut, Danat Duisenbekov, Nick Koning, and Suhas Vijaykumar for helpful discussions. IW-S gratefully acknowledges support from the Miller Institute for Basic Research in Science. A.I. tools were used for coding, proofreading, and for technical details regarding the uniformity (in $q$) results. The prose was written entirely by hand by the authors.

\appendix

\bibliographystyle{abbrvnat}
\bibliography{references}
\newpage
\section{Proofs of the Main Results}
We use the stochastic $\bar o$ notation suggested in \citet{waudby-smith_distribution-uniform_2026} where 
\[
X_n = \bar o_{\cP}(a_n) \iff \forall \eps > 0, \lim_{m \to \infty}\sup_{\P \in \cP}\PM{\sup_{k \geq m} a_k\n\abs{X_k} > \eps} = 0,
\]and
\[
X_n = \bar O_{\cP}(a_n) \iff \exists C > 0 \st \lim_{m \to \infty}\sup_{\P \in \cP}\PM{\sup_{k \geq m} a_k\n\abs{X_k} > C} = 0.
\]When $\cP = \P$ a singleton we drop the $\sup_{\P \in \cP}$. The $\bar o$ notation consumes the algebraic identities the regular $O$ notation gives. For instance, $o(f(t))O(g(t)) = o(f(t)g(t))$ \citep[Proposition 2.4]{waudby-smith_distribution-uniform_2026}.

We use a shorthand where if $t \not \in \Z$ and $f(t)$ is only defined on the integers we implicitly define $f(t) = f(\floor t)$. For instance, 
\[
\sup_{n \in [10, 12]} S_n = \sup_{n \in [10, 12]} S_{\floor n} = \max_{n \in \{10,11, 12\}}S_n.
\]We do this because we move from discrete to continuous by flooring. It can be shown that linear interpolation or constructing Brownian bridges are also valid.

Unless otherwise stated, $W(t)$ is a standard Wiener process with $W(0) = 0$.

\subsection{Facts about the One- and Two-Sided Distributions}
We first define more general versions of the random variables that give rise to confidence horizons and group sequential methods.
\begin{definition}\label{def:general-gauss-sup}
Let $\cT \subset (0,\infty)$ be an ordered set. Define the one- and two-sided random variables $\zeta_+, \zeta$ to be
\begin{equation}
\zeta_+(\cT, q) = \sup_{s \in \cT}W(s)s^{q-1} \quad \text{ and } \quad \zeta(\cT,q) = \sup_{s \in \cT}\abs{W(s)}s^{q-1},
\end{equation}with distribution functions
\begin{equation}
F_+(x; \cT, q) = \PM{\zeta_+(\cT, q) \leq x}\quad \text{ and } \quad F(x; \cT, q) = \PM{\zeta(\cT, q) \leq x}.
\end{equation}
\end{definition}
\begin{lemma}\label{lemma:scale-inversion-general-sup}
Let $\lambda >0$. We have that
\[
\zeta(\lambda \cT, q) \overset{d}{=}\lambda^{q-1/2}\zeta(\cT, q) \quad \text{ and }\quad \zeta(\cT\n , 1-q) \overset{d}{=}\zeta(\cT, q),
\]where $\cT \n = \curly{1/s: s \in \cT}$. The analogous statements hold for the one-sided distribution.
\end{lemma}
\begin{proof}[Proof of \cref{lemma:scale-inversion-general-sup}]
We have the following properties of Wiener processes: 
\[
W(\lambda s) \overset{d}{=}\sqrt{\lambda}W(s) \quad \text{ and }\quad W(s) \overset{d}{=}s W(s\n). 
\]Applying these facts to $\zeta$ gives these identities.
\end{proof}
With this lemma, we can state the following as a consequence.
\begin{lemma}\label{lemma: calc-quantile}
Let $\cT = [1, \Delta]$ for some $\Delta \geq 1$. Then
\begin{equation}\label{eq:wt-correct-coverage}
\PM{\sup_{t \in m \cT} \abs{\frac{W(t)}{t}}\frac{t^q}{m^{q-1/2}} \leq x} = \Psi  \quad \text{ and } \quad  \PM{\sup_{t \in m\cT} \frac{W(t)}{t}\frac{t^q}{m^{q-1/2}} \leq x} = \Psi_{+}.
\end{equation}Letting $\Delta = n_K/n_1$ and $\cT^{\rm WT} = \curly{1, n_2/n_1 \dots, \Delta}$ gives
\begin{equation}\label{eq:xi-equivalent-wt-cutoff}
c_{\rm WT}(\alpha, \cT^{\rm WT}, q) \Delta^{q-1/2} = \Xi \n \paren{1-\alpha; \cT^{\rm WT}, q}.
\end{equation}
\end{lemma}

\begin{proof}[Proof of \cref{lemma: calc-quantile}]
The result in \cref{lemma:scale-inversion-general-sup} implies
\[
\zeta(m [1, \Delta], q) \overset{d}{=}\zeta([1, \Delta],q) m^{q-1/2},
\]which gives \cref{eq:wt-correct-coverage}.
To show \cref{eq:xi-equivalent-wt-cutoff} it is sufficient to show
\[
\PM{\sup_{s \in \cT^{\rm WT}} \abs{W(s)}s^{q-1} \leq c_{\rm WT}(\alpha, \cT, q)\Delta^{q-1/2}} = 1-\alpha.
\]
The definition of $c_{\rm WT}$ is given by 
\begin{align}
\alpha &= \PM{\sup_{s \in \cT^{\rm WT}} \abs{W(s)}s^{q-1} \geq c_{\rm WT}(\alpha, \cT, q)\Delta^{q-1/2}}\\
&= 1-\Xi \paren{c_{\rm WT}\paren{\alpha, \cT^{\rm WT}, q}\Delta^{q-1/2}; \cT^{\rm WT},q}.
\end{align}
\end{proof}

Now we show that $F/F_+$ are uniformly continuous over $\Delta$ when $q \neq 1/2$. We will use a Cameron--Martin argument to show anti-concentration, which helps prove equicontinuity.

\begin{lemma}\label{lemma:cameron-martin-lemma}
Let $0 < a < b < \infty$ and consider any $\cT \subset[a,b]$. Then for all $\eta > 0$,
\begin{align}
\sup_{x \in \R}\PM{\zeta(\cT, q) \in (x, x+\eta]}\leq 2\eta C_q(a,b) \quad \text{ and }\quad 
\sup_{x \in \R}\PM{\zeta_+(\cT, q) \in (x, x+\eta]} \leq \eta C_q(a,b),
\end{align}where $C_q^2(a,b) = a^{1-2q}+(1-q)^2\int_a^b v^{-2q}\dV$.
\end{lemma}

\begin{proof}[Proof of \cref{lemma:cameron-martin-lemma}]
We begin with
\begin{align}
\PM{\zeta_+(\cT, q) \in (x, x+\eta]} &= \PM{\sup_{s \in \cT}W(s)s^{q-1} \leq x+\eta} - \PM{\sup_{s \in \cT}W(s)s^{q-1} \leq x} \\
&=\PM{\forall s \in \cT: \paren{W - g}(t)\leq s^{1-q}x} - \PM{\forall s\in \cT: W(t)\leq s^{1-q}x}
\end{align}for $g(s) = \eta s^{1-q}$. Consequently, for $\mathsf Q$ , the law of $W_g$, the difference in probabilities is 
\begin{align}
\mathsf Q\!\paren{\forall s \in\cT: W(s) \leq s^{1-q}x} - \PM{\forall s\in \cT: W(t) \leq s^{1-q}x},
\end{align}which can be bounded by ${\rm TV}(\mathsf Q, \P)$. Using \cref{lemma:tv-wt-gt} and $\norm[\cH]{g}$, the Cameron--Martin norm of $g$, we have the following upper bound.
\[
{\rm TV}\paren{\mathsf Q, \mathsf P} = 2 \Phi\paren{\norm[\cH]{g}/2} - 1 \leq \norm[\cH]{g}/\sqrt{2\pi} \leq \norm[\cH]{g}.
\]
Using the fact that $s \wedge t$ = $a + (s-a) \wedge (t-a)$, one can calculate using tools from \citep{van_der_vaart_reproducing_2008} that
\begin{align}
\norm[\cH]{g}^2 = g(a)^2/a + \int_{a}^{b}g'(v)^2 \dV 
= \eta^2 \paren{a^{1-2q} + (1-q)^2 \int_{a}^{b}v^{-2q}\dV}.
\end{align}Union bounding gives the extra factor of $2$ for the two-sided distribution.
\end{proof}

\begin{lemma}\label{lemma:valid-continuity}
The Cameron--Martin norm, $\norm[\cH]{g}$, is finite in the following cases:
\begin{center}
\begin{tabular}{l|cc}
 & $b<\infty$ & $b=\infty$ \\
\hline
$a>0$ & all $q$ & $q>1/2$ \\
$a=0$ & $q<1/2$ & never \\
\end{tabular}
\end{center}
\end{lemma}
\begin{proof}[Proof of \cref{lemma:valid-continuity}]
It is straightforward to check this by analyzing $\smash{\int_{a}^{b}}v^{-2q}\dV$. When $q = 1/2, \, \smash{\int_{a}^{b}}v^{-2q}\dV = \infty$ if $a = 0$ or $b = \infty$. Thus, when $q = 1/2$ the process is both too spread out and too concentrated; uniform continuity fails. 
\end{proof}

We now use these lemmas to show that when $q < 1/2$ the distribution functions are equicontinuous over  $\Delta \geq 1$.

\begin{lemma}\label{lemma:anti-concentration-small-q}
Let $q < 1/2$ and $\cT \subset [1,\Delta]$ where $\inf \cT = 1$ and $\sup \cT = \Delta$. Then $\forall \eps > 0$, there exists an $\eta \equiv \eta(\eps)$ such that if 
\(
\abs{x-y} \leq \eta,
\)then
\[
\abs{\PM{\zeta(\cT, q) \leq x} - \PM{\zeta(\cT, q) \leq y}} < \eps
\]
\end{lemma}

\begin{proof}[Proof of \cref{lemma:anti-concentration-small-q}]
Fix $\eps > 0$. By the LIL, there exists an $R \equiv R(\eps)$ such that 
\[
\PM{\sup_{t \geq R}\abs{W(t)}t^{q-1} > \eps/2} < \eps/2.
\]
Setting $\cT'\coloneqq \cT \cap [1, R)$ and $\cT'' \coloneqq \cT \cap [R, \infty)$ gives $\zeta(\cT, q) = \max \curly{A, B}$ where
\[
A = \sup_{s \in \cT'}\abs{W(s)}s^{q-1} \quad \text{ and }\quad  B = \sup_{s \in \cT''}\abs{W(s)}s^{q-1}.
\]Set $\eta = \min \curly{\eps/2, \eps/(4 C_q(1,R))}$.

\paragraph{Case 1:}If $x < \eps/2$ then $x+\eta \leq \eps$ by definition of $\eta$. Consequently, 
\begin{align}
\PM{\sup_{s \in \cT} \abs{W(s)}s^{q-1} \in (x, x+\eta]} &\leq \PM{\sup_{s \in \cT} \abs{W(s)}s^{q-1} \in [0, \eps)}\\
&\leq \PM{\abs{W(1)} < \eps} < \eps,
\end{align}
which follows since $\inf \cT = 1$ and $\PM{\sup_{s \in \cT} \abs{W(s)}s^{q-1} > \eps} > \PM{\abs{W(1)}1^{q-1} > \eps}$. The last inequality uses 
\[
\PM{\abs{W(1)} < x} = \int_{-x}^{x}\phi(z)\dZ \leq \int_{-x}^{x}\phi(0)\dZ = (2\pi)^{-1/2}(2x) < x.
\]

\paragraph{Case 2:} If $x \geq \eps/2$,
\begin{align}
\PM{\zeta(\cT, q) \in (x, x+\eta]} &\leq \PM{B > \eps/2} + \PM{A \in (x, x+\eta]}\\
&\leq \eps/2 + 2\eta\, C_q(1, R) \leq \eps.
\end{align}The first line uses the fact that if $B < \eps/2$ and $x \geq \eps/2$ then $\PM{\max \curly{A, B} \in (x, x + \eta]} = \PM{A \in (x, x + \eta]}$.
\end{proof}
Now, we prove the equicontinuity statement for the one-sided distribution. Like the one-sided Robbins--Siegmund distribution \citep{waudby-smith_distribution-uniform_2026}, we have continuity over a subset of the support. This restriction is benign since it enables the analyst to construct high $1-\alpha$ confidence horizons.

\begin{lemma}\label{lemma:anti-concentration-small-q-one-sided}
Let $q < 1/2$ and $\cT \subset [1,\Delta]$ where $\inf \cT = 1$ and $\sup \cT = \Delta$. Fix $\alpha_1 > 1/2$ and let $\varsigma = \Phi \n (\alpha_1)> 0$. Then $\forall \eps > 0$, there exists an $\eta \equiv \eta(\eps, \alpha_1)$ such that if $x,y > \varsigma$ and
\(
\abs{x-y} \leq \eta,
\) then
\[
\abs{\PM{\zeta_+(\cT, q) \leq x} - \PM{\zeta_+(\cT, q) \leq y}} < \eps
\]
\end{lemma}

\begin{proof}[Proof of \cref{lemma:anti-concentration-small-q-one-sided}]
Fix $\eps > 0$. By the LIL, there exists an $R \equiv R(\eps, \alpha_1)$ such that 
\[
\PM{\sup_{t \geq R}\abs{W(t)}t^{q-1} > \varsigma/2} < \eps/2.
\]
Set $\cT'\coloneqq \cT \cap [1, R)$ and $\cT'' \coloneqq \cT \cap [R, \infty)$. Then, $\zeta_+(\cT, q) = \max \curly{A, B}$ where
\[
A = \sup_{s \in \cT'}W(s)s^{q-1} \quad \text{ and }\quad  B = \sup_{s \in \cT''}W(s)s^{q-1}.
\]Since $\inf \cT = 1$ we have $\zeta_+(\cT, q) \geq W(1)$. This implies that  $ F_+\n(\alpha_1; \cT, q) \geq \Phi\n(\alpha_1) =\varsigma $. Thus, we focus on $x \geq \varsigma$ as the only case. Letting $x \geq \varsigma$,
\begin{align}
\PM{\zeta_+(\cT, q) \in (x, x+\eta]} &\leq \PM{B > \varsigma/2} + \PM{A \in (x, x+\eta]}\\
&\leq \eps/2 + \eta\, C_q(1, R) \leq \eps.
\end{align}The reasoning is the same as the proof of \cref{lemma:anti-concentration-small-q}. Setting $\eta \equiv \eta(\eps, \varsigma)= \eps/ (2C_q(1,R))$ gives the desired result.
\end{proof}

Lastly, we prove that a transformed random variable has uniform continuity and tightness properties if $q > 1/2$.
\begin{lemma}\label{lemma:tilde-f-facts}
Let 
\[
\Tilde \zeta(\cT, q) = 
    \zeta(\cT, q) \Delta^{1/2-q} \quad \text{ and }\quad \Tilde F(x; \cT, q) = \PM{\Tilde \zeta (\cT, q) \leq x}.
\]Then, 
\[
\Tilde \zeta(\cT, q) \overset{d}{=}\zeta \paren{\Delta\cT\n, 1-q} \quad \text{ and } \quad \Tilde F(x\Delta^{1/2-q}; \cT, q) = F(x; \cT, q).
\]If $q > 1/2$, $\Tilde F$ is uniformly tight and equicontinuous with respect to $\Delta$. The analogous statement holds for the one-sided distribution.
\end{lemma}

\begin{proof}[Proof of \cref{lemma:tilde-f-facts}]
The first equality is given by \cref{lemma:scale-inversion-general-sup}. This is because
\begin{align}
\zeta(\Delta \cT\n, 1-q) \overset{d}{=} \zeta(\Delta\n \cT, q) \overset{d}{=} \Delta^{1/2-q}\zeta(\cT, q).
\end{align}The second equality simply cancels out the $\Delta$ factors. To show continuity when $q > 1/2$. We note 
\begin{align}
\Tilde F(x; \cT, q) &= \PM{\zeta(\Delta \cT \n, 1-q) \leq x}\\
&= F(x; \Delta \cT\n, 1-q).
\end{align}
Since $q > 1/2$, $1-q <1/2$. Thus, we can invoke \cref{lemma:anti-concentration-small-q} on this set to show equicontinuity over $\Delta$. To show tightness, we note that 
\begin{align}
\Tilde \zeta(\cT, q) \overset{d}{=}\zeta(\Delta \cT\n, 1-q) \leq\zeta([1,\Delta], 1-q) \leq \zeta([1,\infty), 1-q).
\end{align}
Again, we can use the small $q$ results since $1-q < 1/2$. It follows from \cref{lemma:proper-limit} that this random variable is finite almost surely. For the one-sided distribution, we can use \cref{lemma:anti-concentration-small-q-one-sided}.
\end{proof}

We now summarize the properties of the distribution functions we are working with. 
\begin{lemma}\label{lemma:unif-continuity-psi}
Let $\Delta \geq 1, \cT \subset [1, \Delta]$  where $\inf \cT = 1$ and  $\sup \cT = \Delta$.

For $q < 1/2, F/F_+$ have the following properties
\begin{enumerate}[(i)]
    \item \emph{Equicontinuity:} $F(x; \cT, q)$ is equicontinuous and uniformly tight over $\Delta \geq 1$ for any $\cT \subset [1,\Delta]$. $F_+(x; \cT, q)$ is equicontinuous on $[a, \infty)$ for any $a > 0$ and uniformly tight for all $\Delta \geq 1$, and for any $\cT \subset [1,\Delta]$.
    \item \emph{Strictly increasing and invertible:} $F/F_+$ are strictly increasing and hence invertible.
\end{enumerate}
For $q > 1/2$, the distribution functions obey
\begin{enumerate}[(i)]
    \item \emph{Equicontinuity:} $\Tilde F(x; \cT, q)$ is equicontinuous and uniformly tight over $\Delta \geq 1$ for any $\cT \subset [1,\Delta]$. $\Tilde F_+(x; \cT, q)$ is equicontinuous on $[a, \infty)$ for any $a > 0$ and uniformly tight for all $\Delta \geq 1$, and for any $\cT \subset [1,\Delta]$.
    \item \emph{Strictly increasing and invertible:} $\Tilde F/\Tilde F_+$ are strictly increasing and hence invertible.
\end{enumerate}
For $q = 1/2$, the distribution functions have the following properties for any fixed $\Deltabar < \infty$:
\begin{enumerate}[(i)]
    \item \emph{Equicontinuity:} $ F(x; \cT, q)$ and $ F_+(x; \cT, q)$ are equicontinuous over $\Delta \in [1, \Deltabar]$ for any $\cT \subset [1,\Delta]$.
    \item \emph{Strictly increasing and invertible:} $F/ F_+$ are strictly increasing and hence invertible.
\end{enumerate}
\end{lemma}
\begin{proof}[Proof of \cref{lemma:unif-continuity-psi}]
The properties listed when $q < 1/2$ are shown in \cref{lemma:anti-concentration-small-q} and \cref{lemma:proper-limit}. The properties of $\Tilde F$ shown when $q > 1/2$ are proved in \cref{lemma:tilde-f-facts}. For $q = 1/2$, the distribution function is not equicontinuous. However, if we fix $\Deltabar < \infty$, we can invoke \cref{lemma:valid-continuity}, with $a > 0, b < \infty$, which holds for any $q \in \R$.
\end{proof}

Now we state an extension of a Pólya continuity lemma that allows for more general uniform convergence.
\begin{lemma}\label{lemma:dist-uniform-polya-continuity}
Let $F_n(\theta)$ be a sequence of distribution functions converging to $F$ pointwise in $x$ but uniform over $\P \in \cP$ and $\theta \in \Theta$ for $\Theta$ some generic parameter set, meaning 
\begin{equation}\label{eq:pointwise-limit-assumption}
\forall x\in \R, \lim_{n \to \infty}\sup_{\theta \in \Theta}\sup_{\P \in \cP}\abs{F_n(x; \theta) - F(x; \theta)} = 0.
\end{equation}
If $F(\theta)$ is equicontinuous and uniformly tight on $[a,\infty)$, for $a \in [-\infty, \infty)$, then
\[
\lim_{n \to \infty}\sup_{\theta \in \Theta}\sup_{x \in [a,\infty)}\sup_{\P \in \cP}\abs{F_n(x;\theta) - F(x; \theta)} = 0.
\]
\end{lemma}

\begin{proof}[Proof of \cref{lemma:dist-uniform-polya-continuity}]
First assume $a = -\infty$. Fix $\eps > 0$. There exists an $R \equiv R(\eps)$ such that 
\[
F(-R; \theta) - F(-\infty, \theta) \leq \eps\quad \text{ and }\quad F(R; \theta) - F(\infty; \theta)> -\eps,
\]where we take $F/F_n(\infty; \theta) = 1$ and $F/F_n(-\infty; \theta) = 0$. By equicontinuity, there exists an $\eta \equiv \eta(\eps)$ such that if $\abs{x-y} < \eta$ then $\abs{F(x; \theta) - F(y; \theta)} < \eps$ for all $\theta$. Set $-\infty = x_0, x_1 = -R,  \dots, x_{N-1} = R, x_N = \infty$ where the mesh points in between of distance at most $\eta$. Then taking $x \in [x_{i-1}, x_i]$, we have
\begin{align}
F_n(x; \theta) - F(x; \theta) &\leq F_n(x_i; \theta) - F(x_{i-1}; \theta) \leq F_n(x_{i}; \theta) - F(x_i; \theta) + \eps\\
F_n(x; \theta) - F(x; \theta) &\geq F_n(x_{i-1}; \theta) - F(x_i; \theta) \geq F_n(x_{i-1}; \theta) - F(x_{i-1}; \theta) - \eps.
\end{align}Thus for every $x \in \R$, 
\begin{align}
\sup_{\theta \in \Theta}\sup_{\P \in \cP}\abs{F_n(x; \theta) - F(x; \theta)} &\leq \sup_{\theta \in \Theta}\sup_{i} \abs{F_{n}(x_i; \theta) - F(x_i; \theta)} + \eps
\end{align}We note that the first term can be made arbitrarily small if $n$ is sufficiently large by the assumption in \cref{eq:pointwise-limit-assumption}. 

If $F_\theta$ is equicontinuous and uniformly tight only on $[a,\infty) \subset \R$, then the existing proof works with $x_0 = a$. This remark is dedicated to show that this lemma applies to $F_+$, the one-sided distribution.
\end{proof}

\subsection{Strong Approximation Lemmas}
In this section, we provide lemmas that support the ability to swap an empirical process with a Wiener process. We begin by giving the ability to swap partial sums with Gaussian processes and end by showing the ability to swap Gaussian processes with Wiener processes. While some of these lemmas are standard, \cref{lemma:strong-approx-cumul} and \cref{lemma:max-discrete-cont-gap} have not been extensively discussed in the literature and may be interesting to examine and apply to other problems.

We first state a lemma that allows strong approximations after dividing by the cumulative parameter $\Tilde \sigma^2_{\P, t}$ instead of each individual parameter $\sigma^2_{\P, i}$. 

\begin{lemma}\label{lemma:strong-approx-cumul}
Let $\infseqt{X}$ be a sequence of random variables on $\pspace[\cP]$ that, if independent satisfy \cref{condition:no-deterministic}, and otherwise satisfy \cref{condition:non-vanishing-variances} at a rate $\rho$. Moreover, for $\infseqt{Y} \simiid \cN(0,1)$ and some $a$, suppose
\begin{equation}\label{eq:strong-approximation-rate-a}
\abs{\tsum{(X_i - \mu_{\P, i})} - \tsum{Y_i \sigma_{\P, i}}} = \bar o_{\cP}\paren{t^a \log t}.
\end{equation}After (potentially) enriching $\pspace[\cP]$ to include $\infseqt{\Tilde Y}$ and defining $G_t \coloneqq \sum_{i=1}^{t} \Tilde Y_i$ that although equal to $\tsum{Y_i}$ in distribution are not the same random process, the following holds for $\gamma = \max \curly{a, (1-\rho)/2}$:
\[
\abs{\frac{\tsum{(X_i - \mu_{\P, i})}}{\Tilde \sigma_{\P, t}} - G_t} = \bar o_{\cP}\paren{t^\gamma \log t}.
\]
\end{lemma}

\begin{proof}[Proof of \cref{lemma:strong-approx-cumul}]
Call $\tsum{X_i - \mu_{\P, i}} = S_{\P, t}^\circ$. We continue to assume that $t \in \Z$ and only work with the discrete time Gaussian process, $G_t = W(\sigma^2_{\P}t)/\sigma_\P$. Consequently,
\begin{align}
\frac{S_{\P, t}^\circ}{\Tilde \sigma_{\P, t}} - G_t &= \frac{S_{\P, t }^\circ}{\Tilde \sigma_{\P, t}} - \frac{W(\sigma^2_\P t)}{\sigma_\P}\\
\verbose{
&= \frac{S_{\P, t}^\circ}{\Tilde \sigma_{\P, t}} - \frac{W(V_t)}{\Tilde \sigma_{\P, t}} + \frac{W(V_t)}{\Tilde \sigma_{\P, t}} - \frac{W(\sigma^2_\P t)}{\sigma_\P}\\
&= \frac{S_{\P, t}^\circ}{\Tilde \sigma_{\P, t}} - \frac{W(V_t)}{\Tilde \sigma_{\P, t}} + \frac{W(V_t)}{\Tilde \sigma_{\P, t}} - \frac{W(\sigma^2_\P t)}{\Tilde \sigma_{\P, t}} - \frac{W(\sigma^2_\P t)}{\sigma_\P} + \frac{W(\sigma^2_\P t)}{\Tilde \sigma_{\P, t}}\\
}
&= \frac{S_{\P, t}^\circ}{\Tilde \sigma_{\P, t}} - \frac{W(V_t)}{\Tilde \sigma_{\P, t}} + \frac{W(V_t)}{\Tilde \sigma_{\P, t}} - \frac{W(\sigma^2_\P t)}{\Tilde \sigma_{\P, t}} + \sigma_\P W(t) \paren{\frac{1}{\Tilde \sigma_{\P, t}} - \frac{1}{\sigma_\P}}\\
&=  \underbrace{\frac{S_{\P, t}^\circ}{\Tilde \sigma_{\P, t}} - \frac{W(V_t)}{\Tilde \sigma_{\P, t}}}_{\mathrm I} + \underbrace{\frac{W(V_t)}{\Tilde \sigma_{\P, t}} - \frac{W(\sigma^2_\P t)}{\Tilde \sigma_{\P, t}}}_{\mathrm{II}} + \underbrace{W(t) \paren{\frac{\sigma_\P}{\Tilde \sigma_{\P, t}} - 1}}_{\mathrm{III}}.
\end{align}
By assumption, 
\[
\abs{\mathrm{I}} = \bar o_{\cP}\paren{t^a \log t}
\]and \cref{lemma:wiener-increment-bound} yields
\[
\abs{\mathrm{II}} = \bar O_{\cP}\paren{\sqrt{t^{1-\rho}\log t }}= \bar o_{\cP}\paren{t^{(1-\rho)/2} \log t}.
\]The LIL and \cref{condition:non-vanishing-variances} with rate $\rho$ gives
\[
\abs{\mathrm{III}} = \bar O_{\cP}\paren{t^{1/2 - \rho}\sqrt{\log \log t}},
\]which is dominated by $\mathrm{II}$ since $\rho > \rho/2$ for $\rho > 0 $. Summing the terms together gives
\[
\abs{\frac{S_{\P,t}^\circ}{\Tilde \sigma_{\P, t}} - G_t} = \bar o_{\cP}\paren{t^\gamma \log t},
\]where $\gamma = \max\{a, (1-\rho)/2\}$.
\end{proof}
The following lemma gives the ability for a coupling with a discrete time Gaussian process to be upgraded to a coupling with a continuous time Wiener process. 

\begin{lemma}\label{lemma:wiener-strong-approx}
Suppose $\infseqt{X}$ is a sequence of random variables defined on $\pspace[\cP]$ rich enough to contain $\infseqt{Y} \simiid \cN(0,1)$ and define $G_t \equiv \sum_{i=1}^{t}Y_i$. Then if for any $a > 0$,
\[
\abs{\frac{\tsum{(X_i - \mu_{\P, i})}}{t \Tilde \sigma_{\P, t}} - \frac{G_t}{t}} = \bar o_{\cP} \paren{t^{a-1} \log t},
\]the following stronger guarantee holds for all $t \in \R_+$:
\begin{equation}\label{eq:wiener-strong-approximation}
\abs{\frac{\tsum{(X_i - \mu_{\P, i})}}{t \Tilde \sigma_{\P, t}} - \frac{W(t)}{t}} = \bar o_{\cP} \paren{t^{a-1} \log t}.
\end{equation}
\end{lemma}
\begin{proof}[Proof of \cref{lemma:wiener-strong-approx}]
The triangle inequality states
\begin{align}
\abs{\frac{\tsum{X_i - \mu_{\P, i}}}{t \Tilde \sigma_{\P, t}} - \frac{W(t)}{t}} & \leq \underbrace{\abs{\frac{\tsum{X_i - \mu_{\P, i}}}{t \Tilde \sigma_{\P, t}} - \frac{G_t}{t}}}_{\mathrm I} + \underbrace{\abs{\frac{W(t)}{t} - \frac{G_t}{t}}}_{ \mathrm {II}}\\
&= \bar o_{\cP}\paren{t^{a-1} \log t} + \bar O_{\cP} \paren{\frac{\sqrt{\log t}}{t}} = \bar o_{\cP}\paren{t^{a-1} \log t}.
\end{align}We bound $\mathrm I$ by assumption and $\mathrm{II}$ with \cref{lemma:max-discrete-cont-gap}.
\end{proof}
Now, we state the following strong approximation required for independent data \citep{komlos_approximation_1976,waudby-smith_nonasymptotic_2025}.

\begin{lemma}\label{lemma:dist-uniform-kmt}
Let $\infseqt{X}$ be a sequence of independent random variables on $\pspace[\cP]$ satisfying \cref{condition:ui-condition} with parameter $\kappa, \delta$ and \cref{condition:no-deterministic} with parameter $\rho$. Let $\cF_n = \sigma\paren{X_1, \dots, X_n}$. After potentially enriching the probability space so that it  contains a Wiener $W(t)$ it follows
\begin{equation}\label{eq:dist-unif-strong-approximation}
\abs{\frac{\tsum{( X_i - \mu_{\P, i})}}{\Tilde \sigma_{\P, t} t} - \frac{W(t)}{t}} = \bar o_{ \cP}\paren{t^{\gamma} \log t},
\end{equation}
for $\gamma = \max\curly{(1+\delta)/\kappa, (1-\rho)/2}-1$. Since $\delta <(\kappa - 2)/2$ and $\rho > 0$, $\gamma <-1/2$.
\end{lemma}
\begin{proof}[Proof of \cref{lemma:dist-uniform-kmt}]
\citep[Corollary 3.4]{waudby-smith_nonasymptotic_2025} implies for $\infseqt{Y} \simiid \cN(0,1)$,
\[
\abs{\tsum{(X_i - \mu_{\P, i})} - \tsum{Y_i \sigma_{\P, i}}} = \bar o_{ \cP}\paren{t^{(1+\delta)/\kappa}}.
\]Using \cref{lemma:strong-approx-cumul} we can divide by the cumulative parameter to obtain
\[
\abs{\frac{\tsum{( X_i - \mu_{\P, i})}}{\Tilde \sigma_{\P, t}} - G_t} = \bar o_{\cP}\paren{t^{\gamma+1} \log t}.
\]Lastly, invoking \cref{lemma:wiener-strong-approx} and dividing through $t$ gives the desired statement.
\end{proof}

We now consider strong approximation under martingale-dependence adapted from  \citep[Lemma A.2]{waudby-smith_time-uniform_2024} and \citep{strassen_almost_1967}.
\begin{lemma}\label{lemma:strong-approximation-mgale-dependence}
Let $\infseqt{X}$ be a stream of random variables defined on a single probability space $\pspace$ that satisfy \cref{condition:non-vanishing-variances} with parameter $\rho$ and \cref{condition:lindeberg-type-mgale} with rate $\kappa$. After potentially enriching the probability space so that there exists a Wiener $W(t)$, we have the property that 
\begin{equation}\label{eq:mgale-strong-approximation}
\abs{\frac{\tsum{\paren{X_i - \mu_{\P, i}}}}{\Tilde \sigma_{\P, t} t} - \frac{W(t)}{t}} = \bar o_{ \P} \paren{t^{\gamma} \log t},
\end{equation}for $\gamma = \max\curly{1/4 + 1/(2\kappa), (1-\rho)/2} -1$. Since $\kappa > 2$ and $\rho > 0$, $\gamma < -1/2$.
\end{lemma}
\begin{proof}[Proof of \cref{lemma:strong-approximation-mgale-dependence}]
The proof relies on \citep[Theorem 4.4]{strassen_almost_1967}
which states on $ \P$
\[
\tsum{\paren{X_i - \mu_{\P, i}}} = W \paren{V_{\P, t}} + \bar o_{ \P}\paren{h(V_{\P, t})}
\]where $h(v) = \paren{vf(v)}^{1/4}\log v$ is defined for any increasing $f$ yet decreasing $f(v)/v$ that satisfies \citep[Eq 138]{strassen_almost_1967}:
\begin{align}
\sum_{n=1}^{\infty}f(V_n)\n \EE{\paren{X_n - \mu_{\P, n}}^2\indic{(X_n - \mu_{\P, n})^2 > f(V_n)} \mid \cF_{n-1}} < \infty.
\end{align}Letting $f(v) = v^{2/\kappa}$ satisfies the condition by assumption (\cref{condition:lindeberg-type-mgale}). Thus we let $h(v) = v^{1/4 + 1/(2\kappa)}\log v$. We can summarize this by claiming:
\begin{align}
\abs{\tsum{\paren{X_i - \mu_{\P, i}}} - W(V_{\P, t})} &= \bar o_{ \P} \paren{V_{\P, t}^{1/4 + 1/(2\kappa)} \log V_{\P, t}}.
\end{align} Because we assume that $V_{\P,t}/t \to \sigma^2_\P$, it follows $V_{\P, t} \asymp t$. Hence replacing $V_{\P,t}$ with $t$ gives
\begin{align}
\abs{\tsum{\paren{X_i - \mu_{\P, i}}} - W(V_{\P, t})} &= \bar o_{ \P} \paren{t^{1/4 + 1/(2\kappa)} \log t}.
\end{align}
Applying \cref{lemma:strong-approx-cumul} and dividing by $t$ gives 
\begin{align}
    \abs{\frac{\tsum{\paren{X_i - \mu_{\P, i}}}}{ t \Tilde \sigma_{\P, t}} - \frac{W(V_{\P, t})}{t}} &= \bar o_{\P}\paren{t^{\gamma}\log t},
\end{align}for $\gamma = \max \curly{1/4 + 1/(2\kappa), (1-\rho)/2}-1$. Applying \cref{lemma:wiener-strong-approx} gives
 \cref{eq:mgale-strong-approximation}.
\end{proof}
\subsection{Proofs of Main Theorems: \cref{theorem:general-mean} and \cref{theorem:mgale-asympcv}}
This section gives a general lemma that when applied to independent or martingale-dependent data proves the main theorems in \cref{sec:power-q}. Moreover, it applies to the group sequential methods in \cref{sec:theoretical-conseq-gs}. It may be useful for other asymptotic sequential statistics. First, we summarize the conditions.
\begin{condition}\label{cond:general-lemma}
Let $\infseqt{X}$ be an infinite sequence of random variables defined on $\pspace[\cP]$ sufficiently rich so that $\cP$ contains a Wiener process $W(t)$. Consider the following conditions:

\begin{enumerate}
\item[\textbf{C1}]\emph{Variance Estimator Consistency.}
\[
\abs{\frac{\Tilde \sigma^2_{\P, t}}{\hat \sigma_t^2} - 1} = \bar o_{\cP}\paren{\frac{1}{\log t}}.
\]\item[\textbf{C2}] \emph{Additive Gaussian Coupling at a polynomial rate.} For a $\gamma < -1/2$ 
\[
\abs{\frac{S_t - t\Tilde \mu_{\P, t}}{\Tilde \sigma_{\P,t} t} - \frac{W(t)}{t}} = \bar o_{\cP}\paren{t^{\gamma}\log t}.
    \]
\end{enumerate}
\end{condition}

\begin{lemma}\label{lem:general-lemma}
Let $\infseqt{X}$ be a sequence of random variables defined on $\pspace[\cP]$ that satisfy \cref{cond:general-lemma} (\emph{C1} and \emph{C2}). Let $\mathfrak{T}_\Delta \coloneqq \curly{\cT \subseteq [1,\Delta] : \inf \cT = 1,\ \sup \cT = \Delta}$ and $\alpha_1 > 1/2$. For all $q \neq 1/2$, we have
\begin{align}
&\lim_{m \to \infty}\sup_{\P \in \cP}\sup_{\alpha \in (0, 1)} \sup_{\Delta \geq 1} \sup_{\cT \in \mathfrak{T}_\Delta}\abs{\PM{\sup_{t \in m\cT}\abs{\frac{S_t - t\Tilde \mu_{\P, t}}{t \hat \sigma_t}}\frac{t^{q}}{m^{q-1/2}} \leq F \n(\alpha; \cT, q)} - \alpha} = 0 \text{ and } \label{eq:two-sided-general-prop-unif-delta}\\
&\lim_{m \to \infty}\sup_{\P \in \cP}\sup_{\alpha \in [\alpha_1, 1)} \sup_{\Delta \geq 1}\sup_{\cT \in \mathfrak{T}_\Delta}\abs{\PM{\sup_{t \in m\cT}\frac{S_t - t\Tilde \mu_{\P, t}}{t \hat \sigma_t}\frac{t^{q}}{m^{q-1/2}} \leq F_+ \n(\alpha; \cT, q)} - \alpha} = 0
  \end{align}

If $q =1/2$, the approximation is $\cP$-weakly uniform, meaning for $\Deltabar < \infty$,
\begin{align}
  &\lim_{m \to \infty}\sup_{\P \in \cP}\sup_{\Delta \in [1, \Deltabar)} \sup_{\cT \in \mathfrak{T}_\Delta}\sup_{\alpha \in (0, 1)}\abs{\PM{\sup_{t \in m\cT}\abs{\frac{S_t - t\Tilde \mu_{\P, t}}{t \hat \sigma_t}}\frac{t^{q}}{m^{q-1/2}} \leq F \n(\alpha; \cT, q)} - \alpha} = 0 \text{ and }\label{eq:two-sided-general-prop}\\
  &\lim_{m \to \infty}\sup_{\P \in \cP}\sup_{\Delta \in [1, \Deltabar)}\sup_{\cT \in \mathfrak{T}_\Delta}\sup_{\alpha \in [\alpha_1, 1)}\abs{\PM{\sup_{t \in m\cT}\frac{S_t - t\Tilde \mu_{\P, t}}{t \hat \sigma_t}\frac{t^{q}}{m^{q-1/2}} \leq F_+\n(\alpha; \cT, q)} - \alpha} = 0. \label{eq:one-sided-general-prop}
  \end{align}
The analogous statements hold when $S_t - t \Tilde \mu_{\P, t}$ is replaced with $t\Tilde \mu_{\P, t} - S_t$.
\end{lemma}
\begin{proof}[Proof of \cref{lem:general-lemma}]
The proof is similar to the proof of \citep[Theorem 2.10]{waudby-smith_distribution-uniform_2026}. We first show the two-sided guarantee in \cref{eq:two-sided-general-prop}. 

In Step 1, we will control the error induced from variance estimation. In Step 2, we
control the error induced from the strong Gaussian approximation. In Step 3, we use the uniform continuity of 
$F$ from \cref{lemma:unif-continuity-psi} and in Step 4 we apply \cref{lemma:dist-uniform-polya-continuity} to upgrade the pointwise statement to the uniform statement. In Step 5, we will return to the case when $q = 1/2$ to show weak-uniformity.

Let $\eps > 0$ and $q \neq 1/2$. By the continuity result in \cref{lemma:unif-continuity-psi}, there exists a single $\eta_0\equiv \eta_0(\eps) > 0$ such that if
\(
\abs{y-z} < \eta_0\) then \(\abs{F(y; \cT, q) -  F(z; \cT, q)} < \eps
\) uniformly over $\cT \subset \mathfrak{T}_\Delta$. Let $\eta = \min \curly{1, \eta_0}$. Let $m > 2$, so that $m (\log m)^a$ is increasing for any power $a$.
\paragraph{Step 1: Controlling errors from variance estimation}

It follows from assumption \emph{C1} that there exists some $m_1 \equiv m_1(\eps)$ sufficiently large so that $\forall m > m_1$ we have 
\[
\sup_{\P \in \cP} \PM{\sup_{t \geq m} \abs{\frac{\Tilde \sigma_t}{\hat \sigma_t}-1} \geq \frac{\eta}{\log t}} < \eps.
\] This means that for every $\P \in \cP$ we have with probability $1-\eps$ under $\P$ that
\begin{align}
\paren{1-\frac{\eta}{\log t}}\abs{\frac{S_t - t\Tilde \mu_{\P, t}}{\Tilde \sigma_{\P, t} t}}&\leq \abs{\frac{S_t - t\Tilde \mu_{\P, t}}{\hat \sigma_t t}} \leq \abs{\frac{S_t - t\Tilde \mu_{\P, t}}{\Tilde \sigma_{\P, t} t}}\paren{1+\frac{\eta}{\log t}}
\end{align} Turning to the expression of interest, we have for any $m \geq m_1$ that for all $\P \in \cP$ and $x \in \R$,
\begin{align}
\PM{\sup_{t \in m\cT} \abs{\frac{S_t - t\Tilde \mu_{\P, t}}{\hat \sigma_t t}} \frac{t^q}{m^{q-1/2}} \leq x} &< \PM{\sup_{ t\in m\cT} \abs{\frac{S_t - t\Tilde \mu_{\P, t}}{\Tilde \sigma_{\P, t} t}} \frac{t^q}{m^{q-1/2}} \leq \paren{1 - \eta/\log m}\n x} + \eps,\\
\PM{\sup_{t \in m\cT} \abs{\frac{S_t - t\Tilde \mu_{\P, t}}{\hat \sigma_t t}} \frac{t^q}{m^{q-1/2}} \leq x} &> \PM{\sup_{t \in m\cT} \abs{\frac{S_t - t\Tilde \mu_{\P, t}}{\Tilde \sigma_{\P, t} t}} \frac{t^q}{m^{q-1/2}} \leq \paren{1 + \eta/\log m}\n x} - \eps.
\end{align}

\paragraph{Step 2: Controlling the error induced by strong Gaussian approximation}
By \emph{C2}, there exists some $m_2 \equiv m_2(\eps)$ so that $\forall m \geq m_2$ we have for all $\P$ in $\cP$ that 
\[
\PM{\sup_{t \geq m} \abs{\frac{S_t - t\Tilde \mu_{\P, t}}{\Tilde \sigma_{\P, t} t} -\frac{W(t)}{t}} \geq \eta t^{\gamma}\log t} < \eps.
\]We split into two cases, $q+\gamma \geq 0$ and $q + \gamma < 0$. First, if $q + \gamma \geq 0$, then for all $m \geq m_2$ we have that for all $\P \in \cP$ and  $x \in \R$
\begin{align}
& \PM{\sup_{t \in m\cT}\curly{ \abs{\frac{S_t - t\Tilde \mu_{\P, t}}{\Tilde \sigma_{\P, t} t}}\frac{t^{q}}{m^{q-1/2}}} \leq x}\\
&< \PM{\sup_{t \in m\cT} \curly{\abs{\frac{W(t)}{t}}\frac{t^q}{m^{q-1/2}}-\eta \frac{t^{q + \gamma} \log(t)}{m^{q-1/2}}} \leq x } + \eps\\
&= \PM{\Delta^{1/2-q}\sup_{t \in m\cT} \curly{ \abs{\frac{W(t)}{t}}\frac{t^q}{m^{q-1/2}}-\eta  \frac{t^{q + \gamma} \log(t)}{m^{q-1/2}}} \leq x\Delta^{1/2-q} } + \eps\\
&\leq \PM{\Delta^{1/2-q}\sup_{t \in m\cT} \abs{\frac{W(t)}{t}}\frac{t^q}{m^{q-1/2}}-\eta  \frac{(m\Delta)^{q + \gamma} \log(m\Delta)}{(m\Delta)^{q-1/2}} \leq x\Delta^{1/2-q} } + \eps \label{eq:maximizer-right-point}\\
&\leq \PM{\Delta^{1/2-q}\sup_{t \in m\cT} \abs{\frac{W(t)}{t}}\frac{t^q}{m^{q-1/2}}-\eta  m^{1/2 + \gamma} \log m \leq x\Delta^{1/2-q} } + \eps \label{eq:m-delta-drop-out}\\
&= \Tilde F \paren{\Delta^{1/2-q}x + \eta m^{\gamma + 1/2}\log m; \cT, q} + \eps
\end{align}
where \cref{eq:maximizer-right-point} uses that the maximizer occurs at $\sup \cT$ if $q + \gamma \geq 0$ and \cref{eq:m-delta-drop-out} uses the fact that if $a < 0$ then $(m\Delta)^a\log(m\Delta) \leq m^a\log(m)$ (for large $m$ and $\Delta \geq 1$). If $q < -\gamma$, then 
\begin{align}
\PM{\sup_{t \in m\cT} \abs{\frac{S_t - t\Tilde \mu_{\P, t}}{\Tilde \sigma_{\P, t} t}}\frac{t^{q}}{m^{q-1/2}} \leq x}
&< \PM{\sup_{t \in m\cT} \curly{ \abs{\frac{W(t)}{t}}\frac{t^q}{m^{q-1/2}}-\eta  \frac{t^{q + \gamma} \log(t)}{m^{q-1/2}}} \leq x} + \eps\\
&\leq \PM{\sup_{t \in m\cT} \abs{\frac{W(t)}{t}}\frac{t^q}{m^{q-1/2}}-\eta m^{1/2+ \gamma} \log(m) \leq x} + \eps \label{eq:maximizer-left-point}\\
&= F \paren{x + \eta m^{\gamma + 1/2}\log(m); \cT, q} + \eps,
\end{align}where \cref{eq:maximizer-left-point} uses that the maximizer occurs at $\inf \cT$ if $q + \gamma < 0$. Similarly, the lower bound is
\begin{align}
\PM{\sup_{t \in m\cT} \abs{\frac{S_t - t\Tilde \mu_{\P, t}}{\Tilde \sigma_{\P, t} t}}\frac{t^{q}}{m^{q-1/2}} \leq x}  &>\begin{cases}
    \Tilde F \paren{\Delta^{1/2-q}x - \eta m^{\gamma + 1/2}\log(m); \cT, q} - \eps & q \geq -\gamma\\
     F \paren{x - \eta m^{\gamma + 1/2}\log(m); \cT, q} -\eps & q < -\gamma
\end{cases}
\end{align}
It then follows after considering Step 1 that for all $m \geq \max \{m_1, m_2\}$ we have that for all $\P \in \cP$ and $x \in \R$, that 
\begin{align}
\PM{\sup_{t \in m\cT} \abs{\frac{S_t - t\Tilde \mu_{\P, t}}{\hat \sigma_t t}} \frac{t^q}{m^{q-1/2}} \leq x}& <\begin{cases}
    \Tilde F \paren{\Delta^{1/2-q} x(1-\eta/\log m)\n + \eta m^{\gamma + 1/2}\log(m); \cT, q}  + 2\eps & q \geq -\gamma\\
    F \paren{x(1-\eta/\log m)\n + \eta m^{\gamma + 1/2}\log(m); \cT, q}  + 2\eps & q < -\gamma,
\end{cases}
\end{align}and similarly,
\begin{align}
\PM{\sup_{t \in m\cT} \abs{\frac{S_t - t\Tilde \mu_{\P, t}}{\hat \sigma_t t}} \frac{t^q}{m^{q-1/2}} \leq x}& >\begin{cases}
    \Tilde F \paren{\Delta^{1/2-q} x(1+\eta/\log m)\n - \eta m^{\gamma + 1/2}\log(m); \cT, q}  -2\eps & q \geq -\gamma\\
    F \paren{ x(1+\eta/\log m)\n - \eta m^{\gamma + 1/2}\log(m); \cT, q}   -2\eps & q < -\gamma,
\end{cases}
\end{align}
\paragraph{Step 3a: Continuity of $\Tilde F$}
We again split up into cases and first assume that $q \geq -\gamma$. 
Since $1/2 + \gamma < 0$, there exists an $m_3 \equiv m_3(1/2)$ such that for all $m \geq m_3$
\begin{equation}\label{eq:second-term-decay}
m^{1/2 + \gamma}\log(m) < 1/2.
\end{equation}
Call $y \equiv \Delta^{1/2-q}x$. There exists an $m_4 \equiv m_4(1/2, y)$ such that $\forall m > m_4$ 
\begin{equation}\label{eq:m-large-of-x}
\frac{\abs{y}}{\log m-1} < 1/2.
\end{equation}
Let $m \geq \max \{m_1, m_2, m_3, m_4\}$. It follows from Step 1 and Step 2 that
\begin{align}
&\PM{\Delta^{1/2-q}\sup_{t \in m\cT} \abs{\frac{S_t - t\Tilde \mu_{\P, t}}{\hat \sigma_t t}} \frac{t^q}{m^{q-1/2}} \leq y} -  \Tilde F(y; \cT, q) \\
& < \Tilde F
\paren{\underbrace{y\paren{1-\eta/\log m}\n + \eta m^{1/2 + \gamma}\log(m)}_{\mathrm I}; \cT, q}  -   \Tilde F(y; \cT, q) + 2\eps.
\end{align}If $\abs{\mathrm{I} - y} < \eta$ then $\Tilde F(\mathrm{I}; \cT, q) - \Tilde F(y; \cT, q) < \eps$. It is straightforward to show
\begin{align}
\abs{\mathrm{I }-y}
&= \eta \abs{y \paren{\frac{1}{\log m - \eta}}+ m^{1/2 + \gamma} \log(m) }\\
&\leq \eta \abs{y \paren{\frac{1}{\log m -1}} + m^{1/2 + \gamma}\log(m)}\label{eq:q-large-than-12},
\end{align}where \cref{eq:q-large-than-12} uses fact that $\eta < 1$. Then using \cref{eq:second-term-decay} and \cref{eq:m-large-of-x} one has for $m$ sufficiently large
\begin{align}
\abs{\mathrm{I} - y} = \eta \abs{y \paren{\frac{1}{\log m - \eta}}+ m^{1/2 + \gamma} \log(m ) }&\leq \eta \abs{1/2 + 1/2}=\eta.
\end{align}Hence for all $m > \max \curly{m_1, m_2, m_3, m_4}$,
\begin{align}
\abs{\PM{\sup_{t \in m\cT} \abs{\frac{S_t - t \Tilde \mu_{\P, t}}{\hat \sigma_t t}}\frac{t^q}{m^{q-1/2}} \leq x} -  F(x; \cT, q)}&= \abs{\PM{\Delta^{1/2-q} \sup_{t \in m\cT} \abs{\frac{S_t - t \Tilde \mu_{\P, t}}{\hat \sigma_t t}}\frac{t^q}{m^{q-1/2}} \leq y} - \Tilde F(y; \cT, q)}\\
&\leq 3\eps,
\end{align} and we can similarly show 
\begin{align}
\Tilde F(y; \cT, q) - \PM{\Delta^{1/2-q}\sup_{t \in m\cT} \abs{\frac{S_t - t \Tilde \mu_{\P, t}}{\hat \sigma_t t}}\frac{t^q}{m^{q-1/2}} \leq y} &> -3\eps,
\end{align}which implies 
\begin{align}
\abs{\PM{\Delta^{1/2-q}\sup_{t \in m\cT} \abs{\frac{S_t - t \Tilde \mu_{\P, t}}{\hat \sigma_t t}}\frac{t^q}{m^{q-1/2}} \leq y} - \Tilde F(y; \cT, q)} < 6\eps.
\end{align}
In addition, when $q \in (1/2, -\gamma)$, the presented argument applies since 
\[
\Delta^{1/2-q}\eta m^{\gamma + 1/2}\log m \leq \eta m^{\gamma + 1/2}\log m.
\]
\paragraph{Step 3b: Continuity of $F$}
Now, assume that $q < -\gamma$. We must now bound $\abs{\mathrm{I} - x} < \eta$ where 
\begin{align}
\abs{\mathrm{I }-x}
&= \eta \abs{x \paren{\frac{1}{\log m -\eta}} + m^{1/2 + \gamma}\log(m)}.
\end{align}There exists an $m_5 \equiv m_5(1/2, x)$ such that $\forall m > m_5$,
\begin{equation}\label{eq:m-large-of-x-b}
\frac{\abs{x}}{\log m-1} < 1/2.
\end{equation}Therefore, for $m \geq M \equiv \max \curly{m_1, m_2, m_3, m_4, m_5}$,
\begin{align}
\eta \abs{x \paren{\frac{1}{\log m - \eta}}+ m^{1/2 + \gamma}\log(m) }&\leq \eta \abs{1/2 + 1/2}=\eta.
\end{align}Hence for all $m > M$, we have for all $\P \in \cP,$
\begin{align}
\PM{\sup_{t \in m\cT} \abs{\frac{S_t - t \Tilde \mu_{\P, t}}{\hat \sigma_t t}}\frac{t^q}{m^{q-1/2}} \leq x} -  F(x; \cT, q) &< 3\eps \text{ and similarly }\\
F(x; \cT, q) - \PM{\sup_{t \in m\cT} \abs{\frac{S_t - t \Tilde \mu_{\P, t}}{\hat \sigma_t t}}\frac{t^q}{m^{q-1/2}} \leq x} &> -3\eps,
\end{align}which implies 
\begin{align}
\abs{\PM{\sup_{t \in m\cT} \abs{\frac{S_t - t \Tilde \mu_{\P, t}}{\hat \sigma_t t}}\frac{t^q}{m^{q-1/2}} \leq x} -  F(x; \cT, q)} < 6\eps.
\end{align}
\paragraph{Step 4a: Upgrading Pointwise $\Tilde F$ Convergence to Uniform Convergence (Pólya continuity)}
In this step, we split into two \emph{different} cases: $q < 1/2$ and $q > 1/2$. First we assume that $q > 1/2$. We have shown
\[
\forall y \in \R, \lim_{m \to \infty}\sup_{\Delta \geq 1}\sup_{\cT \in \mathfrak{T}_\Delta}\sup_{\P \in \cP}\abs{\PM{\Delta^{1/2-q}\sup_{t \in m\cT}\abs{\frac{S_t - t\Tilde \mu_{\P, t}}{t \hat \sigma_t}}\frac{t^{q}}{m^{q-1/2}} \leq y} - \Tilde F(y; \cT, q)} = 0,
\]which is a pointwise in $y$ statement. Given that $\Tilde F$ is equicontinuous and uniformly tight (\cref{lemma:tilde-f-facts}), we can invoke \cref{lemma:dist-uniform-polya-continuity} to upgrade to uniform convergence. Thus we can guarantee that 
\[
\lim_{m \to \infty}\sup_{\Delta \geq 1}\sup_{\cT \in \mathfrak{T}_\Delta}\sup_{y \in \R}\sup_{\P \in \cP}\abs{\PM{\Delta^{1/2-q}\sup_{t \in m\cT}\abs{\frac{S_t - t\Tilde \mu_{\P, t}}{t \hat \sigma_t}}\frac{t^{q}}{m^{q-1/2}} \leq y} -  \Tilde F(y; \cT, q)} = 0,
\]
Then using $\Tilde F(\Delta^{1/2-q}x; \cT, q) = F(x;\cT,q)$ and using $F$ is invertible (\cref{lemma:unif-continuity-psi}) we have that
\[
\lim_{m \to \infty}\sup_{\Delta \geq 1}\sup_{\cT \in \mathfrak{T}_\Delta}\sup_{\alpha \in (0,1)}\sup_{\P \in \cP}\abs{\PM{\sup_{t \in m\cT}\abs{\frac{S_t - t\Tilde \mu_{\P, t}}{t \hat \sigma_t}}\frac{t^{q}}{m^{q-1/2}} \leq  F \n(\alpha; \cT, q)} - \alpha} = 0.
\]

\paragraph{Step 4b: Upgrading Pointwise $F$ Convergence to Uniform Convergence (Pólya continuity)} 
Now, we assume that $q < 1/2$. We can similarly guarantee using \cref{lemma:dist-uniform-polya-continuity} that 
\[
\lim_{m \to \infty}\sup_{\Delta \geq 1}\sup_{\cT \in \mathfrak{T}_\Delta}\sup_{x \in \R}\sup_{\P \in \cP}\abs{\PM{\sup_{t \in m\cT}\abs{\frac{S_t - t\Tilde \mu_{\P, t}}{t \hat \sigma_t}}\frac{t^{q}}{m^{q-1/2}} \leq x} -  F(x; \cT, q)} = 0,
\]
and using $F$ is invertible (\cref{lemma:unif-continuity-psi}) we have that
\[
\lim_{m \to \infty}\sup_{\Delta \geq 1}\sup_{\cT \in \mathfrak{T}_\Delta}\sup_{\alpha \in (0,1)}\sup_{\P \in \cP}\abs{\PM{\sup_{t \in m\cT}\abs{\frac{S_t - t\Tilde \mu_{\P, t}}{t \hat \sigma_t}}\frac{t^{q}}{m^{q-1/2}} \leq  F \n(\alpha; \cT, q)} - \alpha} = 0,
\]which, along with Step 4a, proves \cref{eq:two-sided-general-prop-unif-delta}.

\paragraph{Step 5: Weak Uniformity when $q = 1/2$}\label{paragraph:step5}
Lastly, we note that the strong-uniformity proved above fails when $q = 1/2$. We note that $\eta_0$ is $\Delta$-uniform only if $q \neq 1/2$ and this is the first place uniformity breaks. Moreover, we require the distribution functions in Step 4 to be equicontinuous and uniformly tight. When $q = 1/2$, these conditions fail. However, for any fixed $\Deltabar < \infty$, we can invoke \cref{lemma:valid-continuity} as $a > 0, b < \infty$ for continuity and tightness.

\paragraph{Convergence to $F_+$}

The one-sided proof follows a similar logic to the two-sided proof, but there are some changes. We state the places in the proof that need to be changed for $F_+$.

\paragraph{One-Sided Step 1:}
The major change in the proof comes from having to reconsider controlling variance estimation in Step 1. The one-sided statistic is no longer nonnegative, so the multiplicative perturbation must be sign-aware. Because $u \mapsto u \pm \eta \abs{u}/\log m$ is increasing, we can replace $x\paren{1 \mp \eta/\log m}$ with $x \mp \eta \abs{x}/\log m$.

\paragraph{One-Sided Quantile Uniformity}
For the one-sided distribution, we obtain quantile uniformity over $\alpha \in [\alpha_1, 1)$ for any $\alpha_1 > 1/2$. This changes Step-4, which will now apply \cref{lemma:dist-uniform-polya-continuity} with $\alpha_1$ fixed.
\end{proof}

\subsubsection{Instantiating \cref{lem:general-lemma} for Confidence Horizons and GSMs}
With the lemma defined above in \cref{lem:general-lemma} we can now prove the theorems which fall out as consequences of the lemma. For confidence horizons, note that the uniformity over $\cT \subset \mathfrak{T}_\Delta$ is redundant since we take $\cT = [1,\Delta]$.

\begin{proof}[Proof of \cref{theorem:general-mean} (Two-sided Confidence horizons for independent data)]\label{proof:general-mean}
For $q \neq 1/2$, we must show 
\[
\lim_{m \to \infty} \ \sup_{\P \in \cP} \sup_{\Delta \geq 1}\sup_{\alpha \in (0, 1)}\abs{\PM{\forall t \in [m, \Delta m], \mu_\P  \in \hat \mu_t \pm \mathfrak{B}_t^{(m, \Delta)} }-(1-\alpha)} = 0.
\]Rearranging terms gives the equivalent 
\begin{align}
\lim_{m \to \infty} \ \sup_{\P \in \cP} \sup_{\Delta \geq 1}\sup_{\alpha \in (0, 1)}\abs{\PM{\sup_{t \in [m, \Delta m]} \abs{\frac{S_t - t \Tilde \mu_{\P, t}}{t \hat \sigma_t}}\frac{t^q}{m^{q-1/2}} \leq \Psi \n(1-\alpha; \cT, q)} - (1-\alpha)} = 0.
\end{align} Setting $\cT = [1,\Delta]$ and applying \cref{lemma: calc-quantile}, we find the equivalent equation in \cref{eq:two-sided-general-prop-unif-delta}. Thus it is sufficient to check if \cref{cond:general-lemma} is satisfied. \emph{C1} is satisfied by \cref{condition:strong-consistent-var-estimator} and \emph{C2} is satisfied by \cref{lemma:dist-uniform-kmt}. When $q = 1/2$, \cref{lem:general-lemma} indicates we only obtain the $\cP$-weakly uniform statement.
\end{proof}
\begin{proof}[Proof of \cref{prop:one-sided-confidence-horizons} (One-sided confidence horizons for independent data)]\label{proof:one-sided-mean}
The one-sided proposition follows directly from the two-sided one after checking the relevant conditions in \cref{lem:general-lemma} and noting the $\alpha_1$ quantile uniformity restriction.
\end{proof}

\begin{proof}[Proof of \cref{theorem:mgale-asympcv}(Confidence horizons for martingale-dependent data)]\label{proof:mgale-asympch}
As explained in the proof of \cref{theorem:general-mean}, it is sufficient to check \emph{C1} and \emph{C2} in \cref{cond:general-lemma}. \emph{C1} is satisfied by \cref{condition:strong-consistent-var-estimator}. \emph{C2} is satisfied by \cref{lemma:strong-approximation-mgale-dependence}. When $q \neq 1/2$, \cref{lem:general-lemma} indicates $\P$-strong uniformity. When $q = 1/2$, we only obtain $\P$-weak uniformity.
\end{proof}

\begin{proof}[Proof of \cref{prop:unif-valid-gsm-ind} (GSMs for independent data)]\label{proof:unif-valid-gsm-ind}
We need to show that when $q \neq 1/2$,
\[
\lim_{n_1 \to \infty} \sup_{\P \in \cP} \sup_{ \Delta\geq 1}  \sup_{K \in \N} \sup_{\alpha \in (0, 1)}\abs{\PM{\sup_{n_k \in n_1\cT} \abs{\frac{S_{n_k} - n_k \Tilde \mu_{\P, n_k}}{n_k \hat \sigma_{n_k}}}\frac{n_k^q}{n_1^{q-1/2}} \leq  \Xi\n(\alpha, \Delta, q)} - \alpha} = 0,
\]where $\cT = \curly{1, \dots, \Delta} \subseteq [1,\Delta]$ and $\Delta = n_K/n_1$. We note that the result of \cref{lem:general-lemma} gives the desired limit guarantee. This is because of the uniformity over $\cT \subset \mathfrak{T}_\Delta$ makes the approximation uniform over the size of the set $\cT$. Therefore, we check \emph{C1} and \emph{C2} in \cref{cond:general-lemma} are satisfied by \cref{condition:strong-consistent-var-estimator} and \cref{lemma:dist-uniform-kmt}. 
The final step is to verify that 
\[
\Xi(\alpha, n_1\cT_{n_1}, q) = \PM{\sup_{n_k \in n_1\cT_{n_1}} \abs{W(n_k)}n_k^{q-1}n_1^{1/2 - q} \leq x},
\]which we verify in \cref{lemma: calc-quantile}. 
When $q = 1/2$, we obtain $\cP$-weak uniformity.
\end{proof}

The proof of \cref{prop:unif-valid-gsm-mgale} is straightforward after the proofs of \cref{prop:unif-valid-gsm-ind} and \cref{theorem:mgale-asympcv} detailed above.

\subsection{Neyman Allocation Conditions}\label{sec:check-neyman-conditions}
This section checks the conditions required for the martingale-dependent AsympCH for the Neyman allocation illustration. We fix a $j$ and drop it for ease of notation. Recall
\begin{align}
Z_t(j) \equiv Z_{t} = \frac{\paren{Y_t - \hat \mu_{t-1}}\indic{A_t = j}}{\PM{A_t = j \mid \cF_{t-1}}} + \hat \mu_{t-1}
\end{align}Further, define $\pi_t \equiv \pi_{t-1}(j) = \PM{A_t = j \mid \cF_{t-1}}$ and $\star \pi = \lim_{t \to \infty}\pi_{t-1}$.

\paragraph{Variances Converge}
We first check \cref{condition:non-vanishing-variances}. Fix a single $j$.

\begin{align}
\Var{Z_{t} \mid \cF_{t-1}}
&= \frac{1}{\pi_{t-1}^2(j)} \Var{\paren{Y_t - \hat \mu_{t-1}}\indic{A_t = j}\mid \cF_{t-1}}
\end{align}
The first and second moment are
\begin{align}
\EE{\paren{Y_t - \hat \mu_{t-1}} \indic{A_t = j} \mid \cF_{t-1}} &= \pi_{t-1} \paren{\mu - \hat \mu_{t-1}}\\
   \EE{\paren{Y_t - \hat \mu_{t-1}}^2 \indic{A_t = j} \mid \cF_{t-1}} &= \pi_{t-1} \paren{\sigma^2 + (\mu - \hat \mu_{t-1})^2}
\end{align}
Hence 
\begin{align}
\Var{ Z_t \mid \cF_{t-1}} &= \frac{1}{\pi_{t-1}}\paren{\sigma^2_\P + \paren{\mu - \hat \mu_{t-1}}^2 - \pi_{t-1}\paren{\mu - \hat \mu_{t-1}}^2}\\
&= \underbrace{\frac{\sigma^2_\P}{\pi_{t-1}}}_{\mathrm I} + \underbrace{\frac{1-\pi_{t-1}}{\pi_{t-1}}\paren{\mu - \hat \mu_{t-1}}^2}_{\mathrm {II}} \label{eq:cond-variance-mt}
\end{align}
Under the assumption that $\inf_j \sigma_j^2 > 0$, each arm $j$ will be pulled infinitely often. Therefore by the LIL, we have that $\mathrm{II} = \bar O_{\P}(\log \log t / t)$ which will be dominated by the first term.

To see the convergence of the first term we must consider the convergence of the variance estimators. We assume $r > 2$. If $r > 4$, the LIL activates and the rate is capped at $1/2$. Thus assume $r \in (2,4)$. By the Marcinkiewicz--Zygmund SLLN and \citep[Proposition 2.2]{waudby-smith_distribution-uniform_2026}
\[
\abs{\hat \sigma^2_t - \sigma^2_\P} = \bar o_{\cP}\paren{t^{-\rho}},
\]where $\rho = 1-4/(2 + r)$.
This holds for each estimator meaning
\[
\abs{\frac{\hat \sigma_t}{\sum_{k=1}^{J} \hat \sigma_t(k)} - \frac{\sigma_\P}{\sum_{k=1}^{J}\sigma_\P(k)}} = \bar o_{\P}\paren{t^{-\rho}}.
\]Therefore, we have 
\[
\abs{\pi_{t-1} - \star \pi} = \bar o_{\P}\paren{t^{-\rho}}.
\]Hence the $\rho$ parameter needed for \cref{condition:non-vanishing-variances} is $ \rho = 1 - 4/(2 + r)$, which is greater than $0$ since $r > 2$. Note that we can invoke the Marcinkiewicz--Zygmund SLLN because the estimators are of \iid{} data.

\paragraph{Uniform Lindeberg-Summability}
By assumption, $Y_t$ has an $r$ uniformly integrable moment for $r > 2$. With this fact and uniformly positive variances, we can satisfy \cref{condition:lindeberg-type-mgale} for $\kappa \in (2,r)$. 

\subsection{Auxiliary Lemmas}
\begin{lemma}\label{lemma:proper-limit}
Let $q < 1/2$. For any $\cT \subset [a,b]$, $1 \leq a \leq b \leq \infty$, where $\inf \cT = a, \sup \cT = b$, the distribution functions
\[
F(x; \cT, q) = \PM{\sup_{s \in \cT} \abs{W(s)}s^{q-1} \leq x} \quad \text{ and } \quad F_+(x; \cT, q) = \PM{\sup_{s \in \cT} W(s)s^{q-1} \leq x}
\]are continuous and proper. In particular $F/F_+(x;[1,\infty),q)$ are continuous and proper.
\end{lemma}

\begin{proof}[Proof of \cref{lemma:proper-limit}]
We show that the distributions are proper and continuous.
\paragraph{Properness} By the LIL, for $q<1/2$, $W(t)t^{q-1} = \bar o_{\P}(1)$. So $\zeta(\cT,q),\zeta_+(\cT,q)<\infty$ a.s.\ for any $\cT\subset[a,\infty)$. This implies 
\begin{equation}
\lim_{x \to \infty}F/F_+(x; \cT, q)
 = 1 \quad\text{ and } \quad \lim_{x \to -\infty}F/F_+(x; \cT, q) = 0.
 \end{equation}The second limit for $F_+$ uses $\zeta_+(\cT,q)\ge W(1)$ a.s. Properness for $b<\infty$ follows since the sup of a continuous function over a compact set is a.s. finite.
\paragraph{Continuity} \cref{lemma:anti-concentration-small-q} gives, for every $\eps>0$, an $\eta(\eps)$ continuity parameter that is uniform over $\Delta \geq 1$ for every
$\cT\subset[a,\Delta]$. Since $F(x;\cT\cap[a,\Delta],q)\downarrow F(x;\cT,q)$ pointwise as $\Delta\to\infty$, this $\Delta$-uniform $\eta(\eps)$ passes to the limit. Thus $F/F_+(\cdot;\cT,q)$ are continuous.
\end{proof}

\begin{lemma}\label{lemma:support-strict-increase}
Let $q \in \R$, $\Delta \in [1, \infty)$, and $\cT \subset [1,\Delta]$ where $\inf \cT = 1$ and $\sup\cT = \Delta$. Then $\zeta(\cT,q)$ has support $[0,\infty)$ and $\zeta_+(\cT,q)$ has support $\R$. Hence
\[
F(\cdot\,; \cT, q) : (0,\infty) \to (0,1) \quad \text{ and }\quad F_+(\cdot\,;\cT,q) : \R \to (0,1)
\]
are continuous, strictly increasing bijections, and $F\paren{F\n(\alpha)} = \alpha$ for every $\alpha \in (0,1)$. The same holds for $\Tilde F$ and $\Tilde F_+$.
\end{lemma}

\begin{proof}[Proof of \cref{lemma:support-strict-increase}]

Since $s^{q-1} \leq L \coloneqq \max\curly{1, \Delta^{q-1}}$ on $[1,\Delta]$, the maps
\[
H(w) = \sup_{s \in \cT}\abs{w(s)}s^{q-1} \quad \text{ and } \quad H_+(w) = \sup_{s \in \cT}w(s)s^{q-1}
\]
are $L$-Lipschitz on $C[0,\Delta]$ under $\norm[\infty]{\cdot}$. Take the linear path of slope $a$, $w_a(s) = as$, and set $k_+ = \sup_{s\in\cT}s^q$, $k_- = \inf_{s\in\cT}s^q$. Then $H(w_a) = \abs{a}k_+$ sweeps $[0,\infty)$ and $H_+(w_a) = ak_+ \vee ak_-$ sweeps $\R$ as $a$ ranges over $\R$.

Fixing $v \in [0, \infty)$ and $\delta > 0$, we choose $a$ so that $H(w_a) = v$. Since $w_a$ is continuous with $w_a(0) = 0$, it lies in the support of Wiener measure \citep[Exercise 1.8]{morters_brownian_2010}. Thus,
\[
\PM{\abs{\zeta(\cT, q) - v} < L\delta} \geq \PM{\sup_{s \in [0, \Delta]}{\abs{W(s) - w_a(s)}} < \delta} > 0,
\]
and $v$ is in the support of $\zeta(\cT,q)$. The same argument applies to $H_+$ taking $v \in \R$.

It follows that $F(y) - F(x) = \PM{\zeta(\cT,q) \in (x,y]} > 0$ for any $x < y$ in the support, and $F$ is continuous by \cref{lemma:valid-continuity} since $a = 1 > 0$ and $b = \Delta < \infty$. A continuous, strictly increasing function onto $(0,1)$ is a bijection, so $F\n$ is well defined and $F(F\n(\alpha)) = \alpha$ exactly. Finally, $\Tilde F(x; \cT, q) = F(\Delta^{q-1/2}x; \cT, q)$ is a strictly increasing reparameterization, so the same result holds for $\Tilde F$.
\end{proof}

\begin{lemma}\label{lemma:max-discrete-cont-gap}
Let $\infseqt{Y} \simiid \cN(0,1)$ on some probability space $\pspace[\cP]$ rich enough to contain a Wiener $W(t)$ where at integer values $t, W(t) = \tsum{Y_i}$. Then,
\begin{equation}\label{eq:discretization-error-gaussian}
\abs{W(t) - W(\floor t)} = \bar O_{ \cP}(\sqrt{\log t}).
\end{equation}
\end{lemma}
\begin{proof}[Proof of \cref{lemma:max-discrete-cont-gap}]
The proof is adapted from \citep[Theorem C]{csorgo_how_1979}. The claim made in \cref{eq:discretization-error-gaussian} is equivalent to showing 
$\exists C$ so that
\[
\lim_{m \to \infty}\sup_{\P \in \cP} \PM{\sup_{t\geq m} \abs{\frac{W(t) - W(\floor t)}{\sqrt{\log t}}} \geq C} = 0.
\]Let $m > 3$. Then,
\begin{align}
\PM{\sup_{t \geq m} \abs{\frac{W(t) - W(\floor t)}{\sqrt{\log t}}} \geq C} &= \PM{\bigcup_{k = m}^{\infty} \max_{k \leq s \leq k+1} \abs{\frac{W(s) - W(k)}{\sqrt{\log s}}} \geq C}\\
&\leq \sum_{k=m}^{\infty}\PM{\max_{k \leq s \leq k+1} \abs{\frac{W(s) - W(k)}{\sqrt{\log s}}} \geq C}\\
&\leq \sum_{k=m}^{\infty}\PM{\max_{k \leq s \leq k+1} \abs{W(s) - W(k)} \geq C{\sqrt{\log k}}},
\end{align}where the last line uses the fact that $\log s \geq \log k$ and thus we clip the boundary to the left edge. We continue
\begin{align}
\sum_{k=m}^{\infty}\PM{\max_{k \leq s \leq k+1} \abs{W(s) - W(k)} \geq C{\sqrt{\log k}}} &= \sum_{k=m}^{\infty}\PM{\max_{0 \leq s \leq 1}\abs{W(s) - W(0)} \geq C\sqrt{\log k}}\\
&= \sum_{k=m}^{\infty}\PM{\max_{0 \leq s \leq 1}\abs{W(s)} \geq C\sqrt{\log k}}\\
&\leq 2\sum_{k=m}^{\infty}\PM{\max_{0 \leq s \leq 1} W(s) \geq C\sqrt{\log k}}\\
&= 4\sum_{k=m}^{\infty}\PM{Z\geq C\sqrt{\log k}}\\
&\leq \frac{4}{\sqrt {2 \pi}}\sum_{k=m}^{\infty} \frac{1}{C \sqrt{\log k}}\expp{-\frac{C^2\log k}{2}}\\
&= \frac{4}{\sqrt{2 \pi}} \sum_{k=m}^{\infty} \frac{1}{C\sqrt{\log k}}\frac{1}{k^{C^2/2}}
\end{align}
Letting $C > \sqrt{2}$ gives us the desired convergence. We remark that the third line union bounds $W(s)$, the fourth line uses the reflection principle, and fifth line uses the tail bound of a standard normal.
\end{proof}

\begin{lemma}\label{lemma:wiener-increment-bound}
Let $W(t)$ be a Wiener process on $\pspace[\cP]$. Suppose $V_{\P,t}$ has the property that $\abs{V_{\P, t} - \sigma^2_\P t} = \bar O_{\cP}(t^{1-\rho})$. Then it follows 
\begin{align}
W(V_{\P, t}) - W(\sigma^2_\P t) = \bar O_{\cP}\paren{\sqrt{t^{1-\rho}\log t }}.
\end{align}
\end{lemma}
\begin{proof}[Proof of \cref{lemma:wiener-increment-bound}]
It is sufficient to show that there exists a $C < \infty$ such that 
\begin{align}
\lim_{m \to \infty}\sup_{\P \in \cP}\PM{\sup_{k \geq m} \frac{\abs{W(V_{\P, k}) - W(\sigma^2_{\P} k)}}{\sqrt{k^{1-\rho} \log k}} \geq C} = 0.
\end{align}
By our assumption, there exists an $M$ such that $\forall k> M$ we have with probability $1-\eps$,
\[
\abs{V_{\P, k} - \sigma^2_\P k} < C_0 k^{1-\rho}
\]and define this window  to be $\cI_s = [\sigma^2_\P s - C_0s^{1-\rho}, \sigma^2_\P s + C_0 s^{1-\rho}]$. First,
\begin{align}
\PM{\sup_{k \geq m} \frac{\abs{W(V_{\P, k}) - W(\sigma^2_{\P} k)}}{\sqrt{k^{1-\rho} \log k}} \geq C}  &\leq \PM{\sup_{k \geq m} \frac{\abs{W(V_{\P, k}) - W(\sigma^2_{\P} k)}}{\sqrt{k^{1-\rho} \log k}} \geq C, \abs{V_{\P, k} - \sigma^2_\P k} < C_0k^{1-\rho}} + \eps.
\end{align}Hence we will bound the first term by $\eps$ on the event $\sup_{k \geq m}\abs{V_{\P, k} - \sigma^2_\P k} < C_0k^{1-\rho}$. On this event,
\begin{align}
\PM{\sup_{k \geq m} \frac{\abs{W(V_{\P, k}) - W(\sigma^2_{\P} k)}}{\sqrt{k^{1-\rho} \log k}} \geq C}
&= \PM{\bigcup_{k = m}^{\infty}  \frac{\abs{W(V_{\P, k}) - W(\sigma^2_{\P}k )}}{\sqrt{k ^{1-\rho} \log k}} \geq C}\\
&\leq \sum_{k=m}^{\infty}\PM{\frac{\abs{W(V_{\P, k}) - W(\sigma^2_\P k)}}{\sqrt{k^{1-\rho}\log k}} \geq C}\\
&= \sum_{k =m}^{\infty}\PM{ \abs{W(V_{\P, k}) - W(\sigma^2_\P k)}\geq  C\sqrt{k^{1-\rho}\log k}}\\
&\leq \sum_{k =m}^{\infty}\PM{ \sup_{u \in \cI_k}\abs{W(u) - W(\sigma^2_\P k)}\geq  C\sqrt{k^{1-\rho}\log k}}\\
&\leq \sum_{k=m}^{\infty}\PM{2 \sup_{\smash{u \in [0, 2C_0k^{1-\rho}]}} \abs{W(u)} \geq C \sqrt{k^{1-\rho}\log k}} \label{eq:wiener-max-algebra}\\
&= \sum_{k=m}^{\infty}\PM{\sup_{u \in [0,1]} \abs{W(u)} \geq C/(2\sqrt{2C_0})\sqrt{\log k}}.
\end{align}The last expression limits to $0$ as $m \to \infty$ if $ C > 4 \sqrt{C_0}$. The step to obtain \cref{eq:wiener-max-algebra} is somewhat involved and we explain it in \cref{lemma:involved-wiener-max-bound}.

Hence $\forall \eps > 0$, there exists $M$ such that for all $m > M$ and all $\P \in \cP$,
\[
\PM{\sup_{k \geq m}\frac{\abs{W(V_{\P, k}) - W(\sigma^2_\P k)}}{\sqrt{k^{1-\rho}\log k}} \geq C} < 2\eps.
\]
\end{proof}

\begin{lemma}\label{lemma:involved-wiener-max-bound}
Let $W(t)$ be a Wiener. Then 
\[
\PM{\sup_{u \in [a - b, a + b]}\abs{W(u) - W(a)} \geq x} \leq \PM{2\sup_{u \in [0, 2 b]} \abs{W(u)} \geq x}.
\]
\end{lemma}
\begin{proof}[Proof of \cref{lemma:involved-wiener-max-bound}]
We begin with
\begin{align}
\abs{W(u) - W(a)} \leq \abs{W(u) - W(a-b)} + \abs{W(a) - W(a-b)}.
\end{align}Then let $s = u + b-a$  to obtain 
\begin{align}
\sup_{u \in [a-b, a+b]} \abs{W(u) - W(a)} &\leq \sup_{s \in [0, 2b]}\abs{W(s + a-b) - W(a-b)} + \abs{W(a) - W(a-b)}\\
&\leq \sup_{s \in [0, 2b]}2\abs{W(s+ a - b) - W(a-b)}\label{eq:one-point-in-maximum}\\
&\overset{d}{=} 2\sup_{s \in [0, 2b]}\abs{W(s)},
\end{align}where \cref{eq:one-point-in-maximum} uses that when $s = b$, $W(a) = W(s + a-b)$ and the final line uses the translation invariance of a Wiener.
\end{proof}

\begin{lemma}\label{lemma:suff-mgale-condition}
Let $\infseqt{Y}$ be a sequence of random variables on $\pspace$. Suppose there exists and $r > 2$ and $B < \infty$ such that for all $t \in \N$
\[
\EE{\abs{Y_t}^r \mid \cF_{t-1}} < B.
\]Moreover, assume \cref{condition:non-vanishing-variances}. Then for any $\kappa \in (2,r)$,
\[
\sum_{t=1}^{\infty}\frac{\EE{\paren{Y_t - \mu_{\P, t}}^2\indic{\paren{Y_t - \mu_{\P, t}}^{2} > V_t^{2/\kappa}}\mid \cF_{t-1}}[\P]}{V_t^{2/\kappa}} < \infty.
\]
\end{lemma}

\begin{proof}[Proof of \cref{lemma:suff-mgale-condition}]
Let $\kappa \in (2,r)$. On the event that $\paren{Y_t - \mu_{\P, t}}^2 > V_t^{2/\kappa}$ we have $\abs{Y_t -\mu _{\P, t}} > V_t^{1/\kappa}$. Since $2-r < 0$ it follows that $\abs{Y_t - \mu_{\P, t}}^{2-r} < V_t^{(2-r)/\kappa}$ and $(Y_t - \mu_{\P, t})^2 < \abs{Y_t - \mu_{\P, t}}^r V_t^{(2-r)/\kappa}$ Therefore
\begin{align}
\frac{\EE{\paren{Y_t - \mu_{\P, t}}^2\indic{\paren{Y_t - \mu_{\P, t}}^{2} > V_t^{2/\kappa}}\mid \cF_{t-1}}}{V_t^{2/\kappa}} &\leq \frac{\EE{\abs{Y_t - \mu_{\P, t}}^r V_t^{(2-r)/\kappa}\mid \cF_{t-1}}[\P] }{V_t^{2/\kappa}}\\
&= \frac{\EE{\abs{Y- \mu_{\P, t}}^{r}\mid \cF_{t-1}}[\P] V_t^{(2-r)/\kappa}}{V_t^{2/\kappa}}\\
&\leq \frac{2^{r-1}\EE{\abs{Y_t}^r + \abs{\mu_{\P, t}}^r\mid \cF_{t-1}}}{V_t^{r/\kappa}}\\
&\asymp \frac{2^{r-1}\EE{\abs{Y_t}^r + \abs{\mu_{\P, t}}^r\mid \cF_{t-1}}}{t^{r/\kappa}}
\end{align}
The numerator is upper bounded by some constant. Thus if $\kappa \in (2, r)$ then the series is summable.
\end{proof}

\begin{lemma}\label{lemma:tv-wt-gt}
Let $g$ lie in the Cameron--Martin space. For $\P$ the standard Wiener measure and $\mathsf Q$ the measure of $(W - g)(t)$ we have 
\[
{\rm TV}\paren{\P, \mathsf Q} = 2\Phi \paren{\norm[\cH]{g}/2} - 1.
\]
\end{lemma}

\begin{proof}[Proof of \cref{lemma:tv-wt-gt}]
Corollary 2.4.3 from \citet[Corollary 2.4.3]{bogachev_gaussian_1998} implies that
\[
\frac{\dQ}{\dP}(x) = \expp{\hat g - \frac{1}{2}\norm[\cH]{g}^2}
\]where $\hat g \sim \cN(0, \norm[\cH]{g}^2)$ under $\P$. Therefore, 
\[
{\rm TV}\paren{\P, \Q} = {\rm TV}\paren{\cN\paren{0, \norm[\cH]{g}^2},\, \cN\paren{\norm[\cH]{g}^2, \norm[\cH]{g}^2}}.
\]Using the fact that for $Z_1, Z_2$ standard normals
\[
{\rm TV}(\sigma(Z_1 - \mu_1), \sigma(Z_2- \mu_2)) = 2\Phi \paren{\frac{\abs{\mu_1 - \mu_2}}{2\sigma}}-1,
\]gives the result.
\end{proof}

\section{Distribution Functions for Constant, Root, and Linear CH}
This section derives closed-form distribution functions for $\Psi$ and $\Psi_+(x; \Delta, q)$ when $q \in \curly{0, 1/2, 1}$.
\paragraph{Notation.} For the proofs of stochastic processes below, we use the notation $\PM{\cdot}[x]$ to denote that the process started (at $t = 0$) at $x$. If the subscript is omitted, then $\PM{\cdot} = \PM{\cdot}[0]$. In addition, we use the bivariate normal notation where 
\[
\Phi_2(z_1, z_2; \rho) = \PM{ Z_1 \leq z_1,  Z_2 \leq z_2},
\]where $( Z_1,  Z_2) $ are bivariate unit variance normals with correlation $\rho$.
\subsection{Distribution Functions When \texorpdfstring{$q\in \curly{0,1}$}{q in {0,1}}}\label{sec:exact-quantiles-q01}

\subsubsection{Two-Sided Distributions}
We are first interested in the distributions of $\zeta(\Delta, q)$ when $q \in \{0, 1\}$.
When $q = 0$ we can use the fact that a certain scaled Wiener process equals itself in distribution. It is straightforward to show using $W(s) = sW(1/s)$ that
\[
\PM{\sup_{s \in [1, \Delta]} \abs{W(s)}s^{-1} \leq x} = \PM{\sup_{s \in [1, \Delta]} \abs{W(s)} \leq x \sqrt{\Delta}}.
\]
Therefore $\PM{\zeta(\Delta, 0) \leq x} = \PM{\zeta(\Delta, 1) \leq x \sqrt{\Delta}}$ and thus results for $q = 1$ give results for $q = 0$.

\begin{lemma}
Let $W(t)$ be a standard Wiener process. Define
\[
g_{k,x}(c; \rho) = \Phi_2\paren{\frac{(4k+c)x}{\sqrt{\Delta}}, x; \rho} - \Phi_2 \paren{\frac{(4k+c)x}{\sqrt{\Delta}}, -x; \rho}
\]Then 
\begin{align}
\PM{\sup_{s \in [1, \Delta]} \abs{W(s)} \leq x} &= \sum_{k=-\infty}^{\infty} g_{k,x}(1; \Delta^{-1/2}) - g_{k,x}(-1; \Delta^{-1/2}) - g_{k,x}(3; -\Delta^{-1/2}) + g_{k,x}(1; -\Delta^{-1/2}).
\end{align}
\end{lemma}

\begin{proof}
We condition on $W(1)$ and write
\begin{align}
\PM{\sup_{1 \leq s \leq \Delta} \abs{W(s)} < x} &=
\int_{-\infty}^{\infty}\PM{\sup_{1 \leq s \leq \Delta} \abs{W(s)} < x \mid W(1) = z}\phi(z)\dZ\\
&= \int_{-x}^{x}\PM{\sup_{1 \leq s \leq \Delta} \abs{W(s)} < x \mid W(1) = z}\phi(z)\dZ\\
&= \int_{-x}^{x}\PM{\sup_{0 \leq s \leq \Delta-1}\abs{W(s)} < x \mid W(0) = z}\phi(z)\dZ\\
&= \int_{-x}^{x}\PM{\sup_{0 \leq s \leq \Delta-1}\abs{W(s)} < x}[z]\phi(z)\dZ
\end{align}Invoking \cref{lemma:wiener-escapes-box} on the last line gives
\begin{align}
\PM{\sup_{1 \leq s \leq \Delta} \abs{W(s)} < x} &= \int_{-x}^{x}\sum_{k=-\infty}^{\infty} \Phi \paren{\frac{(4k + 1)x - z}{\sqrt{\Delta-1}}} - \Phi \paren{\frac{(4k-1)x - z}{\sqrt{\Delta-1}}} - &\\
&\Phi \paren{\frac{(4k+ 3)x + z}{\sqrt{\Delta-1}}} + \Phi \paren{\frac{(4k+1)x + z}{\sqrt{\Delta-1}}}\phi(z)\dZ
\end{align}
One can check that Fubini applies and we can interchange integral and series. By \cref{lemma:normal-cdf-integral} we have a closed-form expression for this integral.\footnote{If we had chosen the Fourier series expression, we would have to integrate a cosine against a normal density which results in complex functions.}

We go through term by term now. 
\begin{align}
\int_{-x}^{x}\Phi\paren{\frac{(4k+1)x - z}{\sqrt{\Delta-1}}}\phi(z)\dZ
  &= \Phi_2\paren{\frac{(4k+1)x}{\sqrt{\Delta}}, x; \Delta^{-1/2}}
   - \Phi_2\paren{\frac{(4k+1)x}{\sqrt{\Delta}}, -x; \Delta^{-1/2}},\\
\int_{-x}^{x}\Phi\paren{\frac{(4k-1)x - z}{\sqrt{\Delta-1}}}\phi(z)\dZ
  &= \Phi_2\paren{\frac{(4k-1)x}{\sqrt{\Delta}}, x; \Delta^{-1/2}}
   - \Phi_2\paren{\frac{(4k-1)x}{\sqrt{\Delta}}, -x; \Delta^{-1/2}},\\
\int_{-x}^{x}\Phi\paren{\frac{(4k+3)x + z}{\sqrt{\Delta-1}}}\phi(z)\dZ
  &= \Phi_2\paren{\frac{(4k+3)x}{\sqrt{\Delta}}, x; -\Delta^{-1/2}}
   - \Phi_2\paren{\frac{(4k+3)x}{\sqrt{\Delta}}, -x; -\Delta^{-1/2}},\\
\int_{-x}^{x}\Phi\paren{\frac{(4k+1)x + z}{\sqrt{\Delta-1}}}\phi(z)\dZ
  &= \Phi_2\paren{\frac{(4k+1)x}{\sqrt{\Delta}}, x; -\Delta^{-1/2}}
   - \Phi_2\paren{\frac{(4k+1)x}{\sqrt{\Delta}}, -x; -\Delta^{-1/2}}.
\end{align}
\end{proof}

\subsubsection{One-Sided Distribution Functions}
There are many ways to find this distribution function, and we present a novel way that bypasses integration, instead relying on bivariate normal functions.

\begin{proposition}
Let $W(t)$ be a standard Wiener process. Then 
\[
\PM{\sup_{s \in [1, \Delta]} W(s) \leq x} = 2 \Phi_2 \paren{x, \frac{x}{\sqrt{\Delta}} ; \frac{1}{\sqrt{\Delta}}} - \Phi(x)
\]
\end{proposition}

\begin{proof}

First consider two types of paths. One set of realizations are above $x$ at time $1$ meaning $W(1) > x$. The complement set of paths are those that are below $x$ at time $1$. This second set of paths are the candidate paths whose maximum could be less than $x$. If $W(1) > x$ then there is no more work to be done.

\begin{figure}[!ht]
    \centering
    \includegraphics[width=0.5\linewidth]{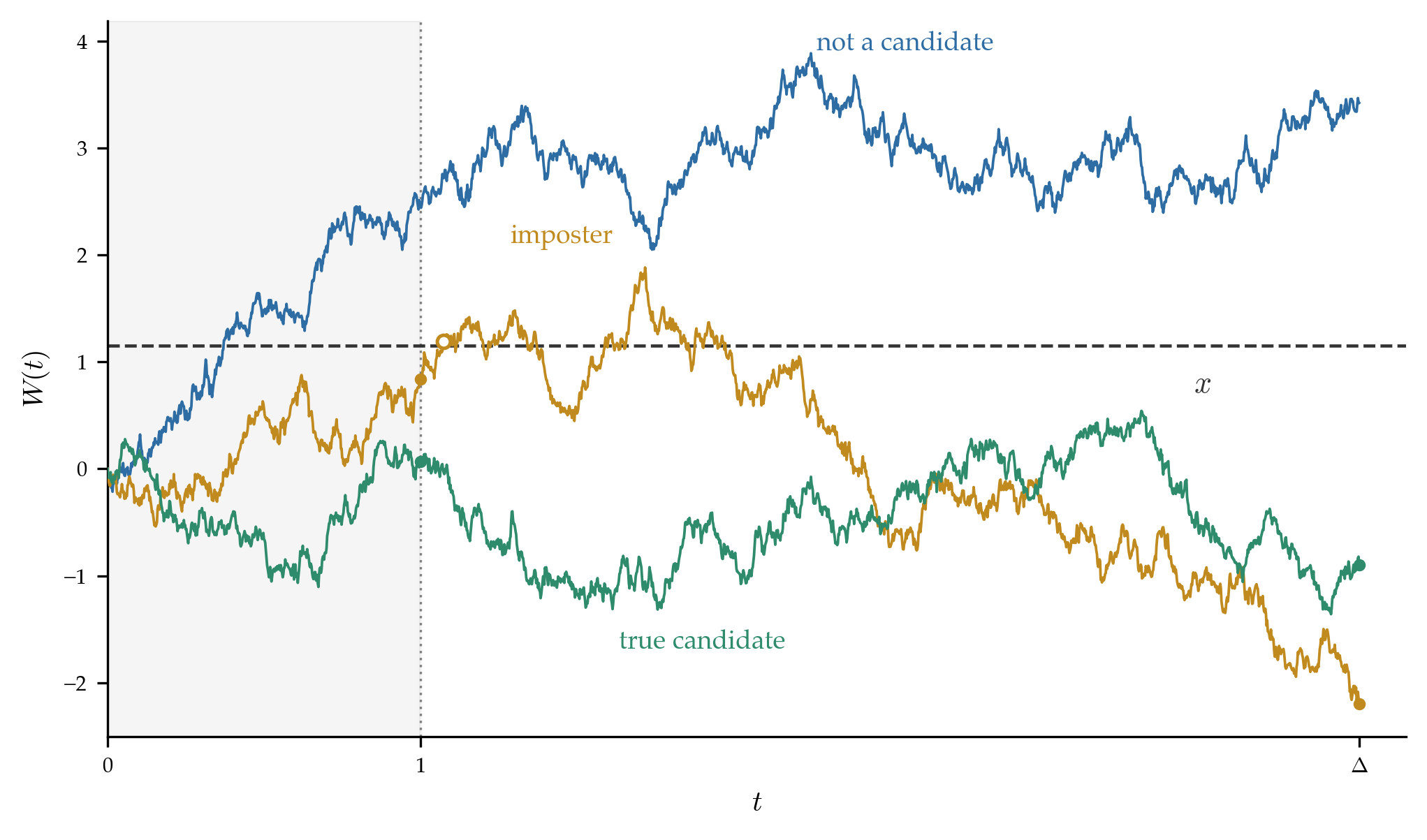}
    \caption{Illustration of candidates and imposters.}
\end{figure}

Within this second set of paths is another partition. If $W(\Delta) > x$ then clearly by the intermediate value theorem it must have crossed in $[1, \Delta]$. Therefore we are interested in the set of paths that begin below $x$ and end below $x$. These are the \emph{true} candidate paths. However, we cannot just count up these paths. This is because there are some ``imposter" paths who begin below $x$ and end below $x$ but hit $x$ in $[1, \Delta]$ by chance.  

Therefore, we claim that 
\begin{align}
F(x) &= \PM{\sup_{s \in [1, \Delta]} W(s) \leq x}\\
&= \PM{\textsf{candidate}} - \PM{\textsf{imposter}}\\
&= \PM{W(1) < x, W(\Delta) <x} - \PM{\textsf{imposter}}.
\end{align}
For the candidate probability we have two normals with known covariance so we can calculate a bivariate normal. For the imposter probability, we must invoke the reflection principle. 

An imposter is a path that begins below $x$, ends below $x$ and hits $x$ in $[1, \Delta]$. Let $A = \{\sup_{s \in [1, \Delta]} W(s) \geq x\}$. We consider
\begin{align}
\PM{W(1) < x, A} &= \PM{W(1) < x, A, W(\Delta) > x} + \PM{W(1) < x, A, W(\Delta) < x}\\
&= 2\PM{W(1) < x, W(\Delta) > x, A}\\
&= 2\PM{W(1) < x, W(\Delta) > x},
\end{align}where the first line uses total probability, the second line uses the symmetry of a Wiener process, and the third line uses the intermediate value theorem. 

Now we also can consider
\begin{align}
\PM{W(1) < x, A} &= \PM{W(1) < x, A, W(\Delta) > x} + \PM{W(1) < x, A, W(\Delta) < x}\\
&=  \PM{W(1) < x, W(\Delta) > x} + \PM{W(1) < x, A, W(\Delta) < x}\\
&= \PM{W(1) < x, W(\Delta) > x} + \PM{\textsf{imposter}}.
\end{align}The first line again uses total probability. The second line uses the intermediate value theorem to remove the $A$. The third line uses the definition of an imposter path: a path that starts below $x$ and ends below $x$ but crosses $x$ between $[1, \Delta]$. Therefore
\begin{align}
\PM{\textsf{imposter}} &= \PM{W(1) < x, A} - \PM{W(1) < x, W(\Delta) > x}\\
&= 2 \PM{W(1) < x, W(\Delta) > x} - \PM{W(1) < x, W(\Delta) > x}\\
&= \PM{W(1) < x, W(\Delta) > x},
\end{align}where the first line uses the last set of algebra and the second line uses the equality of the first term developed in the first set of algebra. 

Bringing everything together, yields
\begin{align}
\PM{\sup_{s \in [1, \Delta]} W(s) \leq x} &= \PM{W(1) < x, W(\Delta) < x} - \PM{W(1) < x, W(\Delta) > x}\\
&= 2\PM{W(1) < x, W(\Delta) < x} - \PM{W(1) < x}\\
&= 2\Phi_2 \paren{x, \frac{x}{\sqrt{\Delta}} ; \frac{1}{\sqrt{\Delta}}} - \Phi(x)
\end{align}where the second line uses $\PM{A \cap B^c} = \PM{A} - \PM{A\cap B}$ and the third line uses the known correlation of a Wiener process. 

Therefore, the closed-form one-sided quantile for $q = 1$ is
\begin{equation}\label{eq:df-one-sided-maximum}
\PM{\sup_{s \in [1, \Delta]} W(s) \leq x} = 2 \Phi_2 \paren{x, \frac{x}{\sqrt{\Delta}} ; \frac{1}{\sqrt{\Delta}}} - \Phi(x),
\end{equation}and when $q = 0$ we have 
\[
\PM{\sup_{s \in [1, \Delta]} W(s)s\n \leq x} = 2 \Phi_2 \paren{x \sqrt{\Delta}, x; \frac{1}{\sqrt{\Delta}}} - \Phi \paren{x \sqrt{\Delta}}.
\]

\end{proof}

\subsection{Distribution Functions When \texorpdfstring{$q = 1/2$}{q = 1/2}}\label{sec:exact-quantiles-q12}
In this section we give the distribution functions of $\Psi, \Psi_+$ when $q = 1/2$. The following calculations amount to extending some recent work on finding the maximum of an Ornstein--Uhlenbeck (OU) process \citep{blanchet-scalliet_distribution_2025, blanchet-scalliet_pseudo-likelihood_2024}. Thus, we first begin with a lemma that allows us to use these OU results.

\begin{lemma}\label{lemma:maximum-equiv-lamperti}
Let $X(t)$ be an Ornstein--Uhlenbeck process with parameters $\lambda = 1, \sigma = \sqrt{2}$. If $X(0) \sim \cN(0,1)$, then 
\begin{equation}\label{eq:maximum-equiv-12-2}
\PM{\sup_{s \in [1, \Delta]} \abs{W(s)}s^{-1/2} \geq x} = \int_{-\infty}^x\PM{\sup_{t \in [0, \log \Delta /2]} \abs{X(t)} \geq x}[z]\phi(z)\dZ + \Phi(-x),
\end{equation}and 
\begin{equation}\label{eq:maximum-equiv-12-1}
\PM{\sup_{s \in [1, \Delta]} W(s)s^{-1/2} \geq x} = \int_{-\infty}^x\PM{\sup_{t \in [0, \log \Delta /2]} X(t) \geq x}[z] \phi(z)\dZ + \Phi(-x).
\end{equation}
\end{lemma}
Therefore, finding distribution functions for $q = 1/2$ amounts to finding them for the Ornstein--Uhlenbeck process and then integrating against a normal density.

\begin{proof}[Proof of \cref{lemma:maximum-equiv-lamperti}]
\begin{align}
\PM{\sup_{s \in [1, \Delta]} \frac{W(s)}{\sqrt s} \geq x} &= \int_{-\infty}^{\infty}\PM{\sup_{s \in [1, \Delta]} \frac{W(s)}{\sqrt s} \geq x\mid W(1) = z}\phi(z)\dZ\\
&= \int_{-\infty}^{x}\PM{\sup_{s \in [1, \Delta]} \frac{W(s)}{\sqrt s} \geq x\mid W(1) = z}\phi(z)\dZ + \int_{x}^{\infty}\phi(z)\dZ.
\end{align}The second term is $1-\Phi(x) = \Phi(-x)$ so we focus on the integrand in the first integral. More specifically we will analyze the probability statement within that integrand.

\begin{align}
\PM{\sup_{s \in [1, \Delta]} \frac{W(s)}{\sqrt s} \geq x\mid W(1) = z}&= \PM{\sup_{s \in [1, \Delta]} \frac{W(1 + s-1)}{\sqrt s} \geq x\mid W(1) = z}\\
&= \PM{\sup_{s \in [1, \Delta]} \frac{W(1) + W(s-1)}{\sqrt s} \geq x\mid W(1) = z}\\
&= \PM{\sup_{s \in [1, \Delta]} \frac{z}{\sqrt s} + \frac{W(s-1)}{\sqrt s} \geq x\mid W(1) = z}.
\end{align}Now, consider the change of variables letting $s = e^{2t}$. This gives $t = \frac{\log s}{2}$ changing the bounds to range to $t \in [0, \frac{\log \Delta}{2}]$. Putting it together gives 
\[
\PM{\sup_{t \in [0, \frac{\log \Delta}{2}]} \underbrace{\frac{z}{e^t} + \frac{W(e^{2t}-1)}{e^t}}_{X(t)} \geq x}[z].
\]Calling $X(t)$ our new time transformed process, we use a result originally from \citet{doob_brownian_1942} and also from \citet{lamperti_semi-stable_1962} that says $X(t)$ is an OU process starting at $z$ with parameters $\lambda = 1, \sigma = \sqrt{2}$. The two-sided proof with absolute values follows similarly.
\end{proof}

\subsubsection{One-Sided Distribution When \texorpdfstring{$q = 1/2$}{q = 1/2}}
There is an active probability literature dedicated to finding closed-form distribution functions for an OU process. \citet{alili_representations_2005} is a common source for the infinite series representation using parabolic cylinder functions. Recently, \citet{blanchet-scalliet_pseudo-likelihood_2024} gives a similar expression to the present one by way of integrating the density, though the application is much different than the present one. \citet{blanchet-scalliet_joint_2020} is another important reference for discussing the maximum of OU-processes.

\begin{proposition}[Distribution Function of $\Psi_+(x; \Delta, 1/2)$]
Let $D_v(\cdot)$ be the parabolic cylinder function with index $v$. With a fixed $x$, we can consider $\paren{v_{j,x}}_{j = 1}^{\infty}$ the ordered sequence of positive zeroes of the function $v \mapsto D_v(-x)$.  Write $\nu_j = v_{j,x}$ and $D'_{\nu_j} = \partial_\nu D_{\nu}(-x) \mid_{\nu = \nu_j}$ is the derivative of $D_{\nu_j}(\cdot)$ with respect to the index $\nu_j$. With this notation, we have
\[
1-\Psi_+(x; \Delta, 1/2) = \Phi(-x) - \phi(x)\sum_{j=1}^{\infty}\frac{D_{\nu_{j}-1}(-x)}{D'_{\nu_j}(-x)}\frac{1-\expp{-\log \Delta\nu_j/2}}{\nu_j} 
\]
\end{proposition}
\begin{proof}

We state Theorem 3.1 from \citet{alili_representations_2005} as a Lemma. Their theorem gives closed-form expressions for hitting times of OU processes.
\begin{lemma}\label{lemma:density-hitting-times-ou}
Let $X(t)$ be an OU process with $X(0) = z$ and parameters $\lambda = 1$. Define the hitting time of level $a$ to be 
\[
\tau_a = \inf\{t > 0: X(t) = a\}.
\]Let $p_{\tau_a, \sigma}(t)$ be the density of this random variable. Then 
\[
p_{\tau_a, \sigma}(t) = -\expp{\frac{z^2}{2\sigma^2}} \expp{\frac{-a^2}{2\sigma^2}}\sum_{j=1}^{\infty}\frac{D_{v_j, -a \sqrt{2}\sigma\n}\paren{-z \sqrt{2}\sigma \n}}{D'_{v_j,-a \sqrt{2}\sigma \n }\paren{-a \sqrt{2}\sigma \n}}\expp{-tv_{j, -a\sqrt{2}\sigma \n}}.
\]We can simplify this for our specific parameters:
\begin{equation}\label{eq:density-hitting-time-ou}
p_{\tau_a}(t) = -\expp{\frac{z^2}{4}}\expp{\frac{-a^2}{4}}\sum_{j=1}^{\infty}\frac{D_{v_j, -a}\paren{-z }}{D'_{v_j,-a  }\paren{-a }}\expp{-tv_{j, -a}}.
\end{equation}
\end{lemma}
We note that \citet{alili_representations_2005} consider the simple case when $\sigma = 1$ but it can be extended for general $\sigma$ by rescaling $X(t)$ by $\sigma$. This amounts to dividing $x,a$ by $\sigma$ in \citep[Theorem 3.1]{alili_representations_2005}.

Using the equivalence between stopping times and supremum we write
\begin{equation}\label{eq:equivalence-integral}
\PM{\tau_a < \beta} = \int_{0}^{\beta}p_{\tau_a}(t)\dT = \PM{\sup_{t \in [0, \beta]} X(t) \geq a}.
\end{equation}
If we did not have to integrate over $z$ the proof would be finished. However, we must integrate over and will presently show how to do so while maintaining the closed-form. For notational convenience write $ \log \Delta /2 = a'$ and $\nu_j = v_{j, x}$ the ordered set of zeroes of $v \mapsto D_v(-x)$.

\begin{align}
\int_{-\infty}^{x} \PM{\sup_{t \in [0, a']} X(t) \geq x}[z]\phi(z)\dZ
&= \int_{-\infty}^{x} \PM{\tau_{x} \leq a'}[z]\phi(z)\dZ\\
&= \int_{-\infty}^{x}\int_{0}^{a'}p_{\tau_x}(t)\dT \phi(z)\dZ,
\end{align}where the second line uses the equivalence discussed in \cref{eq:equivalence-integral}.
Now we input \cref{eq:density-hitting-time-ou}; the series converges uniformly on
$[\delta, \infty)$ for every $\delta > 0$ \citep[Thm.~3.1]{alili_representations_2005},
so we can integrate each term.

\begin{align}
\int_{-\infty}^{x}\int_{0}^{a'}p_{\tau_x}(t)\dT \phi(z)\dZ &= \int_{-\infty}^{x}\int_{0}^{a '} -\expp{\frac{z^2}{4}}\expp{\frac{-x^2}{4}}\sum_{j=1}^{\infty}\frac{D_{\nu_j}\paren{-z }}{D'_{\nu_j}\paren{-x}}\expp{-t\nu_j}\dT \phi(z)\dZ\\
&= \int_{-\infty}^{x} -\expp{\frac{z^2}{4}}\expp{\frac{-x^2}{4}}\sum_{j=1}^{\infty}\frac{D_{\nu_j}\paren{-z }}{D'_{\nu_j}\paren{-x}}\int_{0}^{a '}\expp{-t\nu_j}\dT \phi(z)\dZ.
\end{align}
Integrating with respect to $t$ inside the series, we obtain
\[
\int_{0}^{a'}\expp{-t\nu_j}\dT = \frac{1-\expp{-a'\nu_j}}{\nu_j}.
\]The series converges absolutely by \citep{alili_representations_2005} and therefore interchanging limits is warranted by Tonelli. With that (and more algebra)
\begin{align}
\int_{-\infty}^{x}\int_{0}^{a'}p_{\tau_x}(t)\dT \phi(z)\dZ &= \int_{-\infty}^{x} -\expp{\frac{z^2}{4}}\expp{\frac{-x^2}{4}}\sum_{j=1}^{\infty}\frac{D_{\nu_j}\paren{-z }}{D'_{\nu_j}\paren{-x}}\frac{1-\expp{-a' \nu_j}}{\nu_j} \phi(z)\dZ\\
&= -\expp{\frac{-x^2}{4}}\sum_{j=1}^{\infty} \frac{1-\expp{-a' \nu_j}}{D'_{\nu_j}(-x)\nu_j} \int_{-\infty}^{x} D_{\nu_j}(-z)\expp{\frac{z^2}{4}}\phi(z)\dZ
\end{align}

We now consider the last integral, which is a function of $z$. We have by \cref{lemma:parabolic-integral} that 
\begin{align}
\int_{-\infty}^{x} D_{\nu_j}(-z) \expp{\frac{z^2}{4}}\phi(z)\dZ
&= D_{\nu_j -1}(-x)\expp{\frac{x^2}{4}}\phi(x).
\end{align}

After pulling out the the constants with respect to $j$ and simplifying to $\phi(x)$, we have
\begin{align}
\int_{-\infty}^{x}\PM{\sup_{t \in [0, a']} X(t) \geq x}[z] \phi(z)\dZ = -\phi(x) \sum_{j=1}^{\infty} \frac{1-\expp{-a' \nu_j}}{D'_{\nu_j}(-x)\nu_j} D_{\nu_j -1}(-x).
\end{align}
By \cref{lemma:maximum-equiv-lamperti} and recalling $a' = \log \Delta /2$, we have that
\[
1-\Psi_+(x; \Delta, 1/2)= \Phi(-x) - \phi(x)\sum_{j=1}^{\infty}\frac{D_{\nu_j -1}(-x)}{D'_{\nu_j}(-x)}\frac{1-\expp{-\log \Delta\nu_j/2}}{\nu_j} 
\]
\end{proof}

\subsubsection{Two-Sided Distribution When \texorpdfstring{$q = 1/2$}{q = 1/2}}
Similar to the one-sided distribution we will use \cref{lemma:maximum-equiv-lamperti} and focus on integrating against the distribution of an OU maximum. Therefore, it is sufficient to show the following.

\begin{lemma}\label{lemma:ou-stationary-interval}
Let $X(t)$ be an OU process where $X(0) \sim \cN(0,1)$. Then 
\begin{equation}
\PM{\sup_{s \in [0, w]} \abs{X(s)} < x} = - 2\sqrt{2}\phi(x)\sum_{n \geq 0}\frac{\expp{-\nu_n w} B_{\nu_n-1}(x)}{\nu_n \partial_\mu A_{\nu_n}(x)}.
\end{equation}
\end{lemma}
\begin{proof}[Proof of \cref{lemma:ou-stationary-interval}]
We first recount some of the proof from \citet{blanchet-scalliet_distribution_2025}. The authors build the theory for a general $(c_1, c_2)$ and general types of OU processes, but our situation has $c_1 = -c_2$ and the ``canonical" OU process, which simplifies some expressions. 

First, we must go through the notation. The authors give the building blocks $G_{0, \mu}$ in Section 3. Our OU process has the parameters (in their paper's notation) $a = 0, b = 1, \sigma = \sqrt{2}$. This leaves us with
\[
G_{0, \mu}(x, y) = B_\mu(y) A_\mu(x) - B_\mu(x) A_\mu(y),
\]where 
\[
B_\mu(y) = \frac{y}{\sqrt{2}} M \paren{\frac{1-\mu}{2}; 1.5; \frac{y^2}{2}},
\]and 
\[
A_\mu(y) = M\paren{\frac{-\mu}{2}; 0.5; \frac{y^2}{2}},
\]where $M$ is the hypergeometric (Kummer's) function. First, recall a property of $A_\mu(y)$.

\begin{lemma}\label{lemma:deriv-amu}
Let $A_\mu(y) =  M\paren{\frac{-\mu}{2}; 0.5; \frac{y^2}{2}}$. 
\[
\frac{d}{dy}A_\mu(y) = -\mu y M \paren{\frac{2-\mu}{2}; 1.5; \frac{y^2}{2}} = -\sqrt{2} \mu B_{\mu-1}(y).
\]
\end{lemma}

One can check that 
\begin{align}
G_{0, \mu}(-x, x) &= B_\mu(x) A_{\mu}(-x) - B_\mu(-x) A_\mu(x)\\
&= 2 A_{\mu}(x)B_\mu(x),
\end{align}since $B$ is an odd function and $A$ is an even one.
Now the authors define 
\[
F_{0, \mu}(z) = G_{0, \mu}(-x, z) + G_{0, \mu}(z, x) = 2 B_\mu(x) A_\mu(z),
\]again using the fact that $A$ is even and $B$ is odd.

Let $\paren{\mu_n}_{n \geq 0}$ be the ordered sequence of zeros of $\mu \mapsto G_{0, \mu}(-x, x)$. 
The last line in \citep[Theorem 5.4]{blanchet-scalliet_distribution_2025} gives
\begin{align}
\PM{\tau_{-x, x} \geq w}[z] = -\sum_{n \geq 0}\expp{-\mu_n w}\frac{F^{(0)}_{0, \mu_n}(z)}{\mu_n \partial_\mu G_{0, \mu_n}(-x, x)},
\end{align}
 where
\begin{align}
\partial_\mu G_{0, \mu_n}(-x, x) &= \frac{\partial }{\partial \mu} G_{0, \mu}(-x, x) \mid_{\mu =\mu_n}.
\end{align}Calling $B'_\mu(x)$ and $A'_\mu(x)$ the partial derivatives we use the product rule to write
\[
\partial_\mu G_{0, \mu_n}(-x, x) = 2 \paren{A_{\mu_n}(x) B'_{\mu_n}(x) + A'_{\mu_n}(x)B_{\mu_n}(x)}.
\]Then we write the fraction
\begin{align}
\frac{F^{(0)}_{0, \mu_n}(z)}{\mu_n \partial_\mu G_{0, \mu_n}(-x, x)} &= \frac{B_{\mu_n}(x)A_{\mu_n}(z)}{\mu_n \paren{A_{\mu_n}(x) B'_{\mu_n}(x) + A'_{\mu_n}(x)B_{\mu_n}(x)}}.
\end{align}We can now simplify this equation greatly using the fact that $\mu_n$ is a zero of $\mu \mapsto G_{0, \mu}(-x, x)$. Using this fact, it must be that either (or both) $A_{\mu_n}(x) = 0$ or $B_{\mu_n}(x) = 0$. If $A_{\mu_n}(x) = 0$ we have the fraction is 
\[
\frac{B_{\mu_n}(x)A_{\mu_n}(z)}{\mu_n \paren{A_{\mu_n}(x) B'_{\mu_n}(x) + A'_{\mu_n}(x)B_{\mu_n}(x)}} = \frac{A_{\mu_n}(z)}{\mu_n A'_{\mu_n}(x)}.
\]If $B_{\mu_n}(x) = 0$, the whole term cancels out. Therefore consider the sequence of zeros $\paren{\nu_n}_{n \geq 0}$ which are the set of zeros of $\mu \mapsto A_{\mu}(x)$. Therefore for any $\nu_n$ we have 
\[
\frac{F^{(0)}_{0, \nu_n}(z)}{\nu_n \partial_\mu G_{0, \nu_n}(-x, x)} = \frac{A_{\nu_n}(z)}{\nu_n A'_{\nu_n}(x)},
\]giving that 
\[
\PM{\tau_{-x, x} \geq w}[z] = - \sum_{n \geq 0} \expp{-\nu_n w} \frac{A_{\nu_n}(z)}{\nu_n A'_{\nu_n}(x)}.
\]Now, we must integrate against a normal density.

\begin{align}
\PM{\tau_{-x, x} \geq w} &= \int_{-x}^{x}\PM{\tau_{-x, x} \geq w}[z] \phi(z)\dZ\\
&= \int_{-x}^{x} - \sum_{n \geq 0}\expp{-\nu_n w} \frac{A_{\nu_n}(z)}{\nu_n A'_{\nu_n}(x)}\phi(z)\dZ\\
&= -\sum_{n \geq 0}\expp{-\nu_n w}\frac{1}{\nu_n A'_{\nu_n}(x)} \int_{-x}^{x}A_{\nu_n}(z)\phi(z)\dZ.
\end{align}
We use the differential identities of $A$ and $\phi$ to help with the integral. First, we have that $\phi'(x) = -x\phi(x)$. 
We begin with
\begin{align}
\nu_n A_{\nu_n}(z) &= z \frac{d}{dz}A_{\nu_n}(z) - \frac{d^2}{dz^2}A_{\nu_n}(z)\\
\phi(z)\nu_nA_{\nu_n}(z) &= \phi(z) z \frac{d}{dz}A_{\nu_n}(z) - \phi(z) \frac{d^2}{dz^2}A_{\nu_n}(z)\\
&= - \frac{d}{dz}\phi(z)\frac{d}{dz}A_{\nu_n}(z) - \phi(z) \frac{d^2}{dz^2}A_{\nu_n}(z)\\
&= \frac{d}{dz}\paren{-\frac{d}{dz}A_{\nu_n}(z)\phi(z)},
\end{align}which implies that 
\[
\phi(z) A_{\nu_n}(z) = \frac{d}{dz}\paren{\frac{-\frac{d}{dz}A_{\nu_n}(z)\phi(z)}{\nu_n}} = \frac{d}{dz}\paren{\sqrt{2} B_{\nu_n-1}(z) \phi(z)},
\]where we use \cref{lemma:deriv-amu} in the final equality.
Thus the anti-derivative of $\phi(z) A_{\nu_n}(z)$ is $\sqrt{2}B_{\nu_n -1}(z)\phi(z)$. Evaluating at the boundary conditions gives 
\begin{align}
\int_{-x}^{x} A_{\nu_n}(z)\phi(z)
\dZ&= \sqrt{2}\phi(x)\paren{B_{\nu_n-1}(x) - B_{\nu_n-1}(-x)}\\
&= 2 \sqrt{2}\phi(x) B_{\nu_n-1}(x)
\end{align}where we use the symmetry of $\phi(x)$ and the fact that $B_\mu(x)$ is an odd function. Putting everything together gives 
\[
\PM{\tau_{-x, x} \geq w}[z \sim \cN(0,1)] = - 2\sqrt{2}\phi(x)\sum_{n \geq 0}\frac{\expp{-\nu_n w} B_{\nu_n-1}(x)}{\nu_n \partial_\mu A_{\nu_n}(x)}.
\]

\end{proof}

\subsection{Auxiliary Lemmas}
\begin{lemma}\label{lemma:wiener-escapes-box}
Let $W(s)$ be a Wiener process where $W(0) = x$. Then 
\begin{align}
\PM{\sup_{0 \leq s \leq t} \abs{W(s)} < a}[x] &= \sum_{k=-\infty}^{\infty} \Phi \paren{\frac{(4k + 1)a - x}{\sqrt{t}}}  
-\Phi \paren{\frac{(4k-1)a - x}{\sqrt{t}}} - \\
&\Phi \paren{\frac{(4k +3)a + x}{\sqrt{t}}} + \Phi \paren{\frac{(4k+1)a + x}{\sqrt{t}}}. 
\end{align}
\end{lemma}

\begin{proof}[Proof of \cref{lemma:wiener-escapes-box}]
We begin with \citep[Theorem 7.45]{morters_brownian_2010}, which states for a Wiener process where $W(0) = x$,
\begin{align}
\PM{\sup_{0 \leq s \leq t}W(s) \in (0,a)}[x] &= \sum_{k=-\infty}^{\infty} \Phi \paren{\frac{2ka + a - x}{\sqrt{t}}} - \Phi \paren{\frac{2ka - x}{\sqrt{t}}} \\
&\quad- \Phi \paren{\frac{2ka + a + x}{\sqrt{t}}} + \Phi \paren{\frac{2ka + x}{\sqrt{t}}}
\end{align}

Letting $a$ grow to $2a$ it follows
\begin{align}
\PM{\sup_{0 \leq s \leq t}W(s) \in (0,2a)}[x] &= \sum_{k=-\infty}^{\infty} \Phi \paren{\frac{4ka + 2a - x}{\sqrt{t}}} - \Phi \paren{\frac{4ka - x}{\sqrt{t}}} \\
&\quad- \Phi \paren{\frac{4ka + 2a + x}{\sqrt{t}}} + \Phi \paren{\frac{4ka  + x}{\sqrt{t}}},
\end{align}then shifting the starting point $x$ up to $x + a$ 
\begin{align}
\PM{\sup_{0 \leq s \leq t}W(s) \in (-a,a)}[x] &= \sum_{k=-\infty}^{\infty} \Phi \paren{\frac{(4k + 1)a - x}{\sqrt{t}}} \\
&\quad- \Phi \paren{\frac{(4k-1)a - x}{\sqrt{t}}} - \Phi \paren{\frac{(4k +3)a + x}{\sqrt{t}}} \\
&\quad+ \Phi \paren{\frac{(4k+1)a + x}{\sqrt{t}}}.
\end{align}
\end{proof}

\begin{lemma}\label{lemma:normal-cdf-integral}
Let $\Phi, \phi$ be the distribution function and density of a standard normal. Then 
\[
\int_{-a}^{a}\Phi\paren{c_1 x + c_2}\phi(x)\dX = \Phi_2 \paren{\frac{c_2}{\sqrt{1 + c_1^2}}, a; \frac{-c_1}{\sqrt{1 + c_1^2}}} - \Phi_2 \paren{\frac{c_2}{\sqrt{1 + c_1^2}}, -a; \frac{-c_1}{\sqrt{1 + c_1^2}}}
\]
\end{lemma}

\begin{proof}[Proof of \cref{lemma:normal-cdf-integral}]
When the bounds are infinite this identity is well-known. We prove it again here for finite bounds. Let $X \sim \cN(0,1)$ and $U \sim \cN(0,1)$ independently.
\begin{align}
\int_{-a}^{a}\Phi \paren{c_1x + c_2} \phi(x)\dX &= \EE{\Phi \paren{c_1 X + c_2} \indic{-a < X < a}}\\
&= \EE{\EE{\Phi \paren{c_1 X + c_2}\indic{-a < X < a} \mid X}}\\
&= \EE{\PM{U \leq c_1 X + c_2\mid X}\indic{-a < X < a}}\\
&= \PM{U \leq c_1 X + c_2, -a < X < a}\\
&= \PM{U - c_1X \leq c_2, -a < X < a}\\
&= \PM{U- c_1X \leq c_2, X < a} - \PM{U- c_1X \leq c_2, X < -a}.
\end{align}
Now, we can use properties of sums of normals to write this as a bivariate normal. Writing $\widetilde U = U - c_1X$, we have $Cov(\widetilde U, X) = -c_1$ and $\Var{U} = 1 + c_1^2$. Therefore, the correlation of the two random variables is $-c_1/\sqrt{1 + c_1^2}$ and we can write
\begin{align}
\PM{U - c_1X \leq c_2, X < a} - &\PM{U - c_1X \leq c_2, X < -a}= \PM{\widetilde U \leq c_2, X < a} - \PM{\widetilde U \leq c_2, X < -a}\\
&= \Phi_2 \paren{\frac{c_2}{\sqrt{1 + c_1^2}}, a; \frac{-c_1}{\sqrt{1 + c_1^2}}} - \Phi_2 \paren{\frac{c_2}{\sqrt{1 + c_1^2}}, -a; \frac{-c_1}{\sqrt{1 + c_1^2}}}
\end{align}
\end{proof}

\begin{lemma}\label{lemma:parabolic-integral}
Let $D_\nu(\cdot)$ be the parabolic cylinder function with index $\nu$. Then

\begin{align}\label{eq:parabolic-integral}
\int_{-\infty}^{x}D_{\nu}(-z)\expp{z^2/4}\phi(z)\dZ = D_{\nu-1}(-x)\expp{x^2/4}\phi(x)
\end{align}
\end{lemma}

\begin{proof}[Proof of \cref{lemma:parabolic-integral}]
The main work is to recall the identity
\begin{align}
    \frac{d}{dw}\expp{-w^2/4}D_{\nu}(w) = -\expp{-w^2/4}D_{\nu+1}(w).
\end{align}
Then integrand of the left hand side is equivalent to 
\[
\frac{1}{\sqrt{2\pi}}D_{\nu}(-z) \expp{-z^2/4}.
\]Making a substitution $u=-z$ and noticing that $\lim_{x \to\infty} \expp{-x^2/4}D_{\nu-1}(x) = 0$ gives the proof.
\end{proof}

\section{Additional Discussions}

\subsection{Multiple Disjoint Bounded Windows}
In this section, we briefly (and informally) show the ability to have multiple confidence horizons without union bounding, hence interpolating between group sequential methods and confidence horizons. 

To summarize confidence horizons, we ask for a given horizon $[m, \Delta m]$ how to construct a sequence of confidence intervals valid over this region. What if instead we have a set of horizons that are disjoint:
\[
[m_1, \overline m_1], [m_2, \overline m_2], \dots, [m_K, \overline m_K].
\]This situation may arise if ones ``prior" on the outcome is diffuse or bimodal. The analyst may think that the data have very high signal and few samples are needed to show this, or the analyst may be suspicious of high signal and will suspect a much larger set of samples is needed to estimate the parameter.

We can use the exact same confidence horizon tools developed. Scaling the boundary to be $m$ independent gives the horizon is 
\[
[1, \Delta_1], [\Delta_2, \Delta_3], \dots, [\Delta_{2K-2}, \Delta_{2K-1}].
\]
Thus we need to calculate the distribution of 
\[
\PM{\sup_{t \in\cT_K} W(t)t^{q-1} \leq x}
\]where $\cT_K =  [1, \Delta_1] \cup[\Delta_2, \Delta_3], \dots, [\Delta_{2K-2}, \Delta_{2K-1}]$. We can calculate this through simulation.

\subsection{Challenges of Numerical Integration for Large $K$}\label{sec:ld-approx-comparison}
\begin{figure}[!ht]
    \centering
\includegraphics[width=0.75\linewidth]{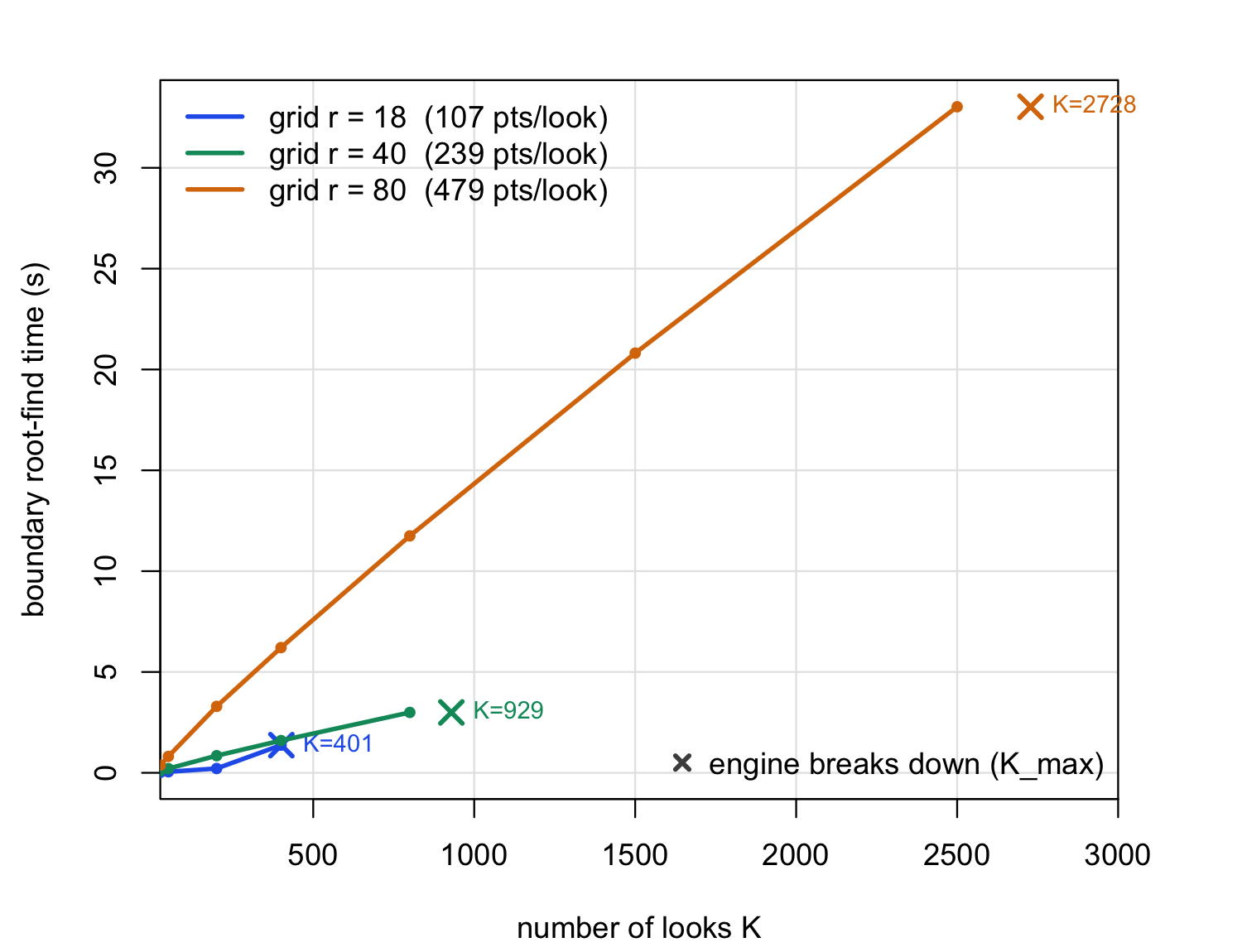}
    \caption{Illustration of numerical calculation of group sequential quantiles as $K$ increases \figcode[R]{figure15}}
    \label{fig:integration-runtime}
\end{figure}
\begin{figure}[!ht]
    \centering
    \includegraphics[width=0.95\linewidth]{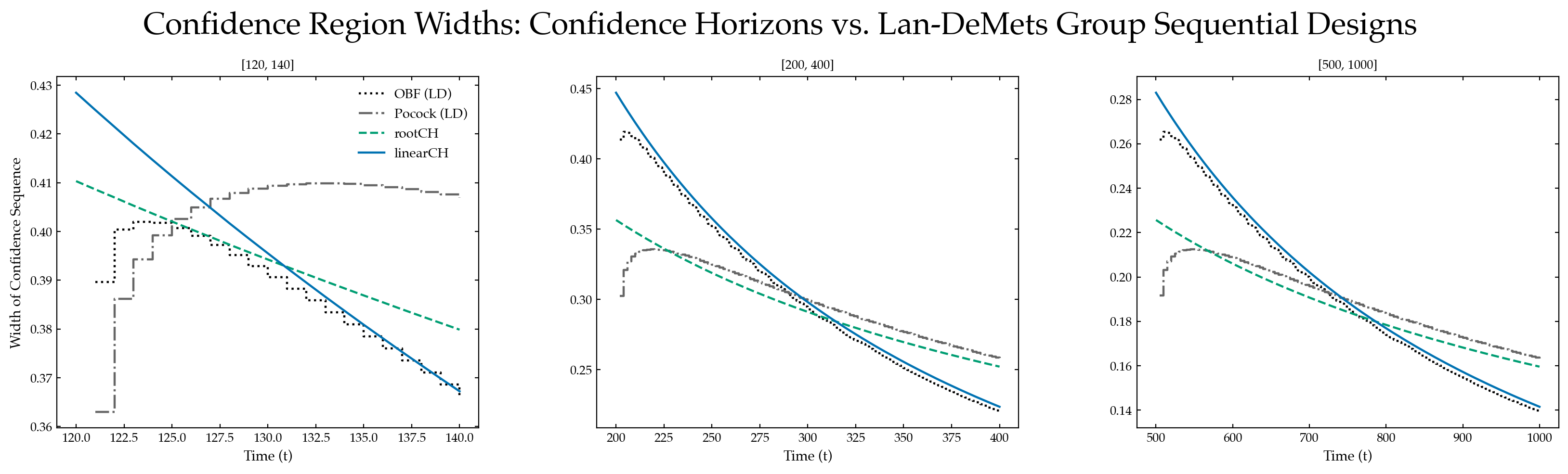}
    \caption{Comparison of confidence horizons with the $\alpha$-spending approximation from \citet{lan_discrete_1983}\figcode{figure16}}
    \label{fig:ld-ch}
\end{figure}

\subsection{Neyman Allocation Experiment Figures}
\begin{figure}[H]
    \centering
\includegraphics[width=0.8\linewidth]{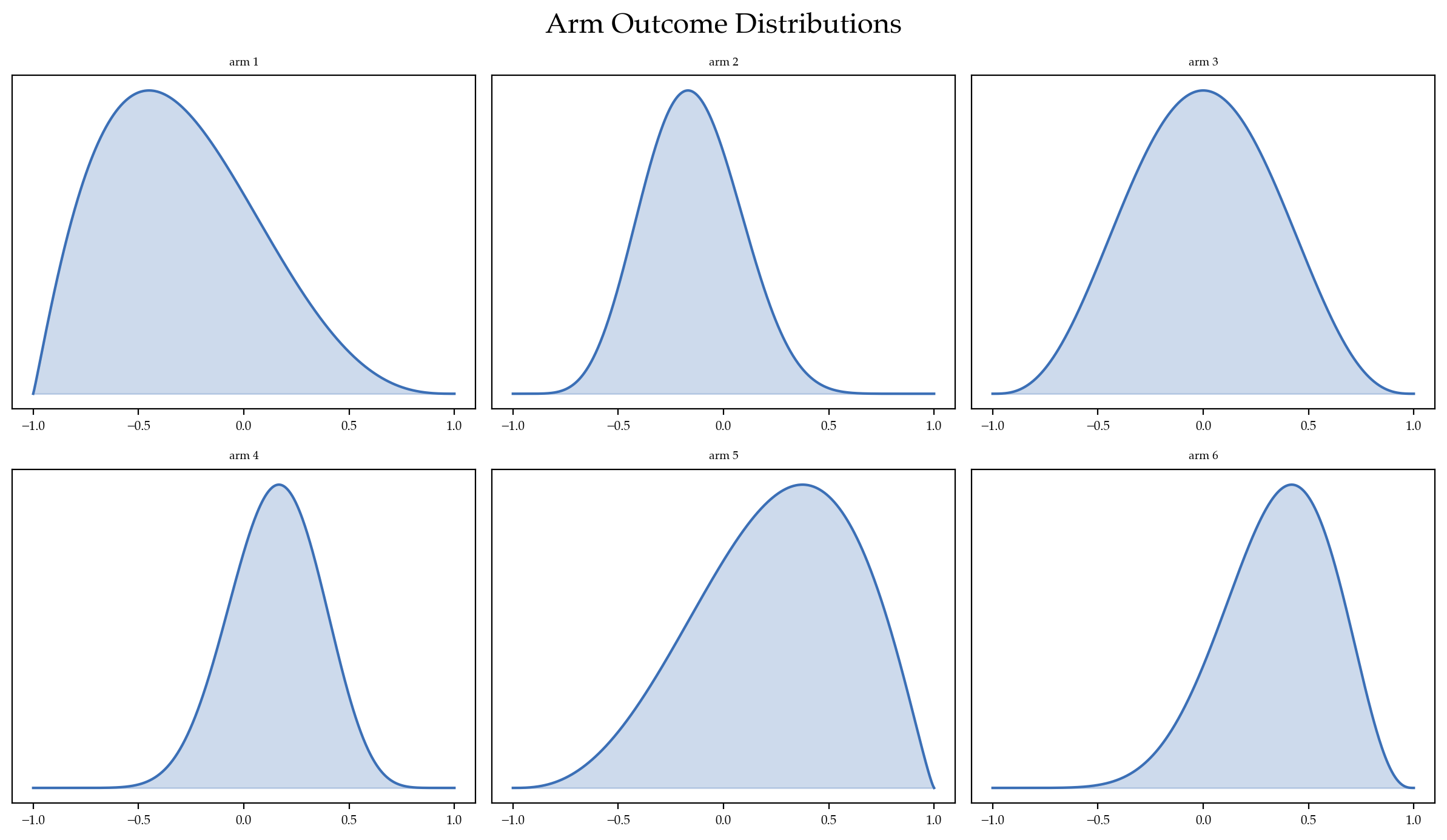}
    \caption{Distributions of the $J = 6$ arms. Notice that while arm 5 and arm 6 have similar means, arm 5 has a much larger spread compared to arm 6. Therefore, \citep[Algorithm 1]{kato2020efficient} should sample from arm 5 more than arm 6. In addition, the arms have different levels of skewness and are evidently not Gaussian. Thus, the confidence horizon coverage will only hold in the limit. \figcode{figure7}}
    \label{fig:outcome-y-arms}
\end{figure}

\begin{figure}[H]
        \centering
        \includegraphics[width=0.95\linewidth]{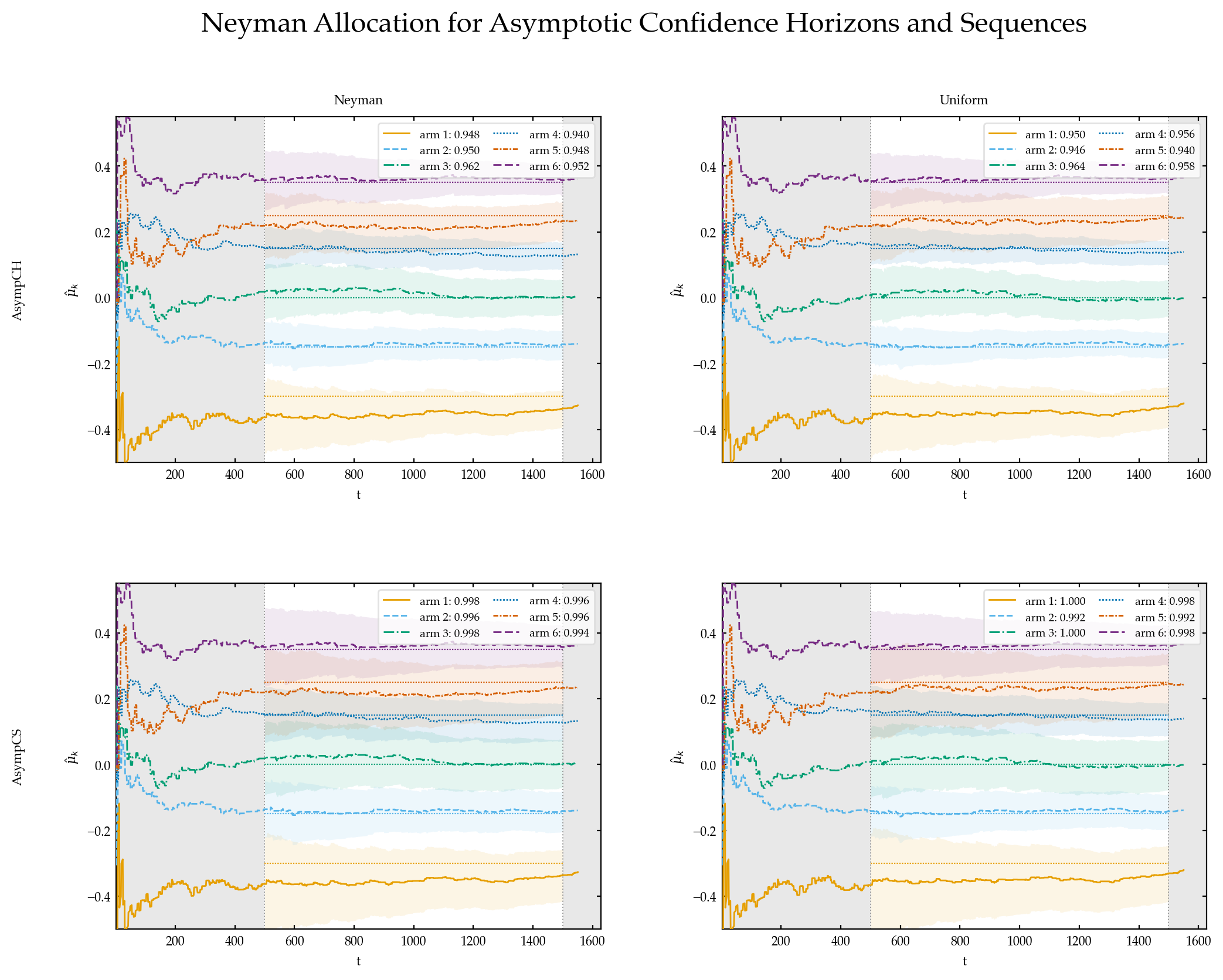}
        \caption{The full Neyman allocation experiment across AsympCS, AsympCH: \figcode{figure17}}
        \label{fig:full-neyman-experiment}
    \end{figure}

\subsection{\texorpdfstring{Power Analysis: The Impact of $\Delta, m$, and $q$}{Power Analysis: The Impact of Delta, m, and q}}
In this section, we answer the question on how to pick parameters to achieve a certain power guarantee. We will use the language of confidence regions throughout this section to maintain consistency with the other sections, although the present problem is fundamentally a testing problem.

For illustration purposes, we will consider a simple but useful example. It is straightforward to obtain the extensions to the broader work we discuss. Let $\infseqt{X} \simiid \P$ with mean $\mu_\P$ and variance $\sigma^2_\P$. An experimenter wishes to construct an AsympCH for the mean $\mu_\P$. The experimenter has, in general, three degrees of freedom. The experimenter can choose $m$, $\Delta$, and $q$. The experimenter must make sure $m$ is sufficiently large, but is otherwise unrestricted. In a typical power analysis the experimenter supposes that the true mean $\mu_\P$ is some fixed value $\mu_1$. Then, if the mean $\mu_\P$ is $\mu_1$, what is the probability the experimenter fails to reject the null? This probability is usually denoted as $\beta$ and power is $1-\beta$. 

We assume that the null distribution is the set of all $\P$ such that $\mu_\P = \mu_0$. Because we are in this asymptotic regime, we will consider \emph{local} alternatives of the form $\mu_\P = \mu_0 + \frac{\mu_1}{\sqrt{m}}$. We can precisely state for some $q \in \R$, 
\[
\beta = \lim_{m \to \infty}\PM{\forall t \in [m, \Delta m]: \mu_0 \in \bar C_{t}^{(m, \Delta)}}.
\]We give a formula for $\beta$ that depends on the experimenter's three degrees of freedom and the signal-to-noise ratio. Hence, we can invert this formula to give for a fixed $\beta$ ($\beta = 0.2$) what values of $(m, \Delta, q)$ are required to obtain that $\beta$.

\begin{proposition}[Power Analysis for an AsympCH]\label{prop:power-analysis-two-sided-asympch}
Let $\infseqt{X} \simiid \P$ with mean $\mu_\P = \mu_0 + \frac{\mu_1}{\sqrt{m}}$, variance $\sigma^2_\P$ and uniformly integrable $\kappa > 2$ moment. The asymptotic power of a two-sided $(1-\alpha)$-AsympCH is 
\begin{equation}
1-\beta = 1-\PM{\sup_{s \in [1, \Delta]} \abs{W(s)s^{q-1} + s^q\frac{\mu_1}{\sigma_\P}} \leq \Psi\n(1-\alpha; \Delta, q)}.
\end{equation}
\end{proposition}
The proof of \cref{prop:power-analysis-two-sided-asympch} is in \cref{proof:power-analysis}. While the above statement is true, the limiting argument abandons the dependence on $m$, an important knob for the experimenter. This is mostly a technicality for the limiting argument to make sense. Instead we give the more intuitive corollary that uses $m$.

\begin{corollary}
Let $\infseqt{X} \simiid \P$. Let $h = \mu_\P - \mu_0$. Then for a fixed finite $m$, we have
\begin{equation}\label{eq:beta-power-two-sided}
1-\beta \approx 1- \PM{\sup_{s \in [1, \Delta]} \abs{W(s)s^{q-1} + s^q m^{1/2}h/\sigma_\P} \leq \Psi \n(1-\alpha; \Delta, q)}.
\end{equation}
\end{corollary}

While it is straightforward to interpret the formula in \cref{eq:beta-power-two-sided} for parameters $m, h, \sigma_\P$ it is less clear the effect of $\Delta, q$ on power. The parameters $\Delta$ and  $q$ appear on either side of the equation so it is not obvious the effects on power. Hence we turn to simulations since the probability statement in \cref{eq:beta-power-two-sided} is unlikely to have a closed-form formula.

\begin{figure}[H]
    \centering
\includegraphics[width=0.95\linewidth]{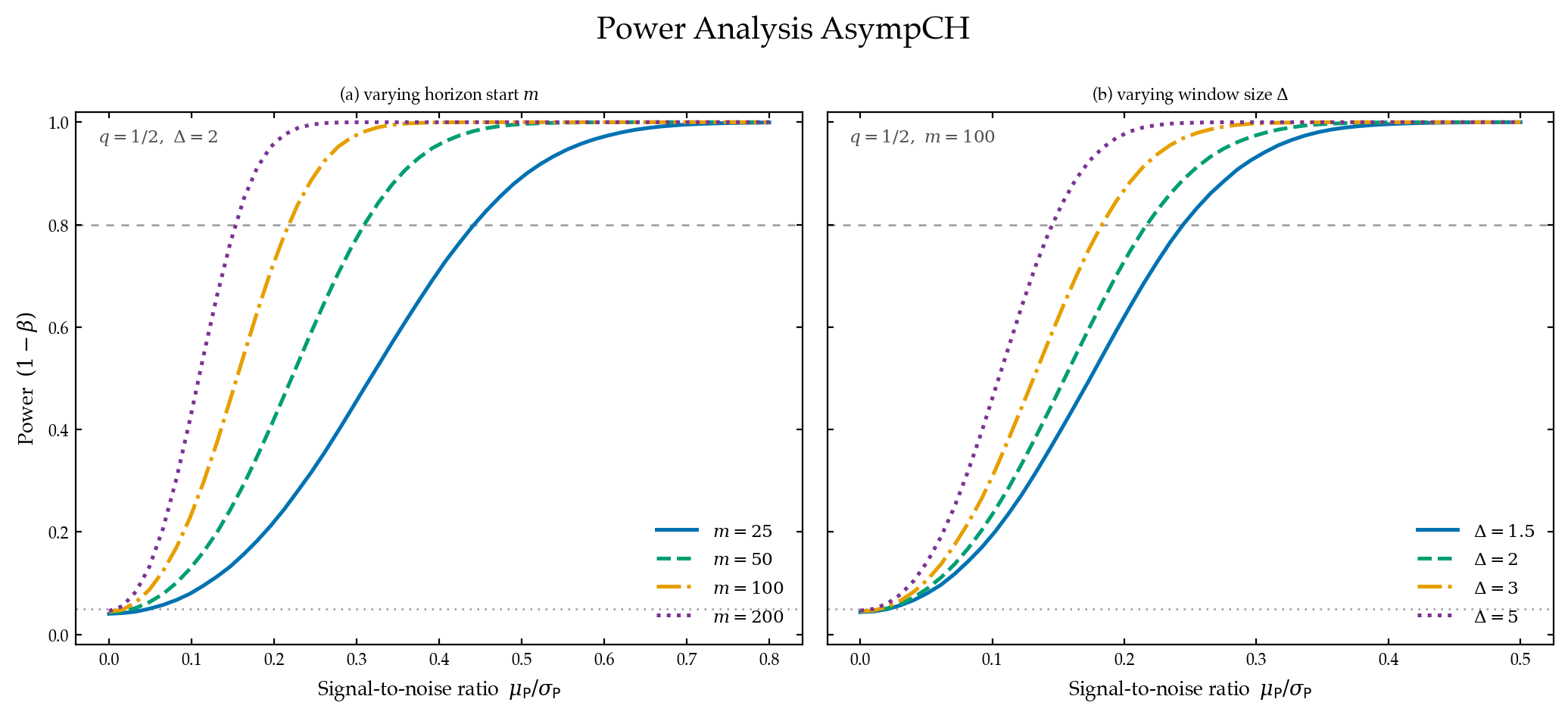}
    \caption{Power analysis over $m, \Delta$. As $m$ increases the AsympCH is more powerful. Taking $m\to \infty$ will yield a power 1 test. In addition, if $\Delta$ increases, the AsympCH is more powerful. \figcode{figure5}}
    \label{fig:power-analysis-m-delta}
\end{figure}

\paragraph{Power Stopping-time Tradeoff}

Power is one metric an experimenter can use to design an experiment. Power has been used to motivate traditional sequential tests from Wald \citep{wald_sequential_1945} and also group sequential methods. In contrast, the anytime-valid inference and e-values literatures find that expected stopping time is the  fundamental object of study \citep{ramdas_hypothesis_2025, waudby-smith_universal_2025}. Partially, the emphasis on expected stopping time is a result of the fact that anytime-valid inference methods should have ``power-one" in the limit \citep{agrawal_stopping_2025, darling_iterated_1967}. The results from \cref{fig:power-analysis-q} suggest that to maximize power, the analyst should take a large $q$. If the experimenter instead uses the expected rejection time as their metric, the results from \cref{fig:expected-stopping-time-q} suggest to take $q  = 1/2$. This raises the question about a tradeoff between power and expected rejection time. We find a similar tradeoff when comparing AsympCHs against group sequential methods.

\begin{figure}[H]
    \centering
    \begin{subfigure}[t]{0.48\linewidth}
        \centering
        \includegraphics[width=\linewidth]{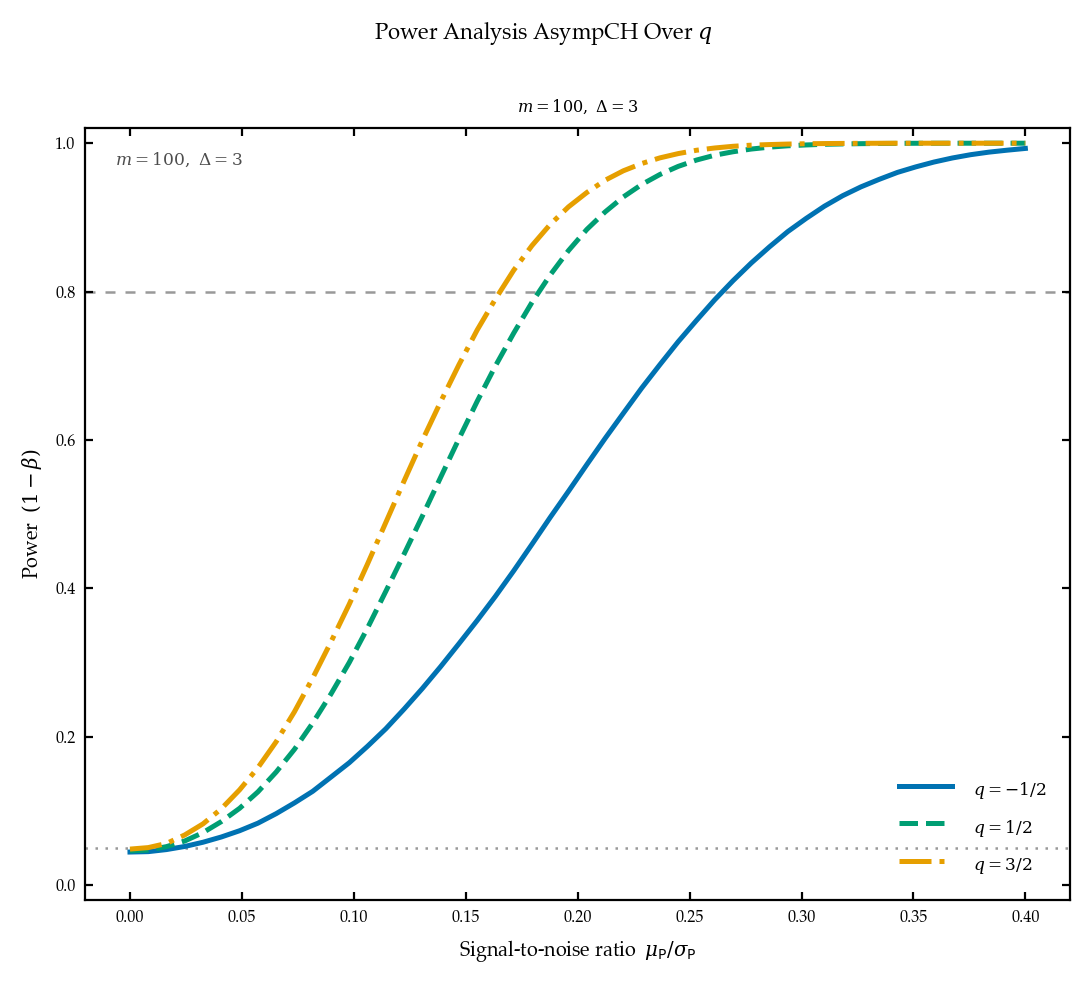}
        \caption{Power of the confidenece horizon while changing $q$.}
        \label{fig:power-analysis-q}
    \end{subfigure}
    \hfill
    \begin{subfigure}[t]{0.48\linewidth}
        \centering
        \includegraphics[width=\linewidth]{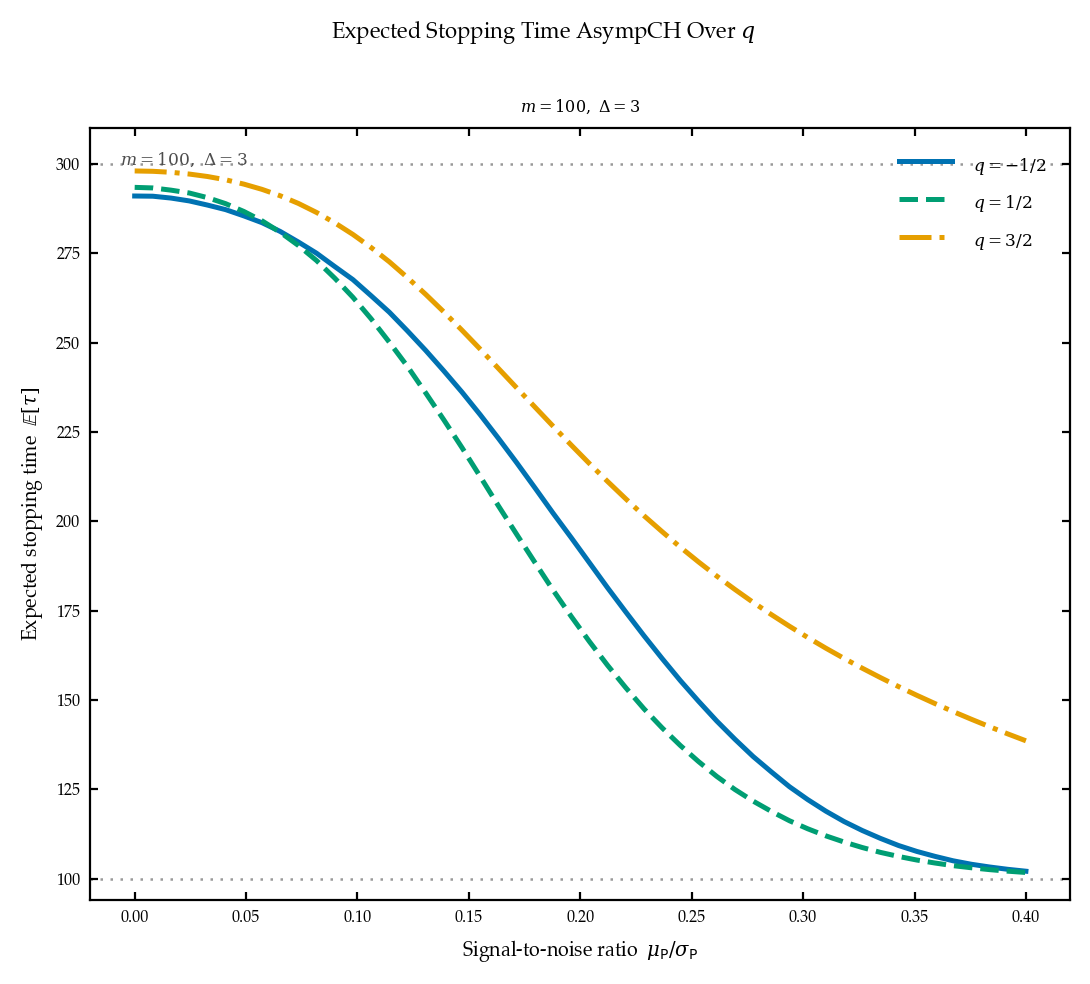}
        \caption{Expected stopping time while changing $q$.}
        \label{fig:expected-stopping-time-q}
    \end{subfigure}
    \caption{Power analysis for the AsympCH for different values of $q$. We see that as $q$ increases, the AsympCH becomes more powerful. This trend remains as $m$ and $\Delta$ vary.\figcode{figure6}
    }
\end{figure}

\begin{proof}[Proof of \cref{prop:power-analysis-two-sided-asympch}]\label{proof:power-analysis}

\begin{align}
\lim_{m \to \infty}\PM{\forall t \in [m, \Delta m]: \mu_0 \in \bar C_{t}^{(m, \Delta)}}
&= \lim_{m \to \infty}\PM{\forall t \in [m, \Delta m]: \mu_0 \in \hat \mu_t \pm \frac{\hat \sigma_t m^{q-1/2}\Psi\n (1-\alpha; \Delta, q)}{t^q}}\\
&= \lim_{m \to \infty}\PM{\forall t \in [m, \Delta m]: \abs{\hat \mu_t - \mu_0}\leq \frac{\hat \sigma_t m^{q-1/2}\Psi\n (1-\alpha; \Delta, q)}{t^q}}\\
&= \lim_{m \to \infty}\PM{\forall t \in [m, \Delta m]: \abs{\frac{S_t - t\mu_0}{t \hat\sigma_t}}\frac{t^q}{m^{q-1/2}}\leq \Psi\n (1-\alpha; \Delta, q)}
\end{align}
When we proved that the AsympCH controlled type-I error, we appealed to the fact that $S_t - t \mu_\P$ was a mean zero random variable. Presently, we do not assume that $\EE{S_t - t\mu_0}[\P] = 0$. Instead we have that $\EE{X_1}[\P] = \mu_\P = \mu_0 + \frac{\mu_1}{\sqrt{m}}$. Nevertheless, we can rearrange \cref{lemma:dist-uniform-kmt} for the result. Recall \cref{lemma:dist-uniform-kmt} states
\[
\abs{\frac{\tsum{X_i - t\mu_\P}}{t \sigma_\P} - \frac{W(t)}{t}} = \bar o_{\cP}\paren{\frac{1}{\sqrt{t}}},
\]which implies that 
\begin{align}
\bar o_{\cP}\paren{\frac{1}{\sqrt{t}}} &= 
\abs{\frac{\tsum{X_i}}{t \sigma_\P} - \frac{W(t)}{t} - \frac{\mu_\P}{\sigma_\P}}\\
&=\abs{\frac{\tsum{X_i} - t\mu_0}{t \sigma_\P} - \frac{W(t)}{t} - \frac{\mu_\P - \mu_0}{\sigma_\P}}\\
&= \abs{\frac{\tsum{X_i} - t\mu_0}{t \sigma_\P} - \frac{W(t)}{t} - \frac{\mu_1}{m^{1/2}\sigma_\P}}\label{eq:partial-sum-wiener-drift}
\end{align}

Returning to Step 2 in \cref{lem:general-lemma} we can say that 

\[
\PM{\sup_{t \in m\cT} \abs{\frac{S_t - t\Tilde \mu_{\P, t}}{\Tilde \sigma_{\P, t}t}}\frac{t^q}{m^{q-1/2}} \leq x} < \PM{\sup_{t \in m\cT} \abs{\frac{W(t)}{t} + \frac{\mu_1}{m^{1/2}\sigma_\P}}\frac{t^q}{m^{q-1/2}} \leq x} + \eps.
\]It is straightforward after following the subsequent steps in the proof that 
\[
\lim_{m \to \infty}\PM{\forall t \in [m, \Delta m]: \mu_0 \in \bar C_{t}^{(m, \Delta)}} = \PM{\sup_{s \in [1, \Delta]} \abs{W(s)s^{q-1}  + \frac{s^q\mu_1}{\sigma_\P}} \leq \Psi \n(1-\alpha; \Delta, q)}.
\]

\end{proof}

\end{document}